%% file: PRA_TDDFT_Finite_dimension_final.tex
\documentclass[pra,aps,superscriptaddress,twocolumn,10pt,floatfix]{revtex4-2}
\usepackage[utf8]{inputenc}
\usepackage{mathrsfs}
\usepackage{dsfont}
\usepackage{hyperref}
\usepackage{amsmath}
\usepackage{amssymb}
\usepackage{amsthm}
\usepackage{amsfonts}
\usepackage{amstext}
\usepackage{amsopn}
\usepackage{amsxtra}
\usepackage{mathrsfs}
\usepackage{dsfont}
\usepackage{color}
\usepackage{graphicx}
\usepackage{xcolor}
\usepackage{cleveref}

\makeatletter
\AddToHook{cmd/appendix/before}{\def\cref@section@alias{appendix}}
\makeatother

\usepackage{tikz}
\usepackage{bbm}

\usetikzlibrary{arrows.meta}

\newtheorem{theorem}{Theorem}

\newtheorem{definition}[theorem]{Definition}

\newcommand{\R}{\mathbb{R}}

\newcommand{\C}{\mathbb{C}}

\newcommand{\bH}{H}
\newcommand{\dps}{\displaystyle}
\newcommand{\ii}{\infty}
\newcommand\1{{\ensuremath {\mathds 1} }}
\renewcommand\phi{\varphi}
\newcommand{\gH}{\mathfrak{H}}

\newcommand{\cS}{\mathcal{S}}
\newcommand{\cO}{\mathcal{O}}

\newcommand{\cA}{\mathcal{A}}
\newcommand{\cT}{\mathcal{T}}

\newcommand{\cV}{\mathcal{V}}

\newcommand{\cC}{\mathcal{C}}
\newcommand{\cU}{\mathcal{U}}

\newcommand{\cM}{\mathcal{M}}

\newcommand{\cE}{\mathcal{E}}

\newcommand{\cN}{\mathcal{N}}
\newcommand{\cJ}{\mathcal{J}}

\newcommand\pscal[1]{{\ensuremath{\left\langle #1 \right\rangle}}}
\newcommand{\norm}[1]{ \left\| #1 \right\|}
\renewcommand{\geq}{\geqslant}
\renewcommand{\leq}{\leqslant}

\renewcommand{\tilde}{\widetilde}
\newcommand{\eps}{\varepsilon}

\newcommand{\nn}{\nonumber}
\newcommand{\rd}{\mathrm{d}}

\newcommand{\dt}{\rd t}

\DeclareMathOperator{\tr}{Tr}
\DeclareMathOperator{\tspan}{span}
\DeclareMathOperator{\sign}{sign}
\newcommand{\bra}[1]{\ensuremath{\langle #1 \vert}}
\newcommand{\ket}[1]{\ensuremath{\vert #1  \rangle}}

\begin{document}

\title{Geometric theory of constrained Schr{\"o}dinger dynamics with application to time-dependent density-functional theory on a finite lattice}

\author{\'Eric~Canc\`es}
\email{eric.cances@enpc.fr}
\affiliation{CERMICS, \'Ecole des Ponts -- Institut Polytechnique de Paris and Inria, 6-8 avenue Blaise Pascal, Cité Descartes, 77455 Marne-la-Vallée, France}

\author{Th\'eo~Duez}
\email{theo.duez@enpc.fr}
\affiliation{CERMICS, \'Ecole des Ponts -- Institut Polytechnique de Paris and Fédération CNRS Bézout, 6-8 avenue Blaise Pascal, Cité Descartes, 77455 Marne-la-Vallée, France}

\author{Jari~van~Gog}
\email{jari.van\_gog@etu.sorbonne-universite.fr}
\affiliation{Laboratoire de Chimie Théorique, Sorbonne Université and CNRS, 4 place Jussieu, 75005 Paris, France}

\author{Asbj{\o}rn~B{\ae}kgaard~Lauritsen}
\email{lauritsen@ceremade.dauphine.fr}
\affiliation{CEREMADE, CNRS, Université Paris-Dauphine, PSL Research University, Place de Lattre de Tassigny, 75016 Paris, France}

\author{Mathieu~Lewin}
\email{Mathieu.Lewin@math.cnrs.fr}
\affiliation{CEREMADE, CNRS, Université Paris-Dauphine, PSL Research University, Place de Lattre de Tassigny, 75016 Paris, France}

\author{Julien~Toulouse}
\email{toulouse@lct.jussieu.fr}
\affiliation{Laboratoire de Chimie Théorique, Sorbonne Université and CNRS, 4 place Jussieu, 75005 Paris, France}

\begin{abstract}
Time-dependent density-functional theory (TDDFT) is a central tool for studying the dynamical electronic structure of molecules and solids, yet aspects of its mathematical foundations remain insufficiently understood. In this work, we revisit the foundations of TDDFT within a finite-dimensional setting by developing a general geometric framework for Schrödinger dynamics subject to prescribed expectation values of selected observables. We show that multiple natural definitions of such constrained dynamics arise from the underlying geometry of the state manifold. The conventional TDDFT formulation emerges from demanding stationarity of the action functional, while an alternative, purely geometric construction leads to a distinct form of constrained Schrödinger evolution that has not been previously explored. This alternative dynamics may provide a more mathematically robust route to TDDFT and may suggest new strategies for constructing nonadiabatic approximations. Applying the theory to interacting fermions on finite lattices, we derive novel Kohn--Sham schemes in which the density constraint is enforced via an imaginary potential or, equivalently, a nonlocal Hermitian operator. Numerical illustrations for the Hubbard dimer demonstrate the behavior of these new approaches.
\end{abstract}

\maketitle

\tableofcontents

\section{Introduction}

Density-functional theory (DFT)~\cite{HohKoh-64,KohSha-PR-65} and its time-dependent extension (TDDFT)~\cite{RunGro-84,vanLeeuwen-99} are among the most powerful and widely used approaches for investigating the static and dynamic electronic structure of molecular and solid-state systems. While the mathematical foundations of DFT are well established~\cite{Lieb-83b}, the theoretical framework of TDDFT remains less rigorously understood. In particular, the original proofs of the Runge–Gross and van Leeuwen theorems~\cite{RunGro-84,vanLeeuwen-99} rely on the assumption of time analyticity of the external potential and of the wavefunction, which fails to hold for singular potentials such as Coulomb potentials~\cite{MaiTodWooBur-10,YanMaiBur-12,FouLamLewSor-16}. The mathematical framework of TDDFT for continuous and lattice systems has seen substantial development~\cite{FarTok-12,RugPenBau-JPA-09,RugLee-11,PenRug-11,RugGiePenLee-12,RugPenLee-15}, yet certain aspects of its foundations could benefit from further clarification. Beyond issues of mathematical rigor, advancing our theoretical understanding of TDDFT could guide the construction of more accurate approximations, potentially overcoming the limitations of the adiabatic approximation, which arguably remains one of the most pressing challenges in TDDFT today (see, e.g., Ref.~\cite{LacMai-NPJCM-23}).

In this work, we revisit the foundations of TDDFT within a finite-dimensional framework. To this end, we first develop a general and abstract formulation of Schrödinger's dynamics subject to prescribed expectation values of selected observables (such as the spatial one-particle density). By emphasizing the geometric structure underlying these constrained dynamics, we demonstrate that there exist several natural ways to define them. The choice that renders the action stationary corresponds to the conventional formulation adopted in standard TDDFT. However, we identify an alternative and equally natural definition based purely on the geometry of the manifold of constrained states. This leads to a distinct form of constrained Schrödinger dynamics, which, to the best of our knowledge, has not been previously explored in the context of TDDFT. This alternative formulation may provide a mathematically more robust approach for TDDFT and may lead to new forms of approximations.

We then apply the general theory to the case of TDDFT for interacting fermions on a finite lattice. In this setting, the new form of constrained Schrödinger's dynamics naturally gives rise to novel types of Kohn–-Sham schemes, in which the prescribed density is enforced through a nonlinear imaginary potential or, equivalently, a nonlocal Hermitian operator. We illustrate these approaches on the Hubbard dimer.

In the companion paper~\cite{MAQUI_TDDFT-26_ppt} we use the geometric framework described in this article to revisit the foundations of TDDFT for continuous systems.

The paper is organized as follows. In Section~\ref{sec:theory}, we establish the general geometric framework for constrained Schrödinger dynamics. Within this framework, Section~\ref{sec:TDVP} introduces the standard formulation based on making the action stationary, called the variational principle. Section~\ref{sec:McLachlan} then presents an alternative formulation, the geometric principle, derived purely from geometric considerations. In Section~\ref{sec:oblique}, we introduce an interpolating approach that connects the variational and geometric principles. In Section~\ref{sec:qubit}, we provide an illustration of the different principles on one qubit. Section~\ref{sec:Hubbard} applies the theory to TDDFT on a finite lattice, while Section~\ref{sec:Hubbard_dimer} provides numerical illustrations for the Hubbard dimer. Finally, the appendices 
contain additional derivations and proofs (Appendices~\ref{app:qubit_oblique}, \ref{app:proof_matrix_K}, and \ref{sec:independence}), an extension of the main results to mixed states (\Cref{app:mixed}),
and an equivalent algebraic derivation of our results (\Cref{sec.algebra.corr.term}).


\section{General geometric theory of constrained Schrödinger dynamics}\label{sec:theory}

We discuss how to modify the time-dependent Schrödinger equation in order to force its solution to satisfy some prescribed constraints. 

\subsection{Description of the problem}

We work in finite dimension $d$, hence the state of our quantum system is represented by a $d$-component state $\psi\in\C^d$. Some of our arguments apply to infinite dimensions as well, but not all of them do. We will particularly emphasize the geometric interpretation of the constrained dynamics. In order to clarify the mathematical structure of our problem, it is useful to consider a general class of constraints taking the form of fixed expectation values
\begin{equation}
\pscal{\psi(t),\cO_m\psi(t)}=o_m(t) ,
\qquad m=1,...,M,
\label{eq:constraints_Oj}
\end{equation}
for some general observables $\cO_1,...,\cO_M$ and all times $t$. Without loss of generality, we can assume that the $\cO_m$'s are linearly independent.
A similar abstract setting was considered before in~\cite{SchGunNoa-95,XuMaoGaoLiu-22,Song-23,PenLee-PRA-25}. We assume that the $\cO_m$'s are all time-independent $d \times d$ Hermitian matrices (see \Cref{sec.algebra.corr.term} for an extension to time-dependent observables). For TDDFT, the observables will be the density operators $\cO_m =\sum_{\sigma\in\{\uparrow,\downarrow\}}a_{m\sigma}^\dagger a_{m\sigma}$ where $a_{m\sigma}^\dagger$ and $a_{m\sigma}$ are the creation and annihilation operators of a particle of spin $\sigma$ at position $m$ over a finite lattice. We will come back to this specific case later in Section~\ref{sec:Hubbard}. Current density-functional theory and density-matrix functional theory also fit into this framework (see Section~\ref{sec.current.DFT}). 

The reference \textbf{Schrödinger equation} is
\begin{equation}
\left\{
\begin{aligned}
i\partial_t\psi^\text{S}(t)&=H(t)\psi^\text{S}(t)\\
\psi^\text{S}(0)&=\psi_0
\end{aligned}
\right.
\label{eq:Schrodinger_unconstrained}
\end{equation}
for some time-dependent Hamiltonian $H(t)$ (a $d\times d$ Hermitian matrix) and some normalized quantum state $\psi_0\in\C^d$. Even if the constraints are satisfied at time $t=0$, the Schrödinger evolution~\eqref{eq:Schrodinger_unconstrained} will in general not preserve them at later times. In particular, if the $o_m$'s are time-independent, this is only the case for all initial states $\psi_0$ when the matrices $\cO_m$'s commute with $H(t)$, i.e. $\cO_m H(t)=H(t)\cO_m$ for all times $t$. Since we do not make such an assumption, the Schrödinger equation~\eqref{eq:Schrodinger_unconstrained} must be modified to enforce the constraints~\eqref{eq:constraints_Oj} at all times. There are several possibilities to do so, that we will discuss at length below. In short, we will introduce several \textbf{modified Schrödinger equations}, that all take the general form
\begin{equation}
\left\{
\begin{aligned}
i\partial_t\psi(t)&=\left( H(t)+F\big(t,\psi(t)\big) \right) \psi(t)\\
\psi(0)&=\psi_0,    
\end{aligned}
\right.
\label{eq:Schrodinger_perturbed}
\end{equation}
with $F\big(t,\psi(t)\big)$ a correction term used to enforce the constraints~\eqref{eq:constraints_Oj} at all times. In our applications the correction term will depend on the time $t$ (through the Hamiltonian $H(t)$ and the constraints $o_m(t)$), together with the instantaneous state $\psi(t)$. In general it could be a more complicated function depending on the whole trajectory $\{\psi(s)\}_{0\leq s\leq t}$ until the current time $t$. In order to ensure that our model is causal, we will not allow it to depend on later times, however. The modified Schrödinger equation~\eqref{eq:Schrodinger_perturbed} ends up being a highly nonlinear Schrödinger equation whose solution $\psi(t)$ will usually be very different from the solution $\psi^\text{S}(t)$ of the reference Schrödinger equation~\eqref{eq:Schrodinger_unconstrained}. In this paper we will particularly focus on the geometric interpretation of the different choices for the correction term $F\big(t,\psi(t)\big)$. We note that, in a TDDFT spirit, for a given Hamiltonian $t \mapsto H(t)$, the correction term $F$ could alternatively be considered as a functional of the trajectory of the expectation values $\{o(s)\}_{0\leq s\leq t}$ and of the initial state $\psi_0$.

\subsection{Geometric structure of the set of constrained states for time-independent constraints}\label{sec:geomectric-structure-constraints}

We now discuss the interpretation of the constraints in the framework of differential geometry. For simplicity, we consider first time-independent constraint values $o_m$ and postpone the time-dependent case to Section~\ref{sec:timedepcons}. Thus, at each time $t$, the wavefunction $\psi(t)$ belongs to the time-independent set
\begin{equation}
\cC =\big\{\psi\in \C^d\text{ satisfying~\eqref{eq:constraints_Oj} and $\|\psi\|=1$}\big\}.
\label{eq:manifold_constraints}
\end{equation}
where $\|\psi\|:=\sqrt{\psi^\dagger\psi}$ denotes the usual norm of $\psi$. Here $\psi^\dagger$ denotes the complex conjugation and transposition of the column vector $\psi$. 

We have to be careful to work with the \emph{real} structure of our state space $\C^d$, and not the complex structure. This is because it is not the same to be differentiable in the real sense or in the complex sense. 
(Differentiable in this context refers to a function of $\psi$ being differentiable with respect to the components $\psi_j$ of $\psi$.)
By real structure we mean that we see any vector $\psi$ in $\C^d$ as a vector in $\R^{2d}$, each component $\psi_j$ being split into its real part $\Re(\psi_j)$ and imaginary part $\Im(\psi_j)$. 
The usual real scalar product in $\R^{2d}$ of two vectors $\psi,\psi'\in\C^d$ is
$$\sum_{j=1}^d\Re(\psi_j)\Re(\psi_j')+\sum_{j=1}^d\Im(\psi_j)\Im(\psi_j')=\Re\pscal{\psi,\psi'}$$
where $\pscal{\psi,\psi'}=\psi^\dagger\psi'=\sum_{j=1}^d\overline{\psi_j}\psi_j'\in\C$ denotes the usual complex scalar product of $\C^d$. The real scalar product $\Re\pscal{\psi,\psi'}$ will play a central role in our study.

The set $\cC$ in \eqref{eq:manifold_constraints} can be decomposed into a regular part $\cM$ and a singular part $\cS$ as $\cC = \cM \cup \cS$. 
The regular part $\cM$ is the set where the constraints are non-redundant (or qualified in the language of constrained optimization) in the sense that the $\cO_m \psi$'s are $\R$-linearly independent. 
It follows that $\cM$ is a smooth  (infinitely differentiable) submanifold of $\R^{2d}$ of dimension $2d-M$,
see for instance \cite[Section 7.7]{Boumal-23} or \cite[Theorem 5.12]{Lee-12}.
The singular part $\cS$ is interpreted as a boundary of $\cM$, where (some of) the constraints become redundant. 
We will restrict ourselves to states evolving only within the regular part $\cM$. 

The tangent space $\cT_\psi$ at a state $\psi\in\cM$ is 
the set of directions $h\in\C^d$ such that the constraints~\eqref{eq:constraints_Oj} are all satisfied by the modified state $\psi+\eps h$ up to an error of size $O(\eps^2)$.
Recalling the expansion
\begin{multline*}
\pscal{\psi+\eps h,\cO_m(\psi+\eps h)}\\=\pscal{\psi,\cO_m\psi}+2\eps\Re\pscal{h, \cO_m\psi}+\eps^2\pscal{h,\cO_m h}
\end{multline*}
due to the Hermiticity of the observable $\cO_m$, we obtain that the tangent space is
\begin{equation}
\cT_\psi:=\big\{h \in\C^d\ |\ \Re\pscal{h,\cO_m\psi}=0,\ \forall m=1,...,M\big\}.
\label{eq:T_psi_M}
\end{equation}
Notice that the real scalar product mentioned before naturally arises. The constraint $\Re\pscal{h,\cO_m\psi}=0$ for all $m$ is interpreted by saying that $h$ must be orthogonal to all the vectors $\cO_m\psi$ for the \emph{real} structure, that is, when seen as vectors in $\R^{2d}$. The \emph{normal space} at a state $\psi\in\cM$ is by definition the orthogonal complement to $\cT_\psi$ for the real structure and thus equals
\begin{align}
\cN_\psi&:=\tspan_\R(\cO_1\psi,...,\cO_M\psi)\nn\\
&=\left\{\sum_{m=1}^M v_m\cO_m\psi,\quad v_1,...,v_M\in\R\right\}.
\label{eq:N_psi_M}
\end{align}
Note that, since we removed the singular states $\psi$ for which the $\cO_m\psi$'s are linearly dependent, we have
\begin{equation}
\dim_\R(\cN_{\psi})=M,\qquad\dim_\R(\cT_\psi)=2d-M.
\label{eq:full_rank}
\end{equation}
An equivalent way of expressing that the vectors $\cO_m\psi$'s are $\R$-linearly independent is that 
the $M\times M$ symmetric overlap (\textit{a.k.a.}~Gram) matrix
\begin{equation}
(S^\psi)_{mn}=\Re\pscal{\cO_m\psi,\cO_n\psi}=\pscal{\psi,\frac{\{\cO_n,\cO_m\}}{2}\psi}, 
\label{eq:overlap}
\end{equation}
with $\{A,B\} := A B + BA$ denoting the anticommutator, 
is invertible, i.e. $\det(S^\psi)\neq0$.
The regular set is thus given by 
\begin{equation}
    \cM = \big\{\psi \in \cC \,|\, \det( S^\psi) \ne 0\big\}.
\end{equation}

We assume that our initial state belongs to the regular part $\psi_0 \in \cM$. That is, we assume that $S^{\psi_0}$ is invertible. 
In addition, we will only work with trajectories $t\mapsto \psi(t)$ defined over an interval of times $t\in[0,T)$ for which $S^{\psi(t)}$ remains invertible.
Note that a matrix very similar to the matrix $S^\psi$ already appeared in a work on the abstract extension of the Hohenberg--Kohn theorem~\cite{XuMaoGaoLiu-22}. For the special case of commuting observables $\cO_m$'s, the matrix $S^\psi$ also appears in a recent work~\cite{PenLee-PRA-25} in the context of constrained search in imaginary time for ground-state DFT.
The invertibility of $S^\psi$ is closely linked to the {\bfseries unique $v$-representability problem} \cite{PenLeu-21} in the case of TDDFT, see \Cref{sec:v-rep.invert-S-K}. We discuss further the invertiblity of $S^\psi$ in \Cref{sec:independence}. Generically, one should expect $S^\psi$ to be invertible.

Next we discuss under which condition a continuously-differentiable trajectory $t\mapsto\psi(t)$ with $\psi_0\in\cM$ stays on the manifold~$\cM$. Differentiating in time, we find for any $\psi(t)$
$$\frac{\rd}{\dt}\pscal{\psi(t),\cO_m\psi(t)}=2\Re\pscal{\partial_t\psi(t),\cO_m\psi(t)}.$$
The left-hand-side vanishes if and only if the right-hand-side does. 
Thus, from the formula for the normal space in \eqref{eq:N_psi_M} above, we conclude that
\emph{a trajectory $t\mapsto\psi(t)$ starting on $\cM$ stays on $\cM$ if and only if $\partial_t\psi(t)$ belongs to the tangent space $\cT_{\psi(t)}$ at all times}:
\begin{equation}
\partial_t\psi(t)\in \cT_{\psi(t)}.
\label{eq:stays_on_M}
\end{equation}

\subsection{Time-dependent constraints}
\label{sec:timedepcons}
Let us next explain the necessary modifications in the case of time-dependent constraint values $o_m(t)$, that is, a set $\cM(t)$ of constrained states that is moving with time. 
Here we again restrict ourselves to states $\psi(t)$ for which $S^{\psi(t)}$ is invertible.
A trajectory now satisfies $\psi(t)\in\cM(t)$ if and only if
$$2\Re\pscal{\partial_t\psi(t),\cO_m\psi(t)}=o_m'(t).$$
Let $\nu_{\psi(t)}$ denote the unique vector in the normal space $\cN_{\psi(t)}$ such that, for all $m$,
\begin{equation}
2\Re\pscal{\nu_{\psi(t)},\cO_m\psi(t)}=o_m'(t).
\label{eq:def_nu_psi}
\end{equation}
Note that existence and uniqueness of such a $\nu_{\psi(t)}$ follow from the fact that \eqref{eq:def_nu_psi} consists of $M$ linearly independent equations, and that $\cN_{\psi(t)}$ is $M$-dimensional.
From the definition~\eqref{eq:N_psi_M} of $\cN_\psi$, we can write \begin{equation}\label{eq:nupsi}
\nu_{\psi(t)}=\sum_{m=1}^Mc_{\psi,m}(t)\cO_m\psi(t)
\end{equation}
with $c_{\psi,m}(t)\in\R$ and thus obtain the linear equation in~$\R^M$
\begin{equation}
S^{\psi(t)}c_\psi(t)= \frac{o'(t)}{2},
\label{eq:def_nu_psi_linear}
\end{equation}
where $o'(t) \in \R^M$ is the column vector with entries $o_m'(t)$, $m=1,...,M$.
The latter admits a unique solution because we always assume that we work in the region where $S^{\psi(t)}$ is invertible. We conclude that \emph{a trajectory $t\mapsto\psi(t)$ starting on $\cM(0)$ satisfies the time-dependent constraints if and only if}
\begin{equation}
\boxed{\partial_t\psi(t)-\nu_{\psi(t)}\in\cT_{\psi(t)}}
\label{eq:stays_on_Mt}
\end{equation}
with $\nu_{\psi(t)}$ defined by~\eqref{eq:nupsi}-\eqref{eq:def_nu_psi_linear}.
The velocity $\partial_t\psi(t)$ must therefore be the sum of two terms. The first term $\nu_{\psi(t)}$ is here to reproduce the normal displacement of the tangent space due to the time variations of the constraints. The second term must be in $\cT_{\psi(t)}$ to ensure that the trajectory continues to fulfill the constraints at all times. When the constraints do not depend on time the unique solution to the linear equation~\eqref{eq:def_nu_psi} is $\nu_{\psi(t)}=0$ and we recover the condition~\eqref{eq:stays_on_M}.

As a conclusion, we have to modify the reference Schrödinger equation~\eqref{eq:Schrodinger_unconstrained} to enforce the property~\eqref{eq:stays_on_Mt}. Writing the modified Schrödinger equation in the form~\eqref{eq:Schrodinger_perturbed}, we arrive at the condition
\begin{equation}
-iH(t)\psi(t)-iF\big(t,\psi(t)\big)\psi(t)-\nu_{\psi(t)}\in\cT_{\psi(t)}
\label{eq:constraint_correction_term}
\end{equation}
on the correction term $F\big(t,\psi(t)\big)$. We see that there are many possible choices, because we can add an arbitrary vector in $i\cT_{\psi(t)}$ to $F\big(t,\psi(t)\big) \psi(t)$ and still obtain~\eqref{eq:constraint_correction_term}.
More precisely, \eqref{eq:constraint_correction_term} says that the space of possible correction terms is an affine space, meaning that if $F_0(t,\psi)$ and $F_1(t,\psi)$ are two possible correction terms at given time $t$ and state $\psi$, then $F_\lambda(t,\psi) = \lambda F_1(t,\psi) + (1-\lambda) F_0(t,\psi)$ is a possible correction as well for any real number $\lambda$. 
In this paper, we discuss multiple natural choices of $F$ and study the resulting equations.

Historically, this problem was studied first by Dirac~\cite{Dirac-30} and Frenkel~\cite[p.~253]{Frenkel-34}. These authors were however considering the manifold of Slater determinants, that has a natural complex structure, so that several of the difficulties we will encounter do not occur in their case. Meyer, K\v{u}car and Cederbaum were probably the first to notice in~\cite{KucMeyCed-88} that several definitions of the constrained Schrödinger dynamics that were the same for Dirac and Frenkel can give different results in other situations. The general theory was further developed in~\cite{BroLatKesLeu-88,Raab-00,HacGuaShiHaeDemCir-20,MarBur-20,LasSu-22} but, to the best of our knowledge, it was never applied to the case of constraining expectation values $\pscal{\psi(t),\cO_m\psi(t)}$. This is the situation of interest for TDDFT. 
In the next two sections we describe the main two methods used in practice to choose the correction term $F$ so that the solution to the modified Schr\"odinger equation stays on the considered manifold at all times, and specify the expressions of $F$ in the case of a manifold defined by such constrained expectation values. We will also introduce a new family of constrained Schrödinger equations that interpolate between the previous two methods, which does not seem to have been considered so far.

\section{Variational principle}\label{sec:TDVP}
We first consider the variational principle, which is the traditional approach in  TDDFT.

\subsection{Stationarity of the action}
Dirac~\cite{Dirac-30} was the first to emphasize the importance of the symplectic structure of Schrödinger's equation and to mention that the exact Schrödinger trajectory~\eqref{eq:Schrodinger_unconstrained} can be recovered by requiring the action to be stationary, similarly as the Wei\ss~action principle in classical mechanics. This point of view was further developed by Kramer and Saraceno in their famous book~\cite{KraSar-81}.

In this section, we investigate under which conditions on the observables $\cO_m$ one can define a unique constrained dynamics solely based on the \textbf{stationarity of the action}. 
As we will recall below, standard TDDFT fits into this framework by applying the action principle to the wavefunction at fixed density (and not to the density itself as in~\cite{RunGro-84,vanLeeuwen-98,Vignale-08}). Because of this link with standard TDDFT, it is important to understand how the action principle works (or, in fact, why it does not work so easily in this particular case).

As before, we first look at time-independent constraint values $o_m$. We give ourselves a state $\psi_0$ in $\cM$ and a trajectory $t\in[0,T]\mapsto \psi(t)$ drawn on $\cM$ with $\psi(0)=\psi_0$. We assume that $\psi(t)$ satisfies $\det(S^{\psi(t)})\neq0$ for all $0\leq t<T$, so that the trajectory stays in a region where $\cM$ is a smooth manifold. We also assume that $\psi$ is continuously differentiable in $t$. The final time $T$ is rather arbitrary and the hope is that it can be sent to infinity, or at least varied as we wish. In practice $T$ will only be finite when the trajectory reaches the boundary~$\cS$ of the manifold $\cM$, that is, for which $\det(S^{\psi(T)})=0$.

We say that a trajectory $t\mapsto \psi(t)$ satisfies the \textbf{variational principle (VP)} if the action functional
\begin{equation}
\cA[\phi]:=\int_0^T\pscal{\phi(t),(i\partial_t-H(t))\phi(t)}\,\dt
\label{eq:action}
\end{equation}
is stationary at the trajectory $\psi$. This means that if we consider trajectories $\phi$ that are obtained by small deformations of the trajectory $\psi$, fixing the two end points $\phi(0)=\psi(0)$ and $\phi(T)=\psi(T)$, the action should only vary to second order in the displacement $\phi-\psi$. The difficulty here is of course that we work in a curved space, hence we have to deform $\psi$, while staying on the manifold $\cM$, in order to impose the constraints. Fortunately, the assumption that $S^{\psi(t)}$ is invertible for all $0\leq t< T$ implies that there is a tubular neighborhood around the trajectory, where we can easily move the curve, by essentially pushing it in directions belonging to the tangent space (see Figure~\ref{fig:TDVP}).

Using integration by parts for the term involving the time derivative, we can write
\begin{multline*}
\cA[\phi]-\cA[\psi]=2\int_0^T \!\! \Re\pscal{\phi(t)-\psi(t),(i\partial_t-H(t))\psi(t)}\,\dt\\
+\cA[\phi-\psi]
\end{multline*}
for any continuously differentiable trajectory $\phi$ such that $\phi(0)=\psi(0)$, $\phi(t) \in \mathcal M$ for all $t \in [0,T]$,  and $\phi(T)=\psi(T)$. The second term on the right-hand side is quadratic in the displacement and thus, we require that the first term vanishes. As $\phi(t)$ and $\psi(t)$ are both on $\mathcal M$, to first order, the small displacement $\phi(t)-\psi(t)$ must be in the tangent space, and
since we can move parts of the path independently from each other, this leads to the local-in-time condition that
$$\Re\pscal{h(t),i\partial_t\psi(t) -H(t)\psi(t)}=0,\qquad \forall h(t)\in \cT_{\psi(t)}$$
or, in other words,
\begin{equation}
\boxed{i\partial_t\psi(t)-H(t)\psi(t)\in \cN_{\psi(t)}.}
\label{eq:TDVP_pre}
\end{equation}
The details of this argument can for instance be read in~\cite[Prop.~3.1]{LasSu-22}. Stationarity of the action can therefore be re-interpreted as the requirement that the residue of the unconstrained Schrödinger equation $i\partial_t\psi(t)-H(t)\psi(t)$ belongs to the normal space of $\psi(t)$ all along the trajectory. This is to make it orthogonal to the small deformations of the path that can only happen in the directions of the tangent space. We emphasize once again the importance of the real structure of the problem in this result due to the appearance of the real scalar product.

In the literature, the above principle is often called the time-dependent variational principle (TDVP)~\cite{KraSar-81,KucMeyCed-88,MarBur-20}, the {Lagrangian action principle}~\cite{HacGuaShiHaeDemCir-20} or the {Kramer--Saraceno principle}~\cite{LasSu-22}. We simply call it the \textbf{variational principle (VP)} for brevity.

\begin{figure}[t]
\includegraphics{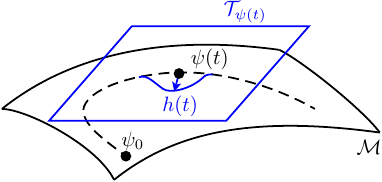}
\caption{In the variational principle, an optimal trajectory $t\mapsto\psi(t)$ is by definition such that the action functional is stationary against small deformations staying on the constraining manifold $\cM$. To first order, such deformations $h(t)$ are vectors in the tangent space $\cT_{\psi(t)}$.
\label{fig:TDVP}}
\end{figure}

So far everything was very general and holds for any kind of manifold. For constraints of the form~\eqref{eq:constraints_Oj} the normal space is given by~\eqref{eq:N_psi_M} and we conclude from~\eqref{eq:TDVP_pre} that we must find 
some real numbers $v_m(t)$ such that
\begin{equation}
\boxed{i\partial_t\psi(t)=\left(H(t)+\sum_{m=1}^Mv_m(t)\,\cO_m\right)\psi(t).}
\label{eq:TDVP}
\end{equation}
This is a very natural modified Schrödinger equation. The action principle says that we should add to our Hamiltonian $M$ time-dependent Lagrange multipliers $v_m(t)$ in order to enforce the constraints $\pscal{\psi(t),\cO_m\psi(t)}=o_m$. We can interpret the $v_m(t)$'s as some kind of external potential used to adjust the desired expectation values. The difficulty is, of course, that the $v_m(t)$'s are unknown and must be found. This is not such an easy task because, for time-independent constraint values $o_m$, we must also fulfill the condition~\eqref{eq:stays_on_M}, and it is not so clear how adding to $-iH(t)\psi(t)$ a vector $-i \sum v_m(t) \cO_m \psi(t) \in i\cN_{\psi(t)}$ could be useful to make the resulting $\partial_t\psi(t)$ belong to $\cT_{\psi(t)}$.

Before we discuss the solution of this problem, let us first turn to time-dependent constraint values $o_m(t)$. If we keep the action principle as it is without change, we arrive at the exact same condition~\eqref{eq:TDVP_pre}, and thus the same
constrained Schrödinger equation~ \eqref{eq:TDVP}. 
This is because the {variation of the} action involves the difference $\phi-\psi$ and therefore the additional vectors $\nu_{\phi(t)} \in \mathcal N_{\phi(t)}$  and $\nu_{\psi(t)} \in \mathcal N_{\psi(t)}$ appearing in the condition~\eqref{eq:stays_on_Mt} are equal at first order. We therefore look for $v_m(t)$ such that
$$-i\left(H(t)+\sum_{m=1}^Mv_m(t)\,\cO_m\right)\psi(t)-\nu_{\psi(t)}\in\cT_{\psi(t)}$$
and we hope that there is a unique solution to this problem.
Neither existence nor uniqueness of such $v_m(t)$'s is obvious. In fact, in TDDFT the existence of such $v_m(t)$'s is closely linked to the question of time-dependent $v$-representability.

\subsection{The symplectic case}
It is possible to solve the above problem under a very natural assumption on $\psi_0$ that involves the \textbf{symplectic structure} of $\R^{2d}$~\cite{RowRymRos-80,KraSar-81}. As we will explain below, this strategy is only valid for some families of observables $(\cO_m)_{m=1}^M$ and \textbf{does not work at all for TDDFT}. We think it is important to explain it in detail to clarify the situation and better emphasize the difficulties of TDDFT. Moreover, the symplectic case is interesting since it may be applied to current-density functional theory, which we consider in \Cref{sec.current.DFT}.

The argument relies on a new $M\times M$ antisymmetric matrix defined by
\begin{equation}
(A^{\psi})_{mn}:=\Im\pscal{\cO_m\psi,\cO_n\psi}=\pscal{\psi,\frac{i[\cO_n,\cO_m]}{2}\psi}
\label{eq:A_symplectic}
\end{equation}
with $[A,B]:=AB-BA$ denoting the commutator of the operators $A$ and $B$.
It is defined similarly as the overlap matrix $S^\psi$ in~\eqref{eq:overlap} but with an imaginary part instead of a real part. The imaginary part is reminiscent of the symplectic structure but we shall not give more details on this. As long as the matrix $A^{\psi(t)}$ is invertible along the dynamics, the following theorem provides existence and uniquess of the $v_m(t)$'s in~\eqref{eq:TDVP}.

\begin{theorem}[Variational principle in symplectic case]\label{thm:TDVP_symplectic}
Let $\cO_1,...,\cO_M$ be a family of $d\times d$ Hermitian matrices.
Assume that on some time interval $[0,T]$, we have
\begin{itemize}
\item a continuous map $t\mapsto H(t)$ of Hermitian matrices;
\item $M$ continuously-differentiable functions $t\mapsto o_1(t),...,o_M(t)$;
\item an initial state $\psi_0\in\C^d$ satisfying the constraints $\pscal{\psi_0,\cO_m\psi_0}=o_m(0)$ for $m=1,...,M$, as well as $\|\psi_0\|=1$.
\end{itemize}
We further assume that $A^{\psi_0}$ is invertible. Then there exists a maximal time $0<T'\leq T$ and \textbf{uniquely defined continuous functions} $t \mapsto v_1(t),...,v_M(t)$ on $[0,T')$ such that the solution $\psi(t)$ to the equation~\eqref{eq:TDVP} with $\psi(0)=\psi_0$ satisfies the constraints $\pscal{\psi(t),\cO_m\psi(t)}=o_m(t)$ for all $m=1,...,M$
and $A^{\psi(t)}$ stays invertible for all $0\leq t<T'$.
\end{theorem}

The theorem says that the inverse problem of finding the $v_m(t)$'s from the constraints can always be solved uniquely for some short time and then continues to do so as long as the antisymmetric matrix $A^{\psi(t)}$ stays invertible. If $T'<T$ then we must have $\det(A^{\psi(T')})=0$. The invertibility of $A^{\psi}$ requires in particular that we have an even number $M$ of constraints, because an antisymmetric matrix in odd dimension always has $0$ in its spectrum.

Notice that the invertibility of the matrix $A^\psi$ implies the invertibility of $S^\psi$. (If $\cO_{k_0}\psi=\sum_{k\neq k_0}\alpha_k\cO_k\psi$ for some real coefficients $\alpha_k$, then the $k_0$-th column of $A^\psi$ is the same linear combination of the other columns.) In particular, the trajectory obtained in the theorem solves the variational principle. Of course it could well be that $A^{\psi(t)}$ ceases to be invertible at a time $T'$ whereas $S^{\psi(t)}$ still is. In that case we have nothing to say about what happens at later times.

Let us finally remark that there is no need to impose the normalization constraint $\|\psi(t)\|=1$ in this theory. The latter is automatically satisfied for the solution to~\eqref{eq:TDVP} because it is a Schrödinger equation involving a (time-dependent) Hermitian matrix. In fact, we cannot add the identity matrix to the list of the constrained observables $\cO_m$ because this would make the matrix $A^{\psi}$ not invertible.

Let us now quickly describe the proof of Theorem~\ref{thm:TDVP_symplectic}. We need to show that the $v_m(t)$'s exist and are uniquely defined. To this end we differentiate the constraints and obtain
\begin{align}
o_m'(t)&=\pscal{\psi(t),i\left[H(t)+\sum_{n=1}^M v_n(t)\cO_n,\cO_m\right]\psi(t)}\nn\\
&=\pscal{\psi(t),i[H(t),\cO_m]\psi(t)}+2\sum_{n=1}^M(A^{\psi(t)})_{mn}v_n(t),\label{eq:diff_constraint_TDVP}
\end{align}
where we have used that $\pscal{\psi,i[\cO_n,\cO_m]\psi}=2A^\psi_{mn}$.
This is how the antisymmetric matrix $A^\psi$ arises. We have assumed that $A^{\psi_0}$ is invertible and we look for continuous functions $t \mapsto v_m(t)$ and a continuously-differentiable $t \mapsto \psi(t)$. By continuity of the determinant we conclude that $A^{\psi(t)}$ must stay invertible for some short time. Inverting $A^{\psi(t)}$ in~\eqref{eq:diff_constraint_TDVP} allows us to express the $v_m(t)$'s as functions of the current time $t$ and $\psi(t)$. In fact, denoting by $v(t)$ the column vector whose components are the sought-after $v_m(t)$, we conclude that $v(t)=V\big(t,\psi(t)\big)$ with
\begin{align}
V(t,\psi)&:=(A^{\psi})^{-1}\, b(t,\psi),\nn\\
b_m(t,\psi)&:=\frac{o_m'(t)}2+\Im\pscal{H(t)\psi,\cO_m\psi}.
\label{eq:def_v_NL}
\end{align}
This proves that $\psi(t)$ must solve the highly nonlinear Schrödinger equation
\begin{equation}
i\partial_t\psi(t)=\left(H(t)+\sum_{m=1}^M V_m \big(t,\psi(t)\big)\cO_m\right)\psi(t).
\label{eq:TDVP_NL}
\end{equation}
Conversely, if we are able to solve this equation we obtain some $v_m(t)$'s solving~\eqref{eq:diff_constraint_TDVP}. Integrating in time and using that $\pscal{\psi_0,\cO_m\psi_0}=o_m(0)$ we conclude that the desired constraints hold. The existence of a unique solution to~\eqref{eq:TDVP_NL} follows from the Cauchy--Lipschitz theorem, using that the right-hand side is continuous in $t$ and Lipschitz in $\psi$, as long as $A^\psi$ remains invertible. This concludes the proof of the theorem.

\medskip
The correction term is thus
$$F(t,\psi(t)):=\sum_{m=1}^M V_m(t,\psi(t))\cO_m,$$
and $F(t,\psi(t)) \psi(t)$ is related to the \emph{symplectic projection} of $-iH(t)\psi(t)-\nu_{\psi(t)}$ on $\cT_{\psi(t)}$, which is why the symplectic matrix $A^\psi$ arises in this theory (see~\cite{LasSu-22} for more details).

As a last remark, let us consider the solution $\psi^\text{S}(t)$ to the reference Schrödinger equation~\eqref{eq:Schrodinger_unconstrained} and assume that the matrix $A^{\psi^\text{S}(t)}$ is invertible for all $0\leq t< T$. The uniqueness in Theorem~\ref{thm:TDVP_symplectic} implies that if we choose $o_m(t):=\pscal{\psi^\text{S}(t),\cO_m\psi^\text{S}(t)}$ then we must have $v_m(t)=0$ and $\psi(t)=\psi^\text{S}(t)$. It is reassuring that the variational principle does nothing if the constraints are already satisfied, at least under the invertibility assumption on $A^{\psi^\text{S}(t)}$.

\subsection{The case of commuting observables}\label{sec:TDVP_commuting}
We cannot expect that the matrix $A^\psi$ will always be invertible and, in fact, in many cases it never is. In this subsection we consider the special case where the observables $\cO_m$'s commute:
$$\cO_m\cO_n=\cO_n\cO_m, \qquad m,n=1,...,M.$$
Then we have $A^\psi\equiv0$ for every $\psi$ and the above symplectic theory does not work at all. This is the situation in TDDFT since the observables $\cO_m =\sum_{\sigma\in\{\uparrow,\downarrow\}}a_{m\sigma}^\dagger a_{m\sigma}$ commute.

In this case it is possible, but harder, to express the sought-after potentials $v_m(t)$ in terms of $\psi(t)$. Differentiating once as in~\eqref{eq:diff_constraint_TDVP} we obtain, using that the $\cO_m$'s commute,
\begin{equation}
o_m'(t)=\pscal{\psi(t),i[H(t),\cO_m]\psi(t)}.
\label{eq:condition_supp_TDVP}
\end{equation}
While, in the symplectic case when $A^{\psi(t)}$ is invertible, Eq.~\eqref{eq:diff_constraint_TDVP} allows one to identify the $v_m$'s, Eq.~\eqref{eq:condition_supp_TDVP} does not provide any information on the sought-after potentials. On the other hand, evaluating this relation at time $t=0$, we see that we obtain new constraints we had not anticipated! We have to require that the initial state $\psi_0$ satisfies
\begin{equation}
o_m'(0)=-2\Im\pscal{H(0)\psi_0,\cO_m\psi_0},\qquad m=1,...,M,
\label{eq:derivative}
\end{equation}
otherwise there is no hope of finding a solution. In other words, not all initial conditions lead to a well-defined trajectory. For instance, if $\psi_0$ is an eigenfunction of $H(0)$, then $o_m'(0)=0$. This difficulty is well explained in a different context in~\cite{RowRymRos-80}.

Differentiating once more, we obtain a linear equation involving the $v_m(t)$'s:
\begin{align}
o_m''(t) & = -\pscal{\psi(t),[H(t),[H(t),\cO_m]]\psi(t)}
\nonumber
\\ & \quad 
+\pscal{\psi(t),i[H'(t),\cO_m]\psi(t)}
\nonumber
\\ & \quad 
-\sum_{n=1}^M v_n(t)\pscal{\psi(t),[\cO_n,[H(t),\cO_m]]\psi(t)}.
\label{eq:vanLeeuwen_finite_dim}
\end{align}
We will call this the \textbf{van Leeuwen equation} because for TDDFT this is exactly the fundamental equation appearing in Ref.~\cite{vanLeeuwen-99}. 
The new relation~\eqref{eq:vanLeeuwen_finite_dim} suggests to introduce a new real symmetric matrix
\begin{align}
K_{mn}^{\psi}(t)&:=\frac12\pscal{\psi,[\cO_n,[H(t),\cO_m]]\psi}\label{eq:matrix_M_TDVP}\\
&=\Re\pscal{\cO_n\psi,H(t)\cO_m\psi}-\Re\pscal{H(t)\psi,\cO_m\cO_n\psi}.\nn
\end{align}
The matrices $S^\psi$ and $A^\psi$ only involve the observables $\cO_m$'s and can thus be interpreted as purely geometric objects. By contrast, the new matrix $K^{\psi}(t)$ involves the Hamiltonian $H(t)$, hence is model-dependent. 

The new constraint~\eqref{eq:derivative} is, loosely speaking, because the inverse problem of finding the $v_m$'s is of order two in time instead of order one as it was in the symplectic case. 
By arguing exactly as before with the addition initial condition~\eqref{eq:derivative}, we can prove the following.

\begin{theorem}[Variational principle for commuting observables]\label{thm:TDVP_commuting}
Let $\cO_1,...,\cO_M$ be a family of \textbf{commuting} $d\times d$ Hermitian matrices.
Assume that on some time interval $[0,T]$, we have
\begin{itemize}
\item a continuously-differentiable map $t\mapsto H(t)$ of Hermitian matrices;
\item $M$ twice continuously-differentiable functions $t\mapsto o_1(t),...,o_M(t)$;
\item a normalized initial state $\psi_0$ satisfying the two constraints
\begin{align*}
\pscal{\psi_0,\cO_m\psi_0}&=o_m(0),\\
2\Im\pscal{H(0)\psi_0,\cO_m\psi_0}&=-o_m'(0),
\end{align*}
for $m=1,...,M$.
\end{itemize}
We further assume that the $M\times M$ symmetric matrix $K^{\psi_0}(0)$ defined in~\eqref{eq:matrix_M_TDVP} is invertible. Then there exists a maximal time $0<T'\leq T$ and \textbf{uniquely defined continuous functions} $t \mapsto v_1(t),...,v_M(t)$ on $[0,T']$ such that the solution $\psi(t)$ to the equation~\eqref{eq:TDVP} with $\psi(0)=\psi_0$ satisfies the constraints $\pscal{\psi(t),\cO_m\psi(t)}=o_m(t)$ for all $m=1,...,M$
and $K^{\psi(t)}(t)$ stays invertible for all $0\leq t<T'$.
\end{theorem}

Exactly as before we note that the invertibility of $K^{\psi(t)}(t)$ implies the full-rank condition~\eqref{eq:full_rank}, i.e., the invertibility of $S^{\psi(t)}$, hence the stationarity of the action for the trajectory $\psi(t)$ provided by the theorem. Also, the normalization $\|\psi(t)\|=1$ is automatic and the identity matrix should not be included in the list of the $\cO_m$'s, otherwise the matrix $K^{\psi(t)}(t)$ cannot be invertible.

If we interpret the $v_m(t)$'s as some kind of time-dependent local potentials, the result is of the same spirit as the Runge--Gross uniqueness theorem~\cite{RunGro-84} for TDDFT, with the important difference that we get both existence and uniqueness, and that we do not need any analyticity~\cite{FouLamLewSor-16}. In the context of TDDFT on a finite lattice, Theorem~\ref{thm:TDVP_commuting} was sketched in~\cite{FarTok-12} and~\cite[Sec.~4]{RugPenLee-15}. Equation~\eqref{eq:condition_supp_TDVP} is then a discrete analogue of the continuity equation.

\section{Geometric principle}\label{sec:McLachlan}
In this section we turn to a completely different principle for defining the constrained dynamics, which is more based on the geometric structure of the manifold of interest $\cM$. The output $\psi(t)$ will be very different.

\subsection{Geometric Schrödinger equation}

We again start with the case of time-independent constraint values $o_m$. We recall that $\cM$ is the regular part of the set defined in~\eqref{eq:manifold_constraints}, and $\cT_\psi$ and $\cN_\psi$ the tangent and normal spaces at $\psi\in\cM$ introduced in~\eqref{eq:T_psi_M} and~\eqref{eq:N_psi_M}. We work in the region $\cM$ in which the matrix
$$(S^\psi)_{mn}=\Re\pscal{\cO_m\psi,\cO_n\psi}$$
is invertible. Recall that any smooth trajectory $t\mapsto \psi(t)$ drawn on $\cM$ must satisfy $\partial_t\psi(t)\in\cT_{\psi(t)}$ for all $t$.

The \textbf{geometric principle (GP)} requires that, all along the trajectory, the vector \textbf{$\partial_t\psi(t)$ is ``the closest it can be'' to $-i H(t)\psi(t)$} (for the distance associated with the norm $\|\cdot\|$ of $\C^d$), within the tangent space $\cT_{\psi(t)}$. In other words, $\partial_t\psi(t)$ must be equal to the orthogonal projection of $-iH\psi(t)$ on $\cT_{\psi(t)}$ for the real inner product $\Re \langle \cdot,\cdot\rangle$ (see Figure~\ref{fig:McLachlan}). This is equivalent to saying that
\begin{equation}
\Re\pscal{h(t),\partial_t \psi(t)+iH(t)\psi(t)}=0,\qquad\forall h(t)\in\cT_{\psi(t)},
\nonumber
\label{}
\end{equation}
or, equivalently,
\begin{equation}
\boxed{\partial_t\psi(t)+iH(t)\psi(t)\in\cN_{\psi(t)}.}
\label{eq:McLachlan_pre}
\end{equation}
The interpretation is that we are trying to be as close as we can to the original Schrödinger equation, while satisfying the constraints. The reader is urged to notice the difference with the variational principle based on the action in~\eqref{eq:TDVP_pre}, which is similar to~\eqref{eq:McLachlan_pre} but with an additional $i$ in front. The condition~\eqref{eq:McLachlan_pre} can of course also be written as $i \partial_t\psi(t)- H(t)\psi(t)\in i\cN_{\psi(t)}$. In general the two spaces $i\cN_{\psi(t)}$ and $\cN_{\psi(t)}$ are different, hence these two conditions will usually give completely different solutions.

A manifold for which the tangent and normal spaces are complex-linear and not just real-linear (i.e., $i\cT_\psi=\cT_\psi$ and $i\cN_\psi=\cN_\psi$) is called a \emph{Kähler manifold}. For such a manifold, the variational principle~\eqref{eq:TDVP_pre} and the geometric principle~\eqref{eq:McLachlan_pre} give the same answer, possibly up to an irrelevant global phase factor. When the manifold is defined by fixing expectation values of some observables $\cO_m$, the spaces $\cT_\psi$ and $\cN_\psi$ will not be complex-linear, however. For instance, if the $\cO_m$'s are commuting matrices (as they are in TDDFT) then we have $\Im\pscal{\cO_m\psi,\cO_n\psi}=0$ for all $m,n$, as we have seen in the previous section. This implies $\Re\pscal{i\cO_m\psi,\cO_n\psi}=0$, that is $i\cN_\psi\subset\cT_\psi$.
More abstractly, the expectation value $\pscal{\psi, \cO_m \psi}$ is a function of both $\psi$ and $\psi^\dagger$ and hence not a holomorphic function of $\psi$. In particular, the manifold $\cM$, which is given by level sets of such functions, cannot be a Kähler manifold.

Since the condition~\eqref{eq:McLachlan_pre} is only based on an orthogonal projection relying on the geometric structure of the constraint manifold $\cM$, we call it the \textbf{geometric principle (GP)}. In the literature, the latter has been given different names. It is often called the \textbf{McLachlan principle}~\cite{McLachlan-64}. The interpretation in terms of a projection on the tangent space was already mentioned by Frenkel in~\cite[p. 253]{Frenkel-34} and, apparently, in an appendix of the Russian version of the book~\cite{Dirac-30} by Dirac (see the comments in~\cite{Lubich-08}). The McLachlan principle was rediscovered by Nazarov in~\cite{Nazarov-85,NazSil-24}.

For some time there was some controversy as to whether the variational and geometric principles are the same or not. In~\cite{LowMuk-72}, it was explained that they coincide if $\cT_\psi$ is complex-linear, but no comment was made about situations in which they would fail to coincide. Meyer, Ku\v{c}ar and Cederbaum~\cite{KucMeyCed-88} studied a particular non-Kähler manifold and were the first to emphasize the difference between the two conditions. They mention that ``the McLachlan principle offers a clearer and more appealing view in the way how an optimal result is determined''. A few months later, the authors of \cite{BroLatKesLeu-88} wrote that their ``investigation is motivated by the fact that considerable confusion and ambiguity exists in the literature concerning this question''.
The situation was reviewed and clarified in several recent works, including~\cite{Raab-00,HacGuaShiHaeDemCir-20,MarBur-20,LasSu-22}.

\begin{figure}[t]
\includegraphics[width=\columnwidth]{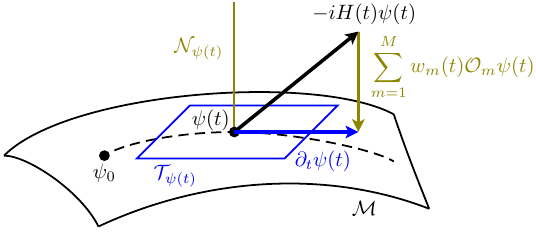}
\caption{In the geometric principle, an optimal trajectory $t\mapsto\psi(t)$ is by definition such that the tangent $\partial_t\psi(t)$ to the trajectory is at every time the orthogonal projection of $-iH(t)\psi(t)$ on the tangent space $\cT_{\psi(t)}$. This projection is $\partial_t \psi(t) = -iH(t)\psi(t) + \sum_{m=1}^M w_m(t) \cO_m(t)\psi(t)$ for some real numbers $w_m(t)$.
\label{fig:McLachlan}}
\end{figure}

Next, we go on with the study of the geometric principle. One difference compared with the variational principle is that we will have to put the normalization of the state, $\|\psi(t)\|=1$, in the list of constraints. More precisely, we require that the identity matrix  $I_d$ belongs to $\textrm{span}_\R(\cO_1,...,\cO_M)$. Using the formula in \eqref{eq:N_psi_M} of $\cN_\psi$, the condition~\eqref{eq:McLachlan_pre} can be re-expressed in the form
\begin{equation}
\boxed{i\partial_t\psi(t)=\left(H(t)+i\sum_{m=1}^Mw_m(t)\,\cO_m\right)\psi(t),}
\label{eq:McLachlan}
\end{equation}
where $w_m(t)$ are some real numbers.
This equation looks exactly like~\eqref{eq:TDVP} except for the additional $i$ in front of the corrective term. This is the general form of the geometric principle, which is thus very different from the variational principle in~\eqref{eq:TDVP}.

Because of the $i$, at first sight it may seem that we are introducing a non-Hermitian perturbation of the Hamiltonian $H(t)$. But this is not the best way to see it. Recall that this is a nonlinear equation where the $w_m(t)$'s depend on $\psi(t)$ itself, and not a general linear equation. In fact, since we add the preservation of the norm of $\psi(t)$ as a constraint, we can rewrite \eqref{eq:McLachlan} in an equivalent form involving a $\psi$-dependent Hermitian Hamiltonian, where the norm-preservation is more apparent. 
Indeed,
from the relation $(\rd/\dt) \pscal{\psi(t),\psi(t)}=0$, we find that
\begin{equation}
\sum_{m=1}^M w_m(t)\pscal{\psi(t),\cO_m\psi(t)}=0.
\label{eq:norm_preservation}
\end{equation}
This tells us that, for every time, $\psi(t)$ is orthogonal to the correction term $F\big(t,\psi(t)\big) \psi(t)$ where
$$F\big(t,\psi(t)\big):=i\sum_{m=1}^Mw_m(t)\,\cO_m.$$
Hence we can rewrite~\eqref{eq:McLachlan} in the equivalent form
\begin{equation}
\boxed{i\partial_t\psi(t)=\Big(H(t)+G\big(t,\psi(t)\big)\Big)\psi(t)}
\label{eq:McLachlan_rank2}
\end{equation}
with the Hermitian operator (for every given $\psi(t)$)
\begin{align}
G\big(t,\psi(t)\big)&=\big|F\big(t,\psi(t)\big)\psi(t)\big\rangle\big\langle\psi(t)\big|
\nonumber\\
&\qquad +\big|\psi(t)\big\rangle\big\langle F\big(t,\psi(t)\big)\psi(t)\big|\nn\\
&=i\left[\sum_{m=1}^Mw_m(t)\cO_m\,,\, |\psi(t)\rangle\langle\psi(t)|\right]\label{eq:geom_pert_pure}
\end{align}
describing the ``geometric'' modification to Schrödinger's equation. This is a nonlocal Hermitian perturbation of rank two that is very different from the simple potential-type perturbation we got in the variational principle.

Next we turn to time-dependent constraint values $o_m(t)$. Recall that we need to fulfill the condition~\eqref{eq:stays_on_Mt} that $\partial_t\psi(t)-\nu_{\psi(t)}\in\cT_{\psi(t)}$. The geometric principle simply requires that $\partial_t\psi(t)-\nu_{\psi(t)}$ be the orthogonal projection of $-iH(t)\psi(t)$ onto $\cT_{\psi(t)}$, leading to the condition that
$$\partial_t\psi(t)-\nu_{\psi(t)}+iH(t)\psi(t)\in\cN_{\psi(t)}.$$
Since $\nu_{\psi(t)}$ already belongs to the normal space $\cN_{\psi(t)}$ by definition, the resulting equation takes the exact same form as in~\eqref{eq:McLachlan} for time-dependent constraints. The potential $w_m$ appearing there is the sum of the component corresponding to the displacement of the tangent space and the one associated with the geometric projection.

Next we discuss the existence and uniqueness of the $w_m(t)$'s in~\eqref{eq:McLachlan}. The statement is the following.

\begin{theorem}[Geometric principle]\label{thm:McLachlan}
Let $\cO_1,...,\cO_M$ be a family of Hermitian $d\times d$ matrices with $I_d\in\tspan_\R(\cO_1,...,\cO_M)$. Assume that on some time interval $[0,T]$, we have
\begin{itemize}
\item a continuous map $t\mapsto H(t)$ of Hermitian matrices;
\item $M$ continuously-differentiable functions $t\mapsto o_1(t),...,o_M(t)$;
\item a normalized initial state $\psi_0\in\C^d$ satisfying the constraints $\pscal{\psi_0,\cO_m\psi_0}=o_m(0)$ for $m=1,...,M$.
\end{itemize}
We further assume that the $M\times M$ symmetric matrix $S^{\psi_0}$ is invertible. Then there exists a maximal time $0<T'\leq T$ and \textbf{uniquely defined continuous functions} $t \mapsto w_1(t),...,w_M(t)$ on $[0,T')$ such that the solution $\psi(t)$ to the equation~\eqref{eq:McLachlan} with $\psi(0)=\psi_0$ satisfies the constraints $\pscal{\psi(t),\cO_m\psi(t)}=o_m(t)$ for all $m=1,...,M$
with $S^{\psi(t)}$ staying invertible for all $0\leq t<T'$.
\end{theorem}

It is very satisfactory that this theorem only relies on the invertibility of the matrix $S^\psi$ which, as we have said, just means that the constraints are independent from each other. Hence the trajectory exists as long as it does not hit the boundary of the manifold $\cM$, where $S^\psi$ ceases to be invertible. This is because the projection on the tangent space is always well defined.

Recall that the short-time existence of a solution for the variational principle, as given by  \Cref{thm:TDVP_commuting}, relies on the invertibility of the matrix $K^{\psi_0}(t=0)$. As we mentioned after Theorem~\ref{thm:TDVP_commuting}, this implies the invertibility of $S^{\psi_0}$, so that the solution to the geometric principle also exists by Theorem~\ref{thm:McLachlan}. In this sense, the geometric principle is more robust than the variational principle. 

As before, if we take the observables from the reference Schrödinger equation, $o_m(t)=\pscal{\psi^{\rm S}(t),\cO_m\psi^{\rm S}(t)}$, and if $S^{\psi^{\rm S}(t)}$ stays invertible, then uniqueness implies that $w_m(t)=0$ and hence $\psi(t)=\psi^{\rm S}(t)$. The reference Schrödinger equation will not be modified if the constraints are already satisfied.

The proof of Theorem~\ref{thm:McLachlan} proceeds in a similar way as before, noticing that, this time,
\begin{align}
o_m'(t)&=\pscal{\left(-iH(t)+\sum_{n=1}^M w_n(t)\cO_n\right)\psi(t),\cO_m\psi(t)}\nn\\
&\qquad +\pscal{\psi(t),\cO_m\left(-iH(t)+\sum_{n=1}^M w_n(t)\cO_n\right)\psi(t)}\nn\\
&=\pscal{\psi(t),i[H(t),\cO_m]\psi(t)}+2\sum_{n=1}^M S_{mn}^{\psi(t)} w_n(t).\label{eq:diff_constraint_McLachlan}
\end{align}
Hence we obtain in vector form $w(t)=W\big(t,\psi(t)\big)$ with
\begin{equation}
W(t,\psi):=(S^{\psi})^{-1} \, b(t,\psi)
\label{eq:def_w_NL}
\end{equation}
with the same vector $b(t,\psi)$ as in~\eqref{eq:def_v_NL}.

\subsection{Interpretation as sources and sinks}
The variational and geometric principles can lead to very different solutions. For instance, for commuting observables the variational principle only works thanks to a complicated interplay between the Hamiltonian $H(t)$ and the observables $\cO_m$, which is expressed within the matrix $K^{\psi}(t)$ in~\eqref{eq:matrix_M_TDVP}. The latter involves double commutators of the form  $[\cO_n,[H(t),\cO_m]]$ and those should not vanish. In particular, if the $\cO_m$'s commute with $H(t)$, the variational principle is just unable to reproduce time-dependent values $o_m(t)$ because the expectation values $\pscal{\psi(t),\cO_m\psi(t)}$ will always be constant in time, whatever $v_m(t)$ we put in the modified Schrödinger equation~\eqref{eq:TDVP}. On the contrary, the geometric principle is perfectly able to make $\pscal{\psi(t),\cO_m\psi(t)}$ be equal to whatever we like, even when the $\cO_m$'s commute with $H(t)$.

To illustrate this fact, let us for instance consider the extreme case where $H(t)\equiv0$ and the observables are $\cO_m=|e_m\rangle\langle e_m|$ (projection onto the $m$th vector $e_m=(0,...,1,...,0)$ of the canonical basis of $\C^d$) for $m=1,...,d$. The solution to the equation~\eqref{eq:TDVP} of the variational principle with an arbitrary external potential $v(t)$ is
$$\psi^\text{VP}(t)=\begin{pmatrix}
e^{-i\int_0^t v_1(s)\,\rd s}\psi_1^\text{VP}(0)\\
\vdots\\
e^{-i\int_0^t v_{d}(s)\,\rd s}\psi^\text{VP}_{d}(0)
\end{pmatrix}.$$
Since this is only adding complex phases, we will never be able to modify the density $|\psi^\text{VP}_m(t)|^2$ this way. On the other hand, the geometric principle gives
$$\psi^\text{GP}(t)=\begin{pmatrix}
e^{\int_0^t w_1(s)\,\rd s}\psi^\text{GP}_1(0)\\
\vdots\\
e^{\int_0^t w_{d}(s)\,\rd s}\psi^\text{GP}_{d}(0)
\end{pmatrix}=\begin{pmatrix}
\sqrt{\frac{o_1(t)}{o_1(0)}}\,\psi^\text{GP}_1(0)\\
\vdots\\
\sqrt{\frac{o_d(t)}{o_d(0)}}\,\psi^\text{GP}_d(0)\\
\end{pmatrix}$$
if we choose $w_m(t)=o_m'(t)/(2o_m(t))$. The numbers $w_m(t)$ can be interpreted as sources when $w_m(t)>0$ and as sinks when $w_m(t)<0$. In this simple example the variational principle is unable to follow the given density because it changes the phase without touching the modulus. The geometric principle works without problem because it changes the modulus and not the phase. 

We note that the potential $w$ is akin to a complex absorbing potential (see, e.g.,~\cite{RisMey-JPB-93,SanCed-02}, and~\cite{Ern-JCP-06,ZhoErn-JCP-12} in the context of DFT) which is normally used to calculate resonances or simulate open quantum systems. But we stress that in our case, the potential $w$ does not entail an open quantum system, since the norm of the state is preserved by the condition~\eqref{eq:norm_preservation}. Moreover, contrary to a complex absorbing potential, the potential $w$ is not externally fixed, but is self-consistently determined along the trajectory to impose the desired constraints. Similarly, the nonlocal version of the correction $G$ in the geometric Schrödinger equation~\eqref{eq:McLachlan_rank2} is akin to the nonlocal term appearing in the master equation of open quantum systems (see, e.g.,~\cite{GebCar-PRL-04,YueRodAsp-PCCP-09,AgoVen-PRB-13}), but again in our case the correction is not externally fixed, but is self-consistently determined.

\section{Oblique principle}
\label{sec:oblique}

So far we have seen two different principles forcing the solution to the Schrödinger equation to fulfill some given constraints in the form $\pscal{\psi(t),\cO_m\psi(t)}=o_m(t)$. The \textbf{variational principle} from Section~\ref{sec:TDVP} requires that
$$\Re\pscal{h(t),i\partial_t \psi(t)-H(t)\psi(t)}=0,\qquad\forall h(t)\in\cT_{\psi(t)},$$
which is equivalent to requiring the existence of $v_1(t),...,v_M(t)\in\R$ such that
\begin{equation}
i\partial_t\psi(t)=\left(H(t)+\sum_{m=1}^Mv_m(t)\cO_m\right)\psi(t).
\label{eq:TDVPbis}
\end{equation}
The \textbf{geometric principle} from Section~\ref{sec:McLachlan} requires that (notice again the factor $i$ in front of $h$)
\begin{equation*}
\Re\pscal{ih(t),i\partial_t\psi(t)-H(t)\psi(t)}=0
\qquad\forall h(t)\in\cT_{\psi(t)},
\label{eq:Rih}
\end{equation*}
which is equivalent to requiring the existence of $w_1(t),...,w_M(t)\in\R$ such that
\begin{equation}
i\partial_t\psi(t)=\left(H(t)+i\sum_{m=1}^Mw_m(t)\cO_m\right)\psi(t).
\label{eq:McLachlanbis}
\end{equation}

\begin{figure}[t]
\begin{tikzpicture}
\draw[-{Latex[length=2mm]}] (-2.3,0) -- (2.3,0);
\draw[-{Latex[length=2mm]}] (0,-2.3) -- (0,2.3);
\draw[thick] (0,0) circle [radius = 2];
\draw (1,0) arc [radius=1, start angle=0, end angle = 30];
\draw (0,0) -- (1.73205080757,1);
\node[right] at (1,0.3) {$\theta$}; 
\node[below right] at (2,0) {VP};
\node[below left] at (-2,0) {VP};
\node[above left] at (0,2) {GP};
\node[below left] at (0,-2) {GP};
\node[above right] at (2,0) {$0$};
\node[above right] at (0,2) {$\frac{\pi}{2}$};
\end{tikzpicture}
\caption{The oblique principle continuously interpolates between the variational and geometric principles, using a parameter $\theta$ similar to an angle. The model shares the properties of the geometric principle for all $\theta\neq0$ modulo $\pi$. In the limit $\theta\to0$ modulo $\pi$ one recovers the variational principle but the limit is very singular. 
\label{fig:oblique_theta}}
\end{figure}
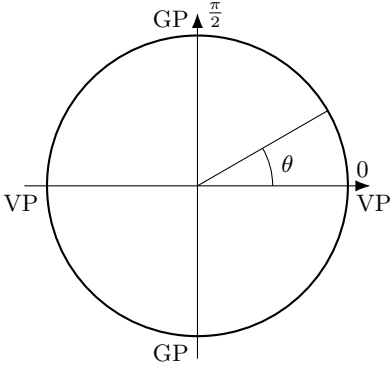

It is possible to continuously interpolate between the two previous principles, a bit like how the Robin boundary condition is an interpolation between the Dirichlet and Neumann boundary conditions in elliptic boundary value problems. We thus fix an angle $-\pi\leq \theta\leq\pi$ and require the condition
\begin{equation}
\Re\pscal{e^{i\theta}h(t),i\partial_t \psi(t)-H(t)\psi(t)}=0,\qquad \forall h(t)\in\cT_{\psi(t)}
\nonumber
\label{eq:oblique_theta_pre}
\end{equation}
so that $\theta=0$ corresponds to the variational principle and $\theta=\pi/2$ to the geometric principle (modulo $\pi$). 
This condition can be rewritten as
\begin{equation}
\boxed{i\partial_t \psi(t)=\left(H(t)+e^{i\theta}\sum_{m=1}^Mu_m^\theta(t)\cO_m\right)\psi(t)}
\label{eq:oblique_theta}
\end{equation}
for some $u_1^\theta(t),...,u_M^\theta(t)\in\R$, where we recall that here $\theta$ is a fixed angle. We call this the \textbf{oblique principle}. We have $u^{0}=-u^{\pi}=v$ (the solution to the variational principle~\eqref{eq:TDVPbis}) and $u^{\pi/2}=-u^{-\pi/2}=w$ (the solution to the geometric principle~\eqref{eq:McLachlanbis}), 
see Figure~\ref{fig:oblique_theta}. Note that in the oblique principle we put both a real potential $v^\theta(t) = \cos (\theta) u^\theta(t)$ and an imaginary potential $iw^\theta(t) = i\sin(\theta) u^\theta(t)$, but we assume they are proportional and not independent.

Because we want to compare the two principles in the framework of TDDFT, we assume in the whole section that the observables $\cO_m$'s commute.  In this case, the variational principle is only well-posed under additional constraints on $\psi_0$ (see Eq.~\eqref{eq:derivative}). For $\theta\neq0$ (mod $\pi$) there is no such issue and the oblique principle is always well-posed. Indeed, we obtain as before
\begin{multline}
o_m'(t)
=\pscal{\psi(t),i[H(t),\cO_m]\psi(t)}\\
+2\sum_{n=1}^M\big(\cos(\theta)\,A^{\psi(t)}_{mn}+\sin(\theta)\,S^{\psi(t)}_{mn}\big)u^\theta_n(t).
\nonumber
\end{multline}
For commuting observables we have $A^\psi\equiv0$ and therefore we get
$$o_m'(t)=\pscal{\psi(t),i[H(t),\cO_m]\psi(t)}+2\sin(\theta)\sum_{n=1}^MS^{\psi(t)}_{mn}u^\theta_n(t).$$
which we can rewrite in a vector form as
$$u^\theta(t)=\frac{W\big(t,\psi(t)\big)}{\sin(\theta)}$$
with the function $W(t,\psi)$ defined in~\eqref{eq:def_w_NL}.
We thus obtain the highly nonlinear equation
\begin{multline}
i\partial_t\psi(t)\\
=\Biggl(H(t)+\bigg(i+\frac1{\tan(\theta)}\bigg)\sum_{m=1}^M W_m \big(t,\psi(t)\big)\cO_m\Biggl)\psi(t).
\label{eq:oblique2}
\end{multline}
The latter can again be written in Hermitian form as
\begin{equation}
\boxed{i\partial_t\psi(t)=\Big(H(t)+G^\theta\big(t,\psi(t)\big)\Big)\psi(t)}
\label{eq:oblique_rank2}
\end{equation}
with the rank-two Hermitian operator
\begin{align*}
G^\theta\big(t,\psi(t)\big)&:=\big|F^\theta\big(t,\psi(t)\big)\psi(t)\big\rangle\big\langle\psi(t)\big|
\nonumber\\
&\phantom{xxx}+\big|\psi(t)\big\rangle\big\langle F^\theta\big(t,\psi(t)\big)\psi(t)\big|
\end{align*}
and
$$F^\theta(t,\psi):=\bigg(i+\frac1{\tan(\theta)}\bigg)\sum_{m=1}^M W_m(t,\psi)\cO_m.$$
We thus obtain the following theorem.
\begin{theorem}[oblique principle]\label{thm:oblique}
Let $\cO_1,...,\cO_M$ be a family of Hermitian $d\times d$ commuting matrices with $I_d\in\tspan_\R(\cO_1,...,\cO_M)$. Assume that on some time interval $[0,T]$, we have
\begin{itemize}
\item a continuous map $t\mapsto H(t)$ of Hermitian matrices;
\item $M$ continuously-differentiable functions $t\mapsto o_1(t),...,o_M(t)$;
\item a normalized initial state $\psi_0\in\C^d$ satisfying the constraints $\pscal{\psi_0,\cO_m\psi_0}=o_m(0)$ for $m=1,...,M$.
\end{itemize}
We further assume that the $M\times M$ symmetric matrix $S^{\psi_0}$ is invertible. Let $-\pi<\theta<\pi$ with $\theta\neq0$. Then there exists a maximal time $0<T'\leq T$ and \textbf{uniquely defined continuous functions} $t \mapsto u^\theta_1(t),...,u^\theta_M(t)$ on $[0,T')$ such that the solution $\psi(t)$ to the equation~\eqref{eq:oblique_theta} with $\psi(0)=\psi_0$ satisfies the constraints $\pscal{\psi(t),\cO_m\psi(t)}=o_m(t)$ for all $m=1,...,M$ with $S^{\psi(t)}$ staying invertible for all $0\leq t<T'$.
\end{theorem}

It is not difficult to see that the potential $u^\theta(t)$ is a smooth function of $\theta$ whenever $\theta$ does not approach 0 or $\pi$. The limit $\theta\to0$ (mod $\pi$) is very singular, however. This is due to the factor $1/\tan(\theta)$ in the corresponding equation~\eqref{eq:oblique2}. To illustrate the possible behavior of the system, we will give some details about the limit $\theta\to0$ in the case of one qubit in \Cref{sec:qubit_oblique} and Appendix~\ref{app:qubit_oblique}.

\section{Illustration of the different principles on a single qubit}
\label{sec:qubit}

As an illustration of the different principles, let us consider one qubit, that is, the simplest non-trivial quantum system with state space $\C^2$. This is equivalent to the model of a single particle on a two-site lattice, already treated in~\cite{Baer-08,LiUll-08,FarTok-12}. We take the time-independent Hamiltonian
\begin{equation}
H=-\sigma_1=\begin{pmatrix}
0&-1\\
-1&0
    \end{pmatrix}
\nonumber
\end{equation}
and write states as
\begin{equation}
\psi=\begin{pmatrix}
\psi_1\\ \psi_2
\end{pmatrix}.
\nonumber
\end{equation}
We take the observables
\begin{equation}
\cO_1=\begin{pmatrix}
1&0\\
0&0
    \end{pmatrix}, \qquad \cO_2=\begin{pmatrix}
0&0\\
0&1
    \end{pmatrix}.
\label{eq:qbit_O}
\end{equation}
When we fix the expectation values of these two observables, the matrix $S^\psi$ characterizing the regular part $\cM$ of the set of states satisfying the constraints is given by
$$S^{\psi}=\begin{pmatrix}
|\psi_1|^2&0\\
0&|\psi_2|^2
\end{pmatrix}.$$
Our theory requires it to be invertible, which means that $\rho_1=|\psi_1|^2$ and $\rho_2=|\psi_2|^2=1-\rho_1$ be strictly positive. In the usual Bloch sphere representation of the qubit state $\psi$, this corresponds to removing the North and South poles, see  Figure~\ref{fig:Bloch-sphere}. Let us now fix the two expectation values $\pscal{\psi,\cO_m\psi}=|\psi_m|^2=:\rho_m$ with $m=1,2$ and describe the corresponding set $\cC$ defined in Eq.~\eqref{eq:manifold_constraints}. If $0<\rho_1=1-\rho_2<1$ then we obtain a circle of latitude fully included in the regular part, hence $\cC=\cM$ in this case. In contrast, if $\rho_1=0$ or $1$ then $\cC$ consists of only one the poles and $\cM$ is empty. In this very simple example, the set $\cC$ is either fully regular or fully singular. Of course, when the constraints depend on time, the trajectory $\psi(t)$ hits the poles when $\rho_1(t)$ reaches $0$ or $1$. We will thus always assume that we work over an interval of times where
$$0<\rho_1(t)<1.$$
The geometric principle will be well-posed over all such times.

\begin{figure}
    \centering
    \includegraphics[width=0.6\linewidth]{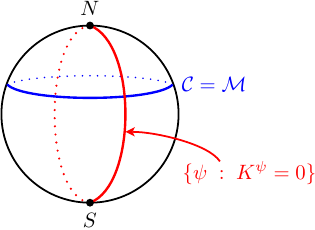}
    \caption{Bloch sphere representation of the states $\psi$ of the qubit. The matrix $S^\psi$ is invertible everywhere except at the South and North poles, corresponding to $\rho_1=0$ or $1$. The $\psi$'s of fixed  density $0<\rho_1=1-\rho_2<1$ correspond to circles of latitude, on which $S^\psi$ is always invertible, hence $\cC=\cM$. On the other hand, the number $K^\psi=\Re(\overline{\psi_2}\psi_1)$ in~\eqref{eq:qubit_K} vanishes on the latitude circle of relative angle $\alpha_2-\alpha_1=\pi/2$ modulo $\pi$, corresponding to $\psi=e^{i\alpha}(\sqrt{\rho_1},\pm i\sqrt{\rho_2})$. The solution to the variational principle can never cross this circle, so that the Bloch sphere is split into two disconnected parts. The solutions to the geometric and oblique principles can perfectly cross the circle and only the two poles have to be avoided.}
    \label{fig:Bloch-sphere}
\end{figure}

\subsection{Variational principle}\label{sec:qubit_TDVP}

For the variational principle, already studied in~Ref.~\cite{FarTok-12}, we only consider the observable $\cO_1$ since $\cO_2=I_2-\cO_1$. The matrix $K^\psi$ appearing in the van Leeuwen equation~\eqref{eq:vanLeeuwen_finite_dim} is just a number in this situation, given by
\begin{align}
K^\psi&=\frac12 \pscal{\psi,[\cO_1,[H,\cO_1]]\psi}\nn\\
&=\Re(\overline{\psi_2}\psi_1)=\sqrt{\rho_1\rho_2}\cos(\alpha_1-\alpha_2).\label{eq:qubit_K}
\end{align}
because $[\cO_1,[H,\cO_1]]=-H$. Here, we have written the qubit state $\psi$ in the form
\begin{equation}
\psi=\begin{pmatrix}
e^{i\alpha_1}\sqrt{\rho_1}\\
e^{i\alpha_2}\sqrt{\rho_2}
\end{pmatrix}.
\label{eq:qbit_form_psi}
\end{equation}
In the Bloch sphere representation, a state $\psi$ of the form~\eqref{eq:qbit_form_psi} corresponds to the point on the latitude circle $\cM$ with longitude $\alpha_2-\alpha_1$. Hence $K^\psi=0$ occurs on the vertical circle where $\alpha_2-\alpha_1=\pi/2$ modulo $\pi$, see Fig.~\ref{fig:Bloch-sphere}.  

Let us now discuss time-dependent $v$-representability within the variational principle. We give ourselves a function $0\leq \rho_1(t)\leq1$ and work under the constraint
$$\pscal{\psi(t),\cO_1\psi(t)}=|\psi_1(t)|^2=\rho_1(t).$$
Then we will automatically get
$$|\psi_2(t)|^2=1-|\psi_1(t)|^2=1-\rho_1(t)=:\rho_2(t).$$
As we said above, we assume $0<\rho_1(t)<1$. Next we ask what kind of functions $\rho_1(t)$ can be attained with the dynamics given by the variational principle
\begin{equation}
i\partial_t\psi(t)=\big(H+v_1(t)\cO_1\big)\psi(t)
\label{eq:qubit_TDVP}
\end{equation}
for an arbitrary $v_1(t)$. Differentiating the constraint as in~\eqref{eq:condition_supp_TDVP}, we find
\begin{equation}
\rho_1'(t)=2\Im\big(\overline{\psi_2(t)}\psi_1(t)\big).
\label{eq:additional_constraint_TDVP}
\end{equation}
Using the inequality $|\Im(\overline{\psi_2(t)}\psi_1(t))|\leq \sqrt{\rho_1(t)\rho_2(t)}$ we conclude that $\rho_1(t)$ must satisfy the constraint
\begin{equation}
\frac{|\rho_1'(t)|}{\sqrt{\rho_1(t)(1-\rho_1(t))}}\leq2.
\label{eq:representability_condition_TDVP}
\end{equation}
The interpretation is that the density cannot vary too fast~\cite{Baer-08,LiUll-08,Ver-PRL-08,RugPenLee-15}. Intuitively this is because our qubit described with the Hamiltonian $H$ can only change its state at finite speed. We will prove below that the condition~\eqref{eq:representability_condition_TDVP} is almost necessary and sufficient for the representability of the density. More precisely, Theorem~\ref{thm:TDVP_commuting} will tell us that any density satisfying~\eqref{eq:representability_condition_TDVP} with a strict inequality ($<$) instead of a large inequality ($\leq$) is representable with a unique potential $v_1(t)$.

To see this, we look for solutions of~\eqref{eq:qubit_TDVP} in the form~\eqref{eq:qbit_form_psi} with time-dependent angles. Plugging in the equation~\eqref{eq:qubit_TDVP} involving the unknown potential $v_1(t)$, we obtain the system of differential equations
\begin{equation}
\left\{
\begin{aligned}
\alpha_1'(t)
& = -v_1(t)+e^{i(\alpha_2(t)-\alpha_1(t))}\sqrt{\frac{\rho_2(t)}{\rho_1(t)}}+i\frac{\rho_1'(t)}{2\rho_1(t)}, 
\\
\alpha_2'(t)
& = e^{i(\alpha_1(t)-\alpha_2(t))}\sqrt{\frac{\rho_1(t)}{\rho_2(t)}}+i\frac{\rho_2'(t)}{2\rho_2(t)}.    
\end{aligned}
\right.
\label{eq:angles_TDVP}
\end{equation}
The imaginary part provides
\begin{equation}
\sin(\alpha_2(t)-\alpha_1(t))=-\frac{\rho_1'(t)}{2\sqrt{\rho_1(t)\rho_2(t)}}=\frac{\rho_2'(t)}{2\sqrt{\rho_1(t)\rho_2(t)}}
\label{eq:qbit-add_constraint}
\end{equation}
so that the relative angle $\beta:=\alpha_2-\alpha_1$ is in fact fixed by the function $\rho_1$. In fact, equation~\eqref{eq:qbit-add_constraint} can be seen to correspond to the condition~\eqref{eq:condition_supp_TDVP}, since
\begin{equation}
\rho_1'(t) = 2\Im\big(\overline{\psi_2(t)}\psi_1(t)\big)=2\sqrt{\rho_1(t)\rho_2(t)}\sin(\alpha_1(t)-\alpha_2(t)).\label{eq:qubit_rho_1_beta}
\end{equation}

To be able to apply Theorem~\ref{thm:TDVP_commuting}, we work on an interval of times over which $|\rho_1'(t)|<2\sqrt{\rho_1(t)\rho_2(t)}$, so that $\cos(\alpha_2(t)-\alpha_1(t))\neq0$ by~\eqref{eq:qbit-add_constraint} and thus $K^{\psi(t)}\neq0$ in~\eqref{eq:qubit_K}. This is exactly~\eqref{eq:representability_condition_TDVP} with a strict inequality. Hence for any initial state $\psi(0)$ satisfying $|\psi_1(0)|^2=\rho_1(0)$ and any $\rho_1(t)$ satisfying the above conditions, we get by Theorem~\ref{thm:TDVP_commuting} a unique solution $v_1(t)$ and $\psi(t)$ over the whole considered interval of times.

To find more explicit formulas, we look at the real parts in~\eqref{eq:angles_TDVP} which provide
\begin{equation}
\left\{
\begin{aligned}
\alpha_1'(t) 
& 
=-v_1(t)+\cos(\beta(t))\sqrt{\frac{\rho_2(t)}{\rho_1(t)}}, 
\\
\alpha_2'(t) 
&= \cos(\beta(t))\sqrt{\frac{\rho_1(t)}{\rho_2(t)}}.
\end{aligned}
\right.
\label{eq:angles_TDVP_bis}
\end{equation}
After substracting the two equations we find the value of the potential
\begin{equation}
v_1(t)=\beta'(t)+\cos(\beta(t))\frac{\rho_2(t)-\rho_1(t)}{\sqrt{\rho_1(t)\rho_2(t)}}
\label{eq:qubit_v_1}
\end{equation}
which allows us to express $\alpha_1$ and $\alpha_2$ as functions of $\beta$ only and thus provides the solution $\psi(t)$. Note that  the potential $v_1(t)$ can also be expressed in terms of $\rho_1(t)$, $\rho_1'(t)$, and $\rho''_1(t)$ using~\eqref{eq:qubit_rho_1_beta}~\cite{FarTok-12}.

To clarify what this all means, let us look at the case of a time-independent constraint, i.e. $\rho_1(t)\equiv\rho_1$. In this case we find from~\eqref{eq:qbit-add_constraint} that $\beta(t)\equiv0$ or $\beta(t)\equiv\pi$ modulo $2\pi$ and $v_1(t)\equiv\cos(\beta) (\rho_2-\rho_1)/\sqrt{\rho_1\rho_2}$ for all times. Therefore, using~\eqref{eq:angles_TDVP_bis}, the solution is
$$\psi(t)=e^{i\alpha_0-iEt}\begin{pmatrix}
\sqrt{\rho_1}\\ \pm\sqrt{\rho_2}
\end{pmatrix}=e^{-iEt}\psi(0)$$
for some $\alpha_0=\alpha_1(0)$ and $E:=-\cos(\beta)\sqrt{\rho_1/\rho_2}$. This means that $\psi(0)$ is an eigenstate of $H+v_1\cO_1$ with eigenvalue $E$ and then $\psi(t)$ is a trivial Schrödinger stationary solution.
Note that the modified Hamiltonian
$$H+v_1\cO_1=\begin{pmatrix}
v_1&-1\\
-1&0
\end{pmatrix}$$
has, for any fixed $v_1$, exactly two simple eigenvalues $E^\pm[v_1]=(v_1\pm\sqrt{v_1^2+4})/2$. Up to a phase factor $e^{i\alpha_0}$, the two eigenstates can thus be written in the form
$$\begin{pmatrix}
\sqrt{\rho_1[v_1]}\\ \pm\sqrt{1-\rho_1[v_1]}
\end{pmatrix}$$
for some $\rho_1[v_1]$. By the Perron--Frobenius theorem, the ground state is the one with positive coefficients and the excited state is the one changing sign. The function  $\rho_1[v_1]$ is  found to be
$$\rho_1[v_1]:=\frac{(\sqrt{v_1^2+4}-v_1)^2}{4+(\sqrt{v_1^2+4}-v_1)^2}.$$
This is the stationary potential-to-density map that we must invert in order to express the potential in terms of $\rho_1$. As we have found before, the inverse is
\begin{equation}
v^{{\rm ad},\sigma}_1[\rho_1]=\sigma\frac{1-2\rho_1}{\sqrt{\rho_1(1-\rho_1)}}=\sigma\frac{\rho_2-\rho_1}{\sqrt{\rho_1\rho_2}}
\label{eq:qubit_TDVP_v_rho}
\end{equation}
with $\sigma=1$ for the ground state and $\sigma=-1$ for the excited state. We interpret this function as the adiabatic potential, hence the notation. Our conclusion is that, in the variational principle with a time-independent constraint, we have to find the unique $v_1$ so that our $\psi(0)$ is either the ground or excited state of the Hamiltonian $H+v_1\cO_1$. The additional constraint~\eqref{eq:additional_constraint_TDVP} at $t=0$ is here to ensure that this is possible. Then the unique solution to our problem is trivial.

Finally, we remark that we can write the full time-dependent solution~\eqref{eq:qubit_v_1} in terms of the adiabatic potential as
\begin{multline}
v_1(t)=v_1^{{\rm ad},\sigma}[\rho_1(t)]\\
+\beta'(t)+\big(\sigma \cos(\beta(t))-1\big)v_1^{{\rm ad},\sigma}[\rho_1(t)]
\label{eq:qubit_v_1_ad+}
\end{multline}
The terms on the second line form the correction to the adiabatic part. They can be expressed in terms of $\rho(t)$ only using~\eqref{eq:qbit-add_constraint}.

\subsection{Geometric principle}
\label{sec:qubit_geometric}
For the geometric principle, we have to fix both $\pscal{\psi,\cO_1\psi}=\rho_1$ and $\pscal{\psi,\cO_2\psi}=\rho_2=1-\rho_1$ to ensure the normalization of $\psi$. Recall that $S^\psi$ stays invertible under the sole condition that $0<\rho_1(t)<1$ for all times. No other condition is needed. This is an important difference compared with the variational principle, which had the additional constraint~\eqref{eq:representability_condition_TDVP} on the velocity $\rho'_1(t)$. With the geometric principle, all time-dependent densities are representable, even those changing very fast. This was already observed in~\cite[Appendix]{LiUll-08}. 

Looking again for the solution in the form
\begin{equation}
\psi(t)=\begin{pmatrix}
e^{i\alpha_1(t)}\sqrt{\rho_1(t)}\\
e^{i\alpha_2(t)}\sqrt{\rho_2(t)}
\end{pmatrix},
\label{eq:qbit_form_psi_bis}
\end{equation}
we obtain a system of ordinary differential equations similar to~\eqref{eq:angles_TDVP}
\begin{equation}
\left\{
\begin{aligned}
\alpha_1'(t) & = -i w_1(t)+e^{i(\alpha_2(t)-\alpha_1(t))}\sqrt{\frac{\rho_2(t)}{\rho_1(t)}}+i\frac{\rho_1'(t)}{2\rho_1(t)},
\\
\alpha_2'(t) & = -i w_2(t)+e^{i(\alpha_1(t)-\alpha_2(t))}\sqrt{\frac{\rho_1(t)}{\rho_2(t)}}+i\frac{\rho_2'(t)}{2\rho_2(t)}.
\end{aligned}
\right.
\label{eq:angles_GP}
\end{equation}
Taking the real parts leads to
$$
\left\{
\begin{aligned}
\alpha'_1(t)& =\sqrt{\frac{\rho_2(t)}{\rho_1(t)}}\cos(\alpha_2(t)-\alpha_1(t)),\\
\alpha'_2(t)& =\sqrt{\frac{\rho_1(t)}{\rho_2(t)}}\cos(\alpha_2(t)-\alpha_1(t)).
\end{aligned}
\right.
$$
The solutions can be expressed in terms of the relative angle $\beta=\alpha_2-\alpha_1$, which itself solves the equation
\begin{equation}
\beta'(t)=\frac{\rho_1(t)-\rho_2(t)}{\sqrt{\rho_1(t)\rho_2(t)}}\cos(\beta(t)).
\label{eq:GV_relative_angle}
\end{equation}
On the other hand, by taking the imaginary parts of~\eqref{eq:angles_GP}, the potential is found to be
\begin{equation}
\left\{
\begin{aligned}
w_1(t)& =\frac{\rho'_1(t)}{2\rho_1(t)}+\sqrt{\frac{\rho_2(t)}{\rho_1(t)}}\sin(\beta(t)), \\
w_2(t)& =\frac{\rho'_2(t)}{2\rho_2(t)}-\sqrt{\frac{\rho_1(t)}{\rho_2(t)}}\sin(\beta(t)).
\end{aligned}
\right.
\label{eq:qubit_GV_potential_beta}
\end{equation}
These formulas can also be found in~\cite[Appendix]{LiUll-08}. When $\rho_1(t)=1/2$ and $\alpha_2(0)=\alpha_1(0)$ or $\alpha_2(0)=\alpha_1(0)+\pi$, we find the stationary solutions corresponding to the two eigenstates of $H$. If $\rho_1(t)\equiv\rho_1$ is time-independent with $\rho_1\neq1/2$ and $\beta(0)$ is not equal to $0$ or $\pi$ modulo $2\pi$, then the relative angle $\beta(t)$ will depend on time through~\eqref{eq:GV_relative_angle} and the solution $\psi(t)$ will not be an eigenstate.

Solving explicitly the differential equation~\eqref{eq:GV_relative_angle}, we can express $\sin(\beta(t))$ as a function of the given density
\begin{equation*}
\sin(\beta(t))=\tanh\left(A_0+\int_0^t\frac{\rho_1(s)-\rho_2(s)}{\sqrt{\rho_1(s)\rho_2(s)}}\,\rd s\right)
\end{equation*}
with $A_0:=\tanh^{-1}(\sin(\beta_0))$. 
Inserting this expression in~\eqref{eq:qubit_GV_potential_beta} gives the potential $w(t)$ as an explicit functional of $\rho$.

\subsection{Oblique principle}\label{sec:qubit_oblique}

We assume $\theta\neq0$ and fix again the expectation values of $\cO_1$ and $\cO_2$ in~\eqref{eq:qbit_O}. Writing the solution in the same form as~\eqref{eq:qbit_form_psi_bis}, after a tedious but straightforward calculation, we obtain that the relative angle $\beta=\alpha_2-\alpha_1$  must solve the equation
\begin{multline}
\tan(\theta)\beta'(t)=\sin(\beta(t))\frac{\rho_1(t)+\rho_2(t)}{\sqrt{\rho_1(t)(\rho_2(t))}}\\
-\frac{\rho_2'(t)}{2\rho_2(t)}+\frac{\rho_1'(t)}{2\rho_1(t)}
-\tan(\theta)\cos(\beta(t))\frac{\rho_2(t)-\rho_1(t)}{\sqrt{\rho_1(t)(\rho_2(t))}}.
\label{eq:qbit_oblique}
\end{multline}
The angles are then given by

$$
\left\{
\begin{aligned}
\alpha'_1(t) & =  \cos(\beta(t))\sqrt{\frac{\rho_2(t)}{\rho_1(t)}} 
\\ & \quad - \frac{1}{\tan(\theta)}\bigg(\frac{\rho_1'(t)}{2\rho_1(t)} +\sin(\beta(t))\sqrt{\frac{\rho_2(t)}{\rho_1(t)}}\bigg), 
\\
\alpha'_2(t)& = \cos(\beta(t))\sqrt{\frac{\rho_1(t)}{\rho_2(t)}} 
\\ & \quad - \frac{1}{\tan(\theta)}\bigg(\frac{\rho_2'(t)}{2\rho_2(t)} -\sin(\beta(t))\sqrt{\frac{\rho_2(t)}{\rho_1(t)}}\bigg),    
\end{aligned}
\right.
$$
whereas the potentials are given by
$$
\left\{
\begin{aligned}
u^\theta_1(t) & = \frac1{\sin(\theta)}\left(\frac{\rho_1'(t)}{2\rho_1(t)}+\sqrt{\frac{\rho_2(t)}{\rho_1(t)}}\sin(\beta(t))\right),\\
u^\theta_2(t)&=\frac1{\sin(\theta)}\left(\frac{\rho_2'(t)}{2\rho_2(t)}-\sqrt{\frac{\rho_1(t)}{\rho_2(t)}}\sin(\beta(t))\right).    
\end{aligned}
\right.
$$

To illustrate what is going on in the limit $\theta\to0$, let us look at the time-independent case $\rho_1(t)\equiv\rho_1$. In Figure~\ref{fig:qubit_oblique} we display the solution $\beta(t)$ for different values of $\theta$ near $0$ and the initial condition $\beta(0)=1.3$. Let us emphasize that the latter does not satisfy the condition~\eqref{eq:qbit-add_constraint} of the variational principle that requires $\beta(0)=0$ modulo $\pi$. When $\theta$ tends to $0$, the function $\beta(t)$ is compressed and looks more and more like a step function with either the value $\pi$ (for $\theta>0$) or $0$ (for $\theta<0$). On the other hand, the corresponding potential $u^\theta$ has a very large peak at the origin and otherwise converges to a constant potential, as we expect. In Appendix~\ref{app:qubit_oblique}, we explain what is going on in details. The system is moving extremely fast to one of the two eigenfunctions of $H+v_1[\rho_1]\cO_1$ that are the only solutions to the variational principle, and then stays there for infinite time. In the limit $\theta\to0$ the ``fast'' part of the trajectory gives rise to Dirac deltas in the potentials and the complex phases, whose role is to modify the initial condition into one that is compatible with the additional condition~\eqref{eq:additional_constraint_TDVP} of the variational principle. If we take $\beta(0)=0$, we observe a similar behavior for $\theta<0$ but not for $\theta>0$.

As a conclusion, the oblique principle converges to the variational principle in the limit $\theta\to0$, but possibly with a different initial condition. This is reflected in the presence of Dirac delta's in the potentials and phases. The new initial condition depends on whether we approach $\theta=0$ from negative or positive values, hence the limit is discontinuous. In this respect, the variational principle is an extreme case of a whole class of better behaved geometric models.

\begin{figure}[t]
\includegraphics[width=\columnwidth]{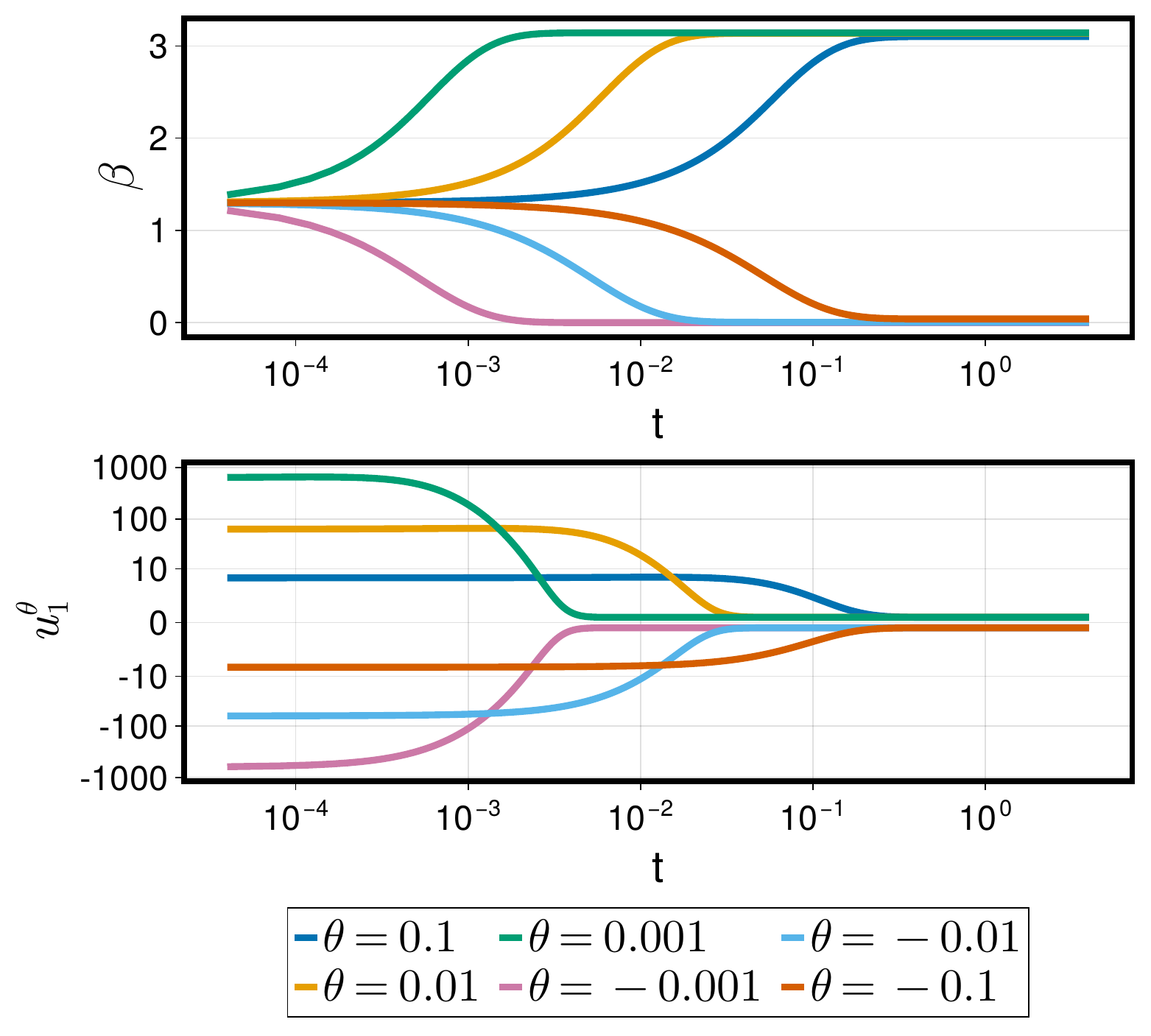}
\caption{Top panel: time evolution in logarithmic scale of the solution $\beta(t)$ obtained from the oblique principle for a single qubit, shown for several values of the parameter $\theta$ with $\beta(0)=1.3$ and $\rho_1=0.7$. Bottom panel: corresponding potentials $u_1^\theta(t)$. For visualization purposes, we plot the transformed quantity $g(u_1^\theta)$ with $g(x) = \operatorname{sign}(x)\log_{10}(1 + |x|)$ in order to highlight the convergence toward Dirac deltas.\label{fig:qubit_oblique}}
\end{figure}

\section{Time-dependent density-functional theory for fermions on a finite lattice}\label{sec:Hubbard}
In this section we apply the previous theory to the case of $N$ spin--$1/2$ fermions hopping on $M$ sites, within the framework of TDDFT. Everything would work in a similar fashion for bosons.

\subsection{Description of the system}
We denote by $a_{m\sigma}^\dagger$ and $a_{m\sigma}$ the creation and annihilation operators of a particle of spin $\sigma\in\{\uparrow,\downarrow\}$ at site number $m\in\{1,...,M\}$. We use the second-quantization formalism for convenience, although we work with a fixed number of particles. We take as observables $\cO_m \equiv \cN_m$ where
\begin{equation}
\cN_m:=a_{m\uparrow}^\dagger a_{m\uparrow}+a_{m\downarrow}^\dagger a_{m\downarrow}
\label{eq:cN}
\end{equation}
is the density operator on site $m$, i.e. the operator that counts the number of particles at site $m$. The density of a $N$-electron state $\Psi$ is the vector $\rho_\Psi$ given by
$$(\rho_\Psi)_m:=\pscal{\Psi,\cN_m\Psi}$$
so that fixing these expectation values amounts to fixing the density, as is appropriate for DFT. Note that the $\cN_m$'s are \textbf{commuting operators}, hence we can apply the theory developed in Section~\ref{sec:TDVP_commuting} for the variational principle. 

The one-particle state space is 
$$
\gH:=\C^{2M}\simeq (\C\times\{\uparrow,\downarrow\})^M
$$
and one-particle states will be seen as vectors $\phi_{m\sigma}$ indexed by pairs of indices $(m,\sigma)\in \{1,...,M\}\times\{\uparrow,\downarrow\}$. We consider the following Hamiltonian written in the usual second-quantized form as
\begin{multline}
\bH_U(t):=\sum_{\sigma,\sigma'\in\{\uparrow,\downarrow\}}\sum_{m,m'=1}^M \Big(h_{m'\sigma', m\sigma}(t)\,a^\dagger_{m\sigma}a_{m'\sigma'}\\
+U_{m'\sigma',m\sigma}\,a^\dagger_{m\sigma}a_{m\sigma}a^\dagger_{m'\sigma'}a_{m'\sigma'}\Big),
\label{eq:Hubbard_H}
\end{multline}
where the one-particle contribution is described by a $2M\times 2M$ time-dependent Hermitian matrix $h_{m\sigma,m'\sigma'}(t)$ and the two-particle contribution is described by a real and symmetric function $(m,\sigma;m',\sigma')\mapsto U_{m\sigma,m'\sigma'}$. For example, for the one-dimensional Hubbard model, $h_{m\sigma,m'\sigma'}(t)$ contains a hopping part involving neighboring sites and possibly an onsite (local) time-dependent external potential,
\begin{equation}
h_{m\sigma,m'\sigma'}(t) = -\tau \delta_{m,m'\pm 1} \delta_{\sigma,\sigma'} + v_{\text{ext},m}(t) \delta_{m,m'} \delta_{\sigma,\sigma'},
\label{eq:Hubbard_h}
\end{equation}
where $\tau$ is the hopping parameter, and the two-particle interaction corresponds to an onsite interaction between different spins
\begin{equation}
U_{m\sigma, m'\sigma'}=\frac{{\cal U}}{2} \delta_{m,m'} (1-\delta_{\sigma,\sigma'})
\label{eq:Hubbard_U}
\end{equation}
where ${\cal U}$ is the on-site interaction parameter. For the moment we keep $h(t)$ and $U$ rather general and defer to \Cref{sec:Hubbard_dimer} the discussion of explicit models.

The reference time-dependent Schrödinger equation (TDSE) is
\begin{equation}
i\partial_t\Psi^{\rm S}(t)=\bH_U(t)\,\Psi^{\rm S}(t),
\label{eq:Hubbard_Schrodinger}
\end{equation}
with an initial state $\Psi^{\rm S}(0) =\Psi_0$.
If the number of particles $N$ is not too large, we can easily compute this solution numerically. This becomes extremely difficult when $N$ increases because $\bH_U(t)$ is a huge matrix of size $\binom{2M}{N} = \frac{(2M)!}{(2M-N)!N!}$ (not considering possible space and spin symmetries). The goal of TDDFT is to replace~\eqref{eq:Hubbard_Schrodinger} by a lower dimensional equation reproducing the exact Schrödinger density $\rho(t) := \rho_{\Psi^{\rm S}(t)}$, but at the expense of introducing rather complicated nonlinear terms.

\subsection{\texorpdfstring{$v$}{v}-representability and the invertibility of  the matrices \texorpdfstring{$S^\Psi$}{Spsi} and \texorpdfstring{$K^\Psi$}{Kpsi}}
\label{sec:v-rep.invert-S-K}

Recall that the geometric and variational principles respectively rely on the two matrices $S^{\Psi(t)}$ and $K^{\Psi(t)}(t)$. Namely, we get a well-defined dynamics as long as these matrices are invertible. 

In this section we consider the case when the initial condition $\Psi_0$ is such that the matrices $S^{\Psi_0}$ and $K^{\Psi_0}(t=0)$ are invertible, and give an interpretation of the invertibility in relation with the Hohenberg--Kohn theorem and the $v$-representability problem. The invertibility of the two matrices will then propagate to short times, leading to the local-in-time existence and uniqueness of solutions by Theorems~\ref{thm:TDVP_commuting} and~\ref{thm:McLachlan}.  As we only consider the initial time $t=0$, we drop the time-variable in our notation and denote the Hamiltonian simply by $\bH_U$.

The (time-independent) \textbf{$v$-representability problem} asks if for a given density $\rho=(\rho_m)_{m=1}^M$ one can find a potential $v=(v_m)_{m=1}^M$ and a state $\Psi$ solving the eigenvalue equation
$$\left(\bH_U+\sum_{m=1}^Mv_m\,\cN_m\right)\Psi=E_0\,\Psi$$
such that $\rho_\Psi=\rho$. In principle $E_0$ could be any eigenvalue of $\bH_U+\cV$ with $\cV:=\sum_{m=1}^Mv_m\,\cN_m$, but we restrict our attention to ground states. Not all densities are $v$-representable in this sense~\cite{EpsRos-76,Englisch-83,Kohn-83,Englisch-84,ChaChaRus-85,UllKoh-02,PenLeu-21,PenLeu-24,BakCsiLaePen-25}. More generally, we say that a density is \textbf{ensemble $v$-representable} if we can find a mixed state $\Gamma$ supported in the ground eigenspace of $\bH_U+\cV$, i.e. $(\bH_U+\cV-E_0) \Gamma=0$, such that $\rho_\Gamma=\rho$. If the ground state is non-degenerate (that is, the associated eigenspace has dimension one), then the two are equivalent, of course. It was shown in~\cite{ChaChaRus-85} that any density satisfying $0<\rho_m<2$ is ensemble $v$-representable but for pure states the situation is less clear. We refer to~\cite{PenLeu-21} for simple examples of densities that are not $v$-representable by a pure state.

If $\rho$ is representable by a pure ground state $\Psi$ of some potential $v$, the next question is whether this $v$ is unique modulo an additive constant, which is called the \textbf{unique $v$-representability problem}~\cite{PenLeu-21}. This problem is addressed by the Hohenberg--Kohn theorem~\cite{HohKoh-64}. The key step is the following. Assume that $\Psi$ is the ground state of two potentials $v^{(1)}$ and $v^{(2)}$ with ground-state energies $E_0^{(1)}$ and $E_0^{(2)}$:
\begin{multline*}
\left(\bH_U+\sum_{m=1}^M v^{(1)}_m\cN_m-E_0^{(1)}\right)\Psi\\=\left(\bH_U+\sum_{m=1}^M v^{(2)}_m\cN_m-E_0^{(2)}\right)\Psi=0.
\end{multline*}
Then we would like to conclude that $v^{(1)}=v^{(2)}$ modulo an additive constant. Subtracting the two equations, the unique $v$-representability property obviously follows if $\Psi$ satisfies that
\begin{equation}
\text{if $\dps \sum_{m=1}^Mv_m\,\cN_m\Psi=0$ for some $v_m$'s, then $\dps v_m\equiv 0.$}
\label{eq:Hohenberg--Kohn}
\end{equation}
We remark that this is exactly requiring that the vectors $\cN_1\Psi,...,\cN_M\Psi$ are $\R$-linearly independent, and thus that the matrix
$$(S^{\Psi})_{mn}=\Re\pscal{\cN_m\Psi,\cN_n\Psi}.$$
is invertible, as is needed in our theory. At this point, we emphasize that the property~\eqref{eq:Hohenberg--Kohn} makes sense for \textbf{any state $\Psi$}. It needs not be a ground state of anything. Only when it is a ground state, this property implies in the Hohenberg--Kohn theorem that the potential $v$ is uniquely determined (modulo an additive constant) from the density $\rho$.
 
To summarize, the geometric theory developed in Section~\ref{sec:McLachlan} relies on the invertibility of $S^{\Psi}$ which can be reformulated as in~\eqref{eq:Hohenberg--Kohn}. If it happens that $\Psi$ is the ground state of some $v$ then this implies that the latter is uniquely defined. The importance of the matrix $S^\Psi$ in DFT was already implicit in~\cite{XuMaoGaoLiu-22}, where a similar matrix called the ``generalized density correlation matrix'' was introduced. 
As was already mentioned before, the matrix $S^{\Psi}$ also appears in the recent work~\cite{PenLee-PRA-25} dealing with ground-state DFT with constrained search in imaginary time.

When $\Psi$ is in addition a \textbf{non-degenerate ground state} of $\bH_U+\cV$, the invertibility of $S^{\Psi}$ is in fact related to that of the potential-to-density map, in the neighborhood of $v$. Let us quickly explain this claim. We apply a small variation $v\to v+\delta v$ and compute the resulting variation in the density to leading order, known as the linear response function $\chi[v]$ (i.e., the derivative of the potential-to-density map). A simple calculation provides the formula
\begin{multline}
\sum_{m=1}^M\pscal{\delta v, \chi[v] \delta v}_{\R^M}\\
=-2\pscal{\sum_{m=1}^M(\delta v)_m\cN_m\Psi,(\bH_U-E_0)^{-1}_\perp\sum_{m=1}^M(\delta v)_m\cN_m\Psi}.
\nonumber
\label{eq:diff_pot_to_den}
\end{multline}
We used the convention that $\sum_m (\delta v)_m\rho_m=0$, which means that the vector $\sum_{m=1}^M(\delta v)_m\cN_m\Psi$ belong to the orthogonal to $\Psi$, so that we can invert the matrix $\bH_U-E_0$ (where $E_0$ is the ground-state energy) on that space thanks to the assumed non-degeneracy. We denoted the inverse by $(\bH_U-E_0)^{-1}_\perp$. The spectrum of $(\bH_U-E_0)^{-1}_\perp$ is included in the interval $[(E_{\mathrm{max}} - E_0)^{-1}, g^{-1}]$, where $E_{\mathrm{max}}$ is the largest eigenvalue of $\bH_U$ and $g = E_1 - E_0 > 0$ is the gap above the ground-state energy. Thus,
\[
\frac{\pscal{\delta v, S^\Psi \delta v}_{\R^M}}{E_{\mathrm{max}}-E_0}
    \leq -\pscal{\delta v, \chi[v] \delta v}_{\R^M}
    \leq \frac{\pscal{\delta v, S^\Psi  \delta v}_{\R^M}}{g}.
\]
This proves that the invertibility of $S^\Psi$ is \textbf{equivalent to that of the derivative $\chi[v]$ of the potential-to-density map}. By the implicit function theorem, the latter implies the invertibility of the full potential-to-density map in a neighborhood of $v$ (modulo an additive constant). This gives yet another interpretation of the invertibility of $S^\Psi$, for non-degenerate ground states, which is in the spirit of~\cite[Thm.~13]{PenLeu-21}.

Proving the invertibility of $S^\Psi$ is not always an easy task. We need to make sure that $\Psi(m_1\sigma_1,...,m_N\sigma_N)$ is non-zero for sufficiently many $m_1,...,m_N\in\{1,...,M\}$. For any such indices we obtain $v_{m_1}+\cdots +v_{m_N}=0$ and if we have at least $M$ independent such conditions we obtain $v\equiv0$ as desired. It is clear that the set of all $\Psi$'s satisfying~\eqref{eq:Hohenberg--Kohn} is open and dense (because we can perturb any $\Psi$ so as to make all the coefficients non zero). But if we restrict our attention to ground states, the situation is more complicated. For continuous systems, the corresponding condition~\eqref{eq:Hohenberg--Kohn} follows from the unique continuation property, which is rather delicate to establish~\cite{Garrigue-20,LewLieSei-23_DFT}. In the discrete case considered here, the property~\eqref{eq:Hohenberg--Kohn} is not always true~\cite{SchGod-95,RosVerAlm-18,PenLeu-21}. If $N=1$ then the relation~\eqref{eq:Hohenberg--Kohn} is clearly equivalent to $(\rho_\Psi)_m>0$ for all sites. For $N\geq2$ the strict positivity of $\rho_\Psi$ is definitely necessary but not sufficient. For instance, if we place $N=4$ particles on $M=2$ sites then we have only one possible $\Psi$, describing two electrons of opposite spins per site. But then the relation~\eqref{eq:Hohenberg--Kohn} gives $2v_1+2v_2=0$ and thus only $v_1=-v_2$. The examples provided in~\cite{PenLeu-21} suggest that~\eqref{eq:Hohenberg--Kohn} is a quite generic property, however, hence natural to assume for ground states. More about the invertibility of $S^\Psi$ can be read in Appendix~\ref{sec:independence}.

\medskip
\paragraph*{Relation to the matrix \texorpdfstring{$K^\Psi$}{Kpsi}.}
Next we discuss the other matrix important for our abstract theory, which is
$$K^{\Psi}_{mn}=\frac12\pscal{\Psi,[\cN_n,[\bH_U,\cN_m]]\Psi}.$$
Here we can replace $\bH_U$ by $\bH_0$ since $\cN_m$ commutes with the interaction and then express this matrix solely using the one-particle density matrix 
$$
(\gamma_\Psi)_{m\sigma ,m'\sigma'}:=\big\langle\Psi,a_{m\sigma}^\dagger a_{m'\sigma'}\Psi\big\rangle
$$
as
\begin{equation}
K^{\Psi}_{mn}=\frac12\tr\big([\delta_n,[h,\delta_m]]\gamma_\Psi\big)
\label{eq:Hubbard_K_gamma}
\end{equation}
where $h$ is the one-particle Hamiltonian matrix and $\delta_m$ is the one-particle projection matrix onto the site $m$ (corresponding to the operator $\cN_m$ in second quantization), i.e. $(\delta_m\phi)_{n\sigma}=\phi_{n\sigma}\delta_{n,m}$.

As defined in \eqref{eq:Hubbard_K_gamma}, the matrix $K^{\Psi}$ is never invertible. Indeed, for the vector $ \mathbbm{1} = (1,1,...,1)$ we have 
\begin{align}
    (K^{\Psi} \mathbbm{1})_m 
    & = \sum_n \frac{1}{2} \tr\big([\delta_n,[h,\delta_m]]\gamma_\Psi\big)
    \nonumber
    \\ & =   \frac{1}{2} \tr\big([I_{2M},[h,\delta_m]]\gamma_\Psi\big) = 0, 
    \label{eqn.K1=0}
\end{align}
where $I_{2M}$ is the identity matrix.
This is the statement that the constant potential $v\equiv 1$ is in the kernel of $K^{\Psi}$. This is the same trivial degeneracy discussed above, and we should really consider whether $K^{\Psi}$ is invertible on the space of potentials orthogonal to the constant potentials. For this to be true it is important that $h$ contains off-diagonal terms because the diagonal ones give vanishing contributions to \eqref{eq:Hubbard_K_gamma}. 
More precisely, a necessary condition is that the particles are able to hop anywhere, meaning that the graph defined by the non-zero (off-diagonal) matrix elements of $h$ is connected. 
Indeed, if this were not the case, there would be some proper subset $J\subset \{1,\ldots,M\}$ which would be invariant under the action of $h$ (i.e. $h$ commutes with the projection ${\rm diag}(\1_J)$ on vectors supported on the sites in $J$). Then, by a similar argument as in \eqref{eqn.K1=0} above, we would have $K^{\Psi}\mathbbm{1}_J = 0$, contradicting the fact that $K^{\Psi}$ is invertible.  

We have seen in Section~\ref{sec:TDVP_commuting} that when $K^{\Psi}$ is invertible, $S^\Psi$ must also be invertible. It turns out that the converse is also true for non-degenerate ground states. This important fact is stated in the following theorem.

\begin{theorem}[Invertibility of the matrix $K^\psi$]\label{thm:GS_K}
Assume that $\Psi$ is a non-degenerate ground state of $\bH_U+\sum_{m=1}^Mv_m\cN_m$ and that $S^{\Psi}$ is invertible. Then the matrix $K^{\Psi}$ is invertible on the orthogonal of the constant potential $v\equiv1$.
\end{theorem}

The proof of this result is given in Appendix~\ref{app:proof_matrix_K}.

The conclusion of this section is that if we start the dynamics with a state $\Psi(0)$ which is a non-degenerate ground state of $\bH_U(0)$ satisfying the unique $v$-representability condition~\eqref{eq:Hohenberg--Kohn}, then we will be able to use the theorems from Sections~\ref{sec:TDVP} and~\ref{sec:McLachlan} for both the variational and geometric principles.

\subsection{Variational principle}
The variational principle is the standard approach for TDDFT~\cite{GroDobPet-96,MarUllNogRubBurGro-06,Ullrichs-11,MarMaiNogGroRub-12}. It requires finding a potential $v(t)$ so that the solution to the associated equation
\begin{equation}
i\partial_t\Psi^{\rm V}(t)=\left(\bH_U(t)+\sum_{m=1}^Mv_m(t)\cN_m\right)\,\Psi^{\rm V}(t)
\nonumber
\end{equation}
has the desired density, i.e. $\rho_{\Psi^{\rm V}(t)} = \rho(t)$. Written in this form, the potential $v$ is only defined up to an additive time-dependent constant that generates a harmless but arbitrary time-dependent global phase in $\Psi^{\rm V}$. We can remove this gauge freedom by assuming for instance $v_M\equiv0$ (this corresponds to erasing $\cN_M$ from the list of fixed observables), or by imposing a sum rule such as $\sum_{m=1}^M v_m(t)\rho_m(t)=0$.

Of course, not every $\rho(t)$ will be representable by a time-dependent potential $v(t)$~\cite{Baer-08,LiUll-08}. The condition~\eqref{eq:condition_supp_TDVP} on the first-order derivative of the density reads in our case
\begin{equation}
\rho_m'(t)=\pscal{\Psi^{\rm V}(t),i[\bH_0(t),\cN_m]\Psi^{\rm V}(t)}
\label{eq:condition_supp_TDVP_Hubbard}
\end{equation}
where the two-particle interaction and the potential term $\sum_{m=1}^Mv_m(t)\cN_m$ do not appear since they commutes with the observables $\cN_m$. Because $\bH_0(t)$ and $\cN_m$ are one-body operators, the commutator is again a one-body operator, we can rewrite~\eqref{eq:condition_supp_TDVP_Hubbard} in the form
\begin{equation}
\rho_m'(t)=i\tr\big([h(t),\delta_m]\gamma_{\Psi^{\rm V}(t)}\big).
\label{eq:condition_supp_TDVP_Hubbard_bis}
\end{equation}
In the right-hand side of~\eqref{eq:condition_supp_TDVP_Hubbard_bis} we can remove the site-diagonal parts $h_{m,\sigma,m\sigma'}(t)$ of the matrix $h(t)$ because those commute with $\delta_m$. Only the off-diagonal coefficients of $h(t)$ matter. They describe the possibility that particles can hop between sites. The relation~\eqref{eq:condition_supp_TDVP_Hubbard_bis} looks like the continuity equation for continuous system, with the right-hand side involving the discrete analogue of the negative of the divergence of the current density. This equation sets some bounds on how fast $\rho(t)$ can vary in time. For instance, assuming that the spectrum of $h(t)$ is included in the interval $[-C,C]$ over the considered interval of times (i.e., $h(t)$ is bounded by $C$ in the matrix supremum norm), one can see that
$$|\rho'_m(t)|\leq 2C\sqrt{N\rho_m(t)}.$$

It TDDFT the dynamics is often started with a ground state $\Psi_0$ of $\bH_U(0)$. If the latter is non-degenerate and satisfies the unique $v$-representability condition~\eqref{eq:Hohenberg--Kohn}, we have explained in the previous subsection that the matrix $K^{\Psi_0}$ is invertible in the orthogonal to constant potentials. We can thus apply Theorem~\ref{thm:TDVP_commuting} and we automatically get a unique solution for some short time, hence we have a uniquely $v$-representable time-dependent density. 

Along the dynamics, the exact potential $v(t)$ giving the reference density $\rho(t)$ can be obtained from the second-order derivative $\rho''(t)$ of the reference density using the van Leeuwen equation~\eqref{eq:vanLeeuwen_finite_dim}
\begin{align}    
& 2\sum_{n=1}^M  K_{mn}^{\Psi^{\rm V}\!(t)}(t) \, v_n(t)
\nonumber
\\ & \qquad =
-\rho_m''(t)
-\pscal{\Psi^{\rm V}(t),[\bH_U(t),[\bH_0(t),\cN_m]]\Psi^{\rm V}(t)}
\nn 
\\ & \qquad\qquad +\pscal{\Psi^{\rm V}(t),i[\bH_0'(t),\cN_m]\Psi^{\rm V}(t)},
\label{eq:vanLeeuwen_Hubbard}
\end{align}
with
$$K^{\Psi}_{mn}(t)=\frac12\pscal{\Psi,[\cN_n,[\bH_0(t),\cN_m]]\Psi}.$$
Note that in the most common case where the one-body Hamiltonian $\bH_0(t)$ is the sum of the kinetic-energy operator $T$ and a time-dependent local external potential $V_\text{ext}(t)$, the commutator $[\bH_0(t),\cN_m]$ simplifies to $[T,\cN_m]$.

\subsection{Geometric principle}
The geometric principle provides the equation
\begin{equation}
i\partial_t\Psi^{\rm G}(t)=\left(\bH_U(t)+i\sum_{m=1}^Mw_m(t)\cN_m\right)\,\Psi^{\rm G}(t)
\nonumber
\end{equation}
with a purely imaginary local potential which is such that we obtain the desired density, i.e. $\rho_{\Psi^{\rm G}(t)}=\rho(t)$. We recall that, in fact, the correction term $i\sum_{m=1}^Mw_m(t)\cN_m$ can also be rewritten as a rank-two Hermitian perturbation as in~\eqref{eq:McLachlan_rank2}. According to Theorem~\ref{thm:McLachlan}, the only assumption needed to obtain a solution $w_m(t)$ for a short time is the invertibility of $S^{\Psi_0}$, i.e. the $v$-representability condition~\eqref{eq:Hohenberg--Kohn}. No further assumption is needed.

With the geometric principle, we expect to be able to reproduce densities $\rho(t)$ that are not accessible with the variational principle, in particular when $\rho(t)$ varies too fast. 

Along the dynamics, the exact potential $w(t)$ giving the reference density $\rho(t)$ can be obtained from the first-order derivative $\rho'(t)$ of the reference density using the modified continuity equation~\eqref{eq:diff_constraint_McLachlan}
\begin{align}
2\sum_{n=1}^M S_{mn}^{\Psi^{\rm G}\!(t)} w_n(t) = \rho_m'(t)-\pscal{\Psi^{\rm G}(t),i[\bH_0(t),\cN_m]\Psi^{\rm G}(t)},
\label{eq:diff_constraint_McLachlan_Hubbard}
\end{align}
which is much simpler than~\eqref{eq:vanLeeuwen_Hubbard}.

\subsection{Time-dependent geometric Kohn--Sham theory}
In time-dependent Kohn--Sham theory we want to reproduce the exact interacting time-dependent density using an auxilliary system of fictitious non-interacting particles. We describe here the form of the corresponding non-interacting equations.

We start by solving the exact time-dependent Schrödinger equation~\eqref{eq:Hubbard_Schrodinger} for some initial many-body state $\Psi_0$ and obtain the exact Schrödinger density $\rho(t)=\rho_{\Psi^{\rm S}(t)}$. We then ask whether the latter can be reproduced with a non-interacting system based on $\bH_0(t)$ (with the interaction $U$ removed). The variational principle is the standard technique employed in the literature and, to our knowledge, the geometric principle was never considered. We will thus obtain a  completely new scheme that we call the \textbf{geometric Kohn--Sham equation}.

Let us consider an initial state $\Phi_0$ such that $\rho_{\Phi_0}=\rho(0)$. We choose for $\Phi_0$ a Slater determinant
$$\Phi_0(m_1\sigma_1,...,m_N\sigma_N)=(N!)^{-1/2}\det(\phi_j(0)_{m_k\sigma_k})$$
where the orbitals $\phi_1(0),...,\phi_N(0)\in\C^{2M}$ are assumed to form an orthonormal set in $\C^{2M}$. Often, $\Psi_0$ is chosen to be a ground state of $\bH_U(0)$ and $\Phi_0$ the solution to the ground-state Kohn--Sham problem.
But for the moment $\Phi_0$ can be rather arbitrary. The only condition needed to be able to go on is the invertibility of $S^{\Phi_0}$ for the geometric principle or the invertibility of the matrix $K^{\Phi_0}(0)$ for the variational principle. 

We will now consider different variants of the time-dependent Kohn--Sham scheme. For each of them, the time-dependent state is a Slater determinant of the form
$$\Phi^X(t, m_1\sigma_1,...,m_N\sigma_N)=(N!)^{-1/2}\det(\phi_j^X(t)_{m_k\sigma_k}),$$
where $\phi_1^X(t),...,\phi_N^X(t)$ are time-dependent orbitals and $X$ designates the variant considered.

\medskip

\noindent \textit{Standard Kohn--Sham scheme (TDKS).}
For the standard Kohn--Sham scheme based on the  variational principle, the orbitals solve the \textbf{time-dependent Kohn--Sham (TDKS) equations}
\begin{equation}
\boxed{i\partial_t\phi_j^{\rm KS}(t)=\big(h(t)+v_\text{Hxc}(t)\big)\phi^{\rm KS}_j(t)}
\label{eq:KS_TDVP}
\end{equation}
with the so-called Hartree-exchange-correlation (Hxc) potential, represented here by a diagonal matrix $(v_\text{Hxc}(t))_{m\sigma,m'\sigma'} = v_{\text{Hxc},m}(t) \delta_{m,m'} \delta_{\sigma,\sigma'}$, 
to be found so that
$$\sum_{j=1}^N\sum_{\sigma\in\{\uparrow,\downarrow\}}|\phi^{\rm KS}_j(t)_{m\sigma}|^2=\rho_m(t).$$
The potential $v_\text{Hxc}(t)$ will be a very complicated nonlinear function of $\rho(t)$, $\Phi_0$ and $h(t)$ (hence indirectly of $\Psi_0$, $\bH_U(t)$, $\Phi_0$ and $h(t)$). By Theorem~\ref{thm:TDVP_commuting}, we know that there is a unique potential $v_\text{Hxc}(t)$ whenever the matrix $K^{\Phi^{\rm KS}(t)}(t)$ stays invertible, on the orthogonal to the constant potentials.
In this case, the exact potential $v_\text{Hxc}(t)$ can be determined from~\eqref{eq:vanLeeuwen_Hubbard} with the interaction $U$ removed
\begin{align}
& 2\sum_{n=1}^M  K_{mn}^{\Phi^{\rm KS}\!(t)}(t) \, v_{\text{Hxc},n}(t) 
\nonumber 
\\ & \qquad = -\rho_m''(t)
-\pscal{\Phi^{\rm KS}(t),[\bH_0(t),[\bH_0(t),\cN_m]]\Phi^{\rm KS}(t)}
\nonumber
\\ & \qquad\qquad
+\pscal{\Phi^{\rm KS}(t),i[\bH_0'(t),\cN_m]\Phi^{\rm KS}(t)},
\label{eq:vanLeeuwen_HubbardKS}
\end{align}
where $\Phi^{\rm KS}(t)$ is the Kohn--Sham Slater determinant.

\medskip

\noindent \textit{Geometric Kohn--Sham scheme (TDGKS).}
For the geometric principle, things become a little more complicated. We get orbitals $\tilde\phi_j^{\rm GKS}(t)$ evolving according to the equation
\begin{equation}
i\partial_t\tilde \phi_j^{\rm GKS}(t)=\big(h(t)+iw(t)\big)\tilde \phi^{\rm GKS}_j(t)
\label{eq:KS_McLachlan_not_good}
\end{equation}
with the geometric potential, represented here by a diagonal matrix $(w(t))_{m\sigma,m'\sigma'} = w_{m}(t) \delta_{m,m'} \delta_{\sigma,\sigma'}$, but these orbitals will in general not be orthonormal for $t>0$. Indeed, although the potential $w(t)$ is chosen so that the corresponding Slater determinant $\Phi^{\rm GKS}(t)$ stays normalized, hence its dynamics is Hermitian, the individual orbitals of~\eqref{eq:KS_McLachlan_not_good} are evolving with a non-Hermitian dynamics. It is thus better to change our gauge and use a different set of orbitals staying orthonormal, leading to the same Slater determinant. This can be achieved by adding Lagrange multipliers in the form
\begin{equation}
i\partial_t\phi_j^{\rm GKS}(t)=\big(h(t)+iw(t)\big)\phi^{\rm GKS}_j(t)+\sum_{k=1}^N\Lambda_{jk}(t)\phi^{\rm GKS}_k(t).
\nonumber
\end{equation}
Writing the desired condition that $\partial_t\pscal{\phi_k^{\rm GKS},\phi_j^{\rm GKS}}=0$, we find that $\Lambda_{jk}$ must be the anti-Hermitian matrix
$$\Lambda_{jk}:=-i\pscal{\phi_k^{\rm GKS},w\phi_j^{\rm GKS}}.$$
The \textbf{time-dependent geometric Kohn--Sham (TDGKS) equation} therefore reads
\begin{multline}
i\partial_t\phi_j^{\rm GKS}(t)=\big(h(t)+iw(t)\big)\phi^{\rm GKS}_j(t)\\
-i\sum_{k=1}^N\pscal{\phi_k^{\rm GKS}(t),w(t)\phi_j^{\rm GKS}(t)}\phi^{\rm GKS}_k(t).
\label{eq:KS_McLachlan}
\end{multline}
In this way of writing, the potential $w(t)$ is defined only up to an additive constant, similarly as in the variational principle.

Let us now rewrite the last term of~\eqref{eq:KS_McLachlan} using the one-particle density matrix $\gamma^{\rm GKS}(t)=\sum_{k=1}^N|\phi^{\rm GKS}_k(t)\rangle\langle\phi^{\rm GKS}_k(t)|$ as
\begin{multline*}
\sum_{k=1}^N\pscal{\phi_k^{\rm GKS}(t),w(t)\phi_j^{\rm GKS}(t)}\phi^{\rm GKS}_k(t)
\\ =\gamma^{\rm GKS}(t)w(t)\phi_j^{\rm GKS}(t).
\end{multline*}
We can therefore rewrite the geometric Kohn--Sham equation of~\eqref{eq:KS_McLachlan} in the somewhat condensed Hermitian form
\begin{equation}
\boxed{
i\partial_t\phi_j^{\rm GKS}(t)=\Big(h(t)\\+i\big[w(t),\gamma^{\rm GKS}(t)\big]\Big)\phi^{\rm GKS}_j(t)
}
\label{eq:KS_McLachlan_commutator}
\end{equation}
or equivalently, in the von Neumann form
\begin{equation}
i\partial_t\gamma^{\rm GKS}(t)=\Big[h(t)+i\big[w(t),\gamma^{\rm GKS}(t)\big],\gamma^{\rm GKS}(t)\Big].
\nonumber
\end{equation}
The geometric term coming from the density constraint therefore takes the form of a commutator at the level of the orbitals, as was already the case in~\eqref{eq:McLachlan_rank2}. This is a nonlocal Hermitian operator of rank $\leq2N$. The geometric term can be interpreted as a kind of exchange term because its matrix elements are
\begin{multline*}
\left( i\big[w(t),\gamma^{\rm GKS}(t)\big]\right)_{m\sigma,m'\sigma'} 
\\
= i\big(w_m(t)-w_{m'}(t)\big)\gamma^{\rm GKS}_{m\sigma,m'\sigma'}(t).
\end{multline*}
Kohn--Sham models with such a nonlocal term of geometric origin have never been considered, to our knowledge. The existence and uniqueness of the potential $w(t)$ is provided by Theorem~\ref{thm:McLachlan}, under the sole assumption that $S^{\Phi^{\rm GKS}(t)}$ stays invertible (see Appendix~\ref{sec:independence} about the invertibility of $S^{\Phi}$).
In this case, the exact geometric potential $w(t)$ can be obtained from the modified continuity equation~\eqref{eq:diff_constraint_McLachlan_Hubbard}
\begin{multline}
2\sum_{n=1}^M S_{mn}^{\Phi^{\rm GKS}\!(t)} w_n(t) 
\\
 = \rho_m'(t)-\pscal{\Phi^{\rm GKS}(t),i[\bH_0(t),\cN_m]\Phi^{\rm GKS}(t)}.
\label{eq:modconteqGKS}
\end{multline}

\medskip

\paragraph*{Kohn--Sham scheme with geometric correction (TDKS+G).}
The standard Kohn--Sham equation~\eqref{eq:KS_TDVP} and the new geometric Kohn--Sham equation~\eqref{eq:KS_McLachlan_commutator} are both able to reproduce the exact density $\rho(t)$, at least for a short time and under suitable assumptions on the initial state $\Phi_0$. In practice, we however have no access to the exact density $\rho(t)$ and need to find useful approximations of these exact potentials $v_\text{Hxc}$ or $w$. A possible approximation of $v_\text{Hxc}$ for the Hubbard model is the adiabatic local-density approximation (ALDA)~\cite{SchGunNoa-95,LimSilOliCap-PRL-03}. The latter is usually quite accurate in a nearly adiabatic setting but may completely fail in situations that are really out-of-equilibrium. It is natural to ask whether it is possible to correct the error of a given approximate Kohn--Sham model by means of a geometric term as in~\eqref{eq:KS_McLachlan_commutator}. The answer is positive.

We give ourselves any (adiabatic) approximation $v^{\rm app}_\text{Hxc}[\rho(t)]$ of the exact potential $v_\text{Hxc}$ and consider the \textbf{time-dependent Kohn--Sham equation with geometric correction (TDKS+G)}
\begin{equation}
\boxed{
\begin{aligned}
i\partial_t\phi_j^{\rm KS+G}(t)=&\Big(h(t)+v^{\rm app}_\text{Hxc}[\rho_{\gamma^{\rm KS+G}}(t)]\\
&\quad+i\big[w^{\rm corr}(t),\gamma^{\rm KS+G}(t)\big]\Big)\phi^{\rm KS+G}_j(t).
\end{aligned}
}
\label{eq:KS+G}
\end{equation}
The existence of a unique corrective geometric potential $w^{\rm corr}(t)$ can be established by following the exact same lines as in the proof of Theorem~\ref{thm:McLachlan}, under the sole assumption that $S^{\Phi^{\rm KS+G}(t)}$ stays invertible and $\rho\mapsto v^{\rm app}_\text{Hxc}[\rho]$ is continuous. 
In this case, the exact corrective geometric potential $w^{\rm corr}(t)$ can still be obtained from the modified continuity equation~\eqref{eq:diff_constraint_McLachlan_Hubbard}
\begin{multline}   
2\sum_{n=1}^M S_{mn}^{\Phi^{\rm KS+G}\!(t)} w_n^{\rm corr}(t) 
\\ = 
 \rho_m'(t)-\pscal{\Phi^{\rm KS+G}(t),i[\bH_0(t),\cN_m]\Phi^{\rm KS+G}(t)},
\label{eq:modconteqKS+G}
\end{multline}
since the local potential $v^{\rm app}_\text{Hxc}$ does not contribute to the commutator.

Of course, we can also formulate variants of the Kohn--Sham scheme using the oblique principle from Section~\ref{sec:oblique}. We do not write them explicitly for shortness. In \Cref{sec:Hubbard_dimer}, we investigate the form of the exact potentials $v_\text{Hxc}(t)$, $w(t)$ and $w^{\rm corr}(t)$ in several practical situations of interest.

\subsection{Time-dependent current-density functional theory}
\label{sec.current.DFT}
We saw in \Cref{sec:TDVP_commuting} above, that for commuting observables, as is the case for TDDFT, the symplectic formulation of the variational principle (\Cref{thm:TDVP_symplectic}) was of no use. 
A natural related framework, where it is of use, however, is the case of time-dependent current-density functional theory (TDCDFT) as we explain in this section.

We consider a one-dimensional chain and,
in addition to the density observables $\cN_m$,
also the current-density operators
\begin{equation}
    \cJ_m = i \tau \sum_{\sigma \in \{ \uparrow,\downarrow\} } 
    \left(a_{m\sigma}^\dagger a_{(m+1)\sigma} - a_{(m+1)\sigma}^\dagger a_{m\sigma}\right), 
\nonumber
\end{equation}
where $\tau$ is the hopping parameter.
Note that we label the current operator on the edge $m,m+1$ simply by the index $m$.
As discussed above, we can remove $\cN_M$ from the list of observables by choosing the gauge of the added potential to have $v_M = 0$. 
Similarly, we can remove $\cJ_M$ since the total current is fixed by the time-dependent density. (Of course, $\cJ_M=0$, unless we consider periodic boundary conditions, where we identify $M+1$ with $1$.) Thus, the list of observables to constrain is $\{\cN_m, \cJ_m\}_{m=1}^{M-1}$. 
For the variational principle, the question is thus to find real numbers $v_m(t)$ and $\alpha_m(t)$ such that the state $\Psi^{\mathrm{V}}(t)$ solving
\begin{equation}\label{eq.current.DFT.VP}
i\partial_t\Psi^{\mathrm{V}}(t)
\!=\!\!\left[\bH_U(t)+\sum_{m=1}^{M-1} \Big(v_m(t)\cN_m + \alpha_m(t) \cJ_m\Big)\right] \!
\Psi^{\mathrm{V}}(t)
\end{equation}
with initial condition $\Psi^{\mathrm{V}}(0) = \Psi_0$, 
has both density $\rho(t)$ and current density $j(t)$ as prescribed. 
To this end, we apply \Cref{thm:TDVP_symplectic}. 
The matrix $A^{\Psi_0}$ (defined in \eqref{eq:A_symplectic}) has the block form (since the $\cN_m$'s commute)
\begin{equation}
A^{\Psi_0} = \begin{pmatrix}
    0 & B^{\Psi_0} \\ -(B^{\Psi_0})^T & C^{\Psi_0}
\end{pmatrix},
\nonumber
\end{equation}
where 
\begin{align*}
    B^{\Psi_0}_{mn} &= \pscal{\Psi_0, \frac{i[\cJ_n, \cN_m]}{2}\Psi_0}, \\
    C^{\Psi_0}_{mn} &= \pscal{\Psi_0, \frac{i[\cJ_n, \cJ_m]}{2}\Psi_0}. 
\end{align*}
It follows that $A^{\Psi_0}$ is invertible if and only if $B^{\Psi_0}$ is. To see if $B^{\Psi_0}$ is invertible, we 
calculate 
\begin{multline*}
    i[\cJ_n,\cN_m]
    = \tau \sum_{\sigma \in \{ \uparrow,\downarrow\}} 
    \Bigl[
    \delta_{n,m} (a_{(n+1)\sigma}^\dagger a_{n\sigma} + a_{n\sigma}^\dagger a_{(n+1)\sigma})
\\ 
    - \delta_{n+1,m} (a_{(n+1)\sigma}^\dagger a_{n\sigma} + a_{n\sigma}^\dagger a_{(n+1)\sigma})
    \Bigr]. 
\end{multline*} 
Hence, the matrix $B^{\Psi_0}$ is lower triangular, and its determinant is just the product of its diagonal elements
\begin{equation}
    \det(B^{\Psi_0}) = \prod_{m=1}^{M-1} \left[\tau \sum_{\sigma \in \{\uparrow, \downarrow\} } 
    2\Re \pscal{\Psi_0, a_{(m+1)\sigma}^\dagger a_{m\sigma} \Psi_0}\right]. 
\nonumber
\end{equation}
Generically, this determinant is non-zero. If so, then \Cref{thm:TDVP_symplectic}
ensures that there exists unique $v_m(t)$'s and $\alpha_m(t)$'s such that the solution to \eqref{eq.current.DFT.VP} has the desired density and current density, at least for short times.

For a higher-dimensional lattice one can perform a similar analysis, with the difference that the current operators must now carry indices of the edges rather than the sites. The matrix $A^{\Psi_0}$ has the same block form, but the matrix $B^{\Psi_0}$ is no longer a square matrix and determining when the matrix $A^{\Psi_0}$ is invertible is thus slightly more complicated.

Finally, we remark that TDCDFT can also be treated with the geometric principle. In general, the matrix $S^{\Psi_0}$ is invertible. (Here one should not remove the observables $\cN_M, \cJ_M$.)
If this is the case, then \Cref{thm:McLachlan} dictates, that there exists unique real $w_m(t)$'s and $\beta_m(t)$'s such that the state
$\Psi^{\mathrm{G}}(t)$ solving
\begin{equation}
    i\partial_t \Psi^{\mathrm{G}}(t) \!=\!\! \left[H_U(t) + i \sum_{m=1}^{M} \Big(w_m(t) \cN_m + \beta_m(t) \cJ_m\Big)
    \right] \!
    \Psi^{\mathrm{G}}(t)
\nonumber
\end{equation}
with initial condition $\Psi^{\mathrm{G}}(0) = \Psi_0$, has the desired density and current density, at least for short times.

\section{Application to the Hubbard dimer}\label{sec:Hubbard_dimer}

The Hubbard dimer has often been used for testing many-body theories, in particular in the context of DFT and TDDFT (see, e.g., Refs.~\cite{AryGun-PRB-02,ReqPan-PRB-08,ReqPan-PRA-10,FarTok-12,FukFarTokAppKurRub-PRA-13,FukMai-PRA-14,FukMai-PCCP-14,CarFerSmiBur-JPCM-15,CarFerMaiBur-EPJB-18}). It corresponds to $N=2$ electrons hopping on $M=2$ sites with Hamiltonian of the form given in Eqs.~\eqref{eq:Hubbard_H},~\eqref{eq:Hubbard_h} and~\eqref{eq:Hubbard_U}. 

\subsection{Description of the model}
We will work in the spin-singlet subspace of the Hilbert space which is spanned by the three states $\ket{\Psi_1} = a_{1\uparrow}^\dagger a_{1\downarrow}^\dagger \ket{\text{vac}}$, $\ket{\Psi_2} = 2^{-1/2}(a_{1\uparrow}^\dagger a_{2\downarrow}^\dagger - a_{1\downarrow}^\dagger a_{2\uparrow}^\dagger)  \ket{\text{vac}}$ and $\ket{\Psi_3} = a_{2\uparrow}^\dagger a_{2\downarrow}^\dagger \ket{\text{vac}}$, where $\ket{\text{vac}}$ is the vacuum state of second quantization. In this basis, the Hamiltonian is the $3\times3$ matrix
\begin{align*}
\bH_U(t) = \phantom{xxxxxxxxxxxxxxxxxxxxxxxxxxxxxxxxxx}
\nonumber\\
\begin{pmatrix}
 {\cal U} + 2v_{\text{ext},1}(t) & - \sqrt{2}\,\tau & 0 \\
 - \sqrt{2}\,\tau & v_{\text{ext},1}(t)+v_{\text{ext},2}(t) & - \sqrt{2}\,\tau \\
 0 & - \sqrt{2}\,\tau & {\cal U} + 2v_{\text{ext},2}(t)
\end{pmatrix},
\end{align*}
where $\tau$ is the hopping parameter, ${\cal U}$ is the on-site interaction parameter, and $v_{\text{ext},m}(t)$ is the external potential on site $m$. Since any potential is defined up to an arbitrary additive time-dependent constant, all the results depend in fact only the potential difference $\Delta v_{\text{ext}}(t) = v_{\text{ext},1}(t) - v_{\text{ext},2}(t)$. The latter is chosen as the sum of a static part and a time-dependent perturbation
\begin{align*}
\Delta v_{\text{ext}}(t) = \Delta v_{\text{ext}}^0 + \Delta v^{\text{p}}_{\text{ext}}(t),
\end{align*}
with
\begin{align*}
\Delta v^{\text{p}}_{\text{ext}}(t) =  \cE_0 \sin(\omega t),
\end{align*}
which corresponds to the interaction between the dipole moment (chosen as $d_1 = -1/2$ and $d_2 = 1/2$ on the two sites) and a monochromatic electric field with amplitude $\cE_0$ and driving frequency $\omega$. 

Starting from the two-electron ground state $\Psi_0$ at $t=0$, we numerically solve the reference time-dependent Schrödinger equation~\eqref{eq:Hubbard_Schrodinger}, and thus obtain the reference density $\rho(t)$ and its first- and second-order derivatives $\rho'(t)$ and  $\rho''(t)$.

\subsection{Time-dependent Kohn--Sham schemes}
We describe now the different Kohn--Sham schemes that we explore on the Hubbard dimer.

\paragraph*{Standard Kohn--Sham scheme.}
We solve the TDKS equation~\eqref{eq:KS_TDVP} with one doubly-occupied spatial orbital $\varphi^\text{KS}(t)=(\varphi_1^\text{KS}(t),\varphi_2^\text{KS}(t)) \in \mathbb{C}^2$
\begin{align}
i \partial_t \varphi^\text{KS}(t) = \Big( h(t) + v_\text{Hxc}(t) \Big) \varphi^\text{KS}(t),
\label{eq:KSdimer}
\end{align}
starting from the exact ground-state Kohn--Sham orbital $\varphi_0 = \sqrt{\rho(0)/2}$. In~\eqref{eq:KSdimer}, $h(t)$ is now the spin-independent one-particle Hamiltonian
\begin{align*}
h(t) =
\begin{pmatrix}
v_{\text{ext},1}(t) & - \tau\\
- \tau & v_{\text{ext},2}(t)
\end{pmatrix},
\end{align*}
and $v_\text{Hxc}(t)$ is the diagonal matrix containing the Hxc potential
\begin{align*}
 v_\text{Hxc}(t) =
\begin{pmatrix}
v_{\text{Hxc},1}(t) & 0\\
0 &v_{\text{Hxc},2}(t)
\end{pmatrix}.
\end{align*}
We calculate the exact potential $v_\text{Hxc}(t)$ giving the exact density $\rho(t)$ using the van Leeuwen equation~\eqref{eq:vanLeeuwen_HubbardKS}. It can be checked that \eqref{eq:vanLeeuwen_HubbardKS} leads to an explicit expression of the potential difference $\Delta v_{\text{Hxc}}(t) = v_{\text{Hxc},1}(t) - v_{\text{Hxc},2}(t)$ in terms of the density difference $\Delta \rho(t) = \rho_1(t) - \rho_2(t)$ and its first- and second-order derivatives $\Delta \rho'(t)$ and $\Delta \rho''(t)$, provided that $|\Delta \rho'(t)| < 2 \tau \sqrt{4 - (\Delta \rho(t))^2}$,~\cite{FarTok-12,FukFarTokAppKurRub-PRA-13,FukMai-PRA-14} 
\begin{align}
\Delta v_\text{Hxc}(t) =& -\sigma_0 \frac{\Delta \rho''(t) + 4\tau^2 \Delta \rho(t)}{\sqrt{4 \tau^2 (4 - (\Delta \rho(t))^2) - (\Delta \rho'(t))^2}} 
\nonumber\\
&- \Delta v_\text{ext}(t),
\label{eq:DeltavHxc}
\end{align}
where $\sigma_0 = \sign(\pi/2 - |\beta(0)|)$ and $\beta(t) = \arg(\varphi_2^\text{KS}(t)\overline{\varphi_1^\text{KS}(t)})$. The presence of the quantity $\sigma_0$ in~\eqref{eq:DeltavHxc} is a manifestation of the initial-state dependence in TDDFT, as explained in~\cite{FarTok-12,FukFarTokAppKurRub-PRA-13}. In all our simulations, we have $\sigma_0=1$, as we start from the Kohn--Sham ground state for which $\beta(0)=0$. Along our TDKS dynamics we have $|\beta(t)| < \pi/2$, implying that the denominator in~\eqref{eq:DeltavHxc} never reaches zero and the potential $\Delta v_\text{Hxc}(t)$ remains well-defined.

\paragraph*{Exact adiabatic Kohn--Sham scheme.}
To investigate the impact of the adiabatic approximation, we also solve the time-dependent exact adiabatic Kohn--Sham (TDeaKS) equation~\cite{ThiGroKum-PRL-08,FukFarTokAppKurRub-PRA-13,FukMai-PRA-14}
\begin{align*}
i \partial_t \varphi^\text{eaKS}(t) = \Big( h(t) + v_\text{Hxc}^\text{ea,sc}(t) \Big) \varphi^\text{eaKS}(t),
\end{align*}
with the self-consistent exact adiabatic Hxc potential
\begin{align*}
v_\text{Hxc}^\text{ea,sc}(t) = v_\text{Hxc}^\text{gs}[\rho^\text{eaKS}(t)],
\end{align*}
calculated at the self-consistent density $\rho^\text{eaKS}_m(t)=2|\varphi_m^\text{eaKS}(t)|^2$. Here $v_\text{Hxc}^\text{gs}[\rho]$ is the exact ground-state Hxc potential functional
\begin{align*}
v_{\text{Hxc},m}^\text{gs}[\rho] = \frac{\partial E_\text{Hxc}[\rho]}{\partial \rho_m},
\end{align*}
and $E_\text{Hxc}[\rho]$ is the exact ground-state Hxc energy functional, obtained by numerically inverting the ground-state Kohn--Sham problem for any density $\rho$, as in~\cite{CarFerSmiBur-JPCM-15}. We also define the exact adiabatic Hxc potential as the exact ground-state Hxc potential functional evaluated at the exact density $\rho(t)$
\begin{align}
v_\text{Hxc}^\text{ea}(t) = v_\text{Hxc}^\text{gs}[\rho(t)].
\label{eq:vHxcea}
\end{align}
The exact non-adiabatic correlation potential in TDKS is then defined as the difference between the exact Hxc potential and the exact adiabatic Hxc potential 
\begin{align*}
v_{\text{c}}^\text{na}(t) =  v_{\text{Hxc}}(t) - v_\text{Hxc}^\text{ea}(t).
\end{align*}
Again, the relevant gauge-invariant quantities are the potential differences $\Delta v_{\text{Hxc}}^\text{ea,sc}(t) = v_{\text{Hxc},1}^\text{ea,sc}(t) -  v_{\text{Hxc},2}^\text{ea,sc}(t)$, $\Delta v_{\text{Hxc}}^\text{ea}(t) = v_{\text{Hxc},1}^\text{ea}(t) -  v_{\text{Hxc},2}^\text{ea}(t)$ and $\Delta v_{\text{c}}^\text{na}(t) = v_{\text{c},1}^\text{na}(t) - v_{\text{c},2}^\text{na}(t)$.

\begin{figure*}[t]
\includegraphics[scale=0.27,angle=-90]{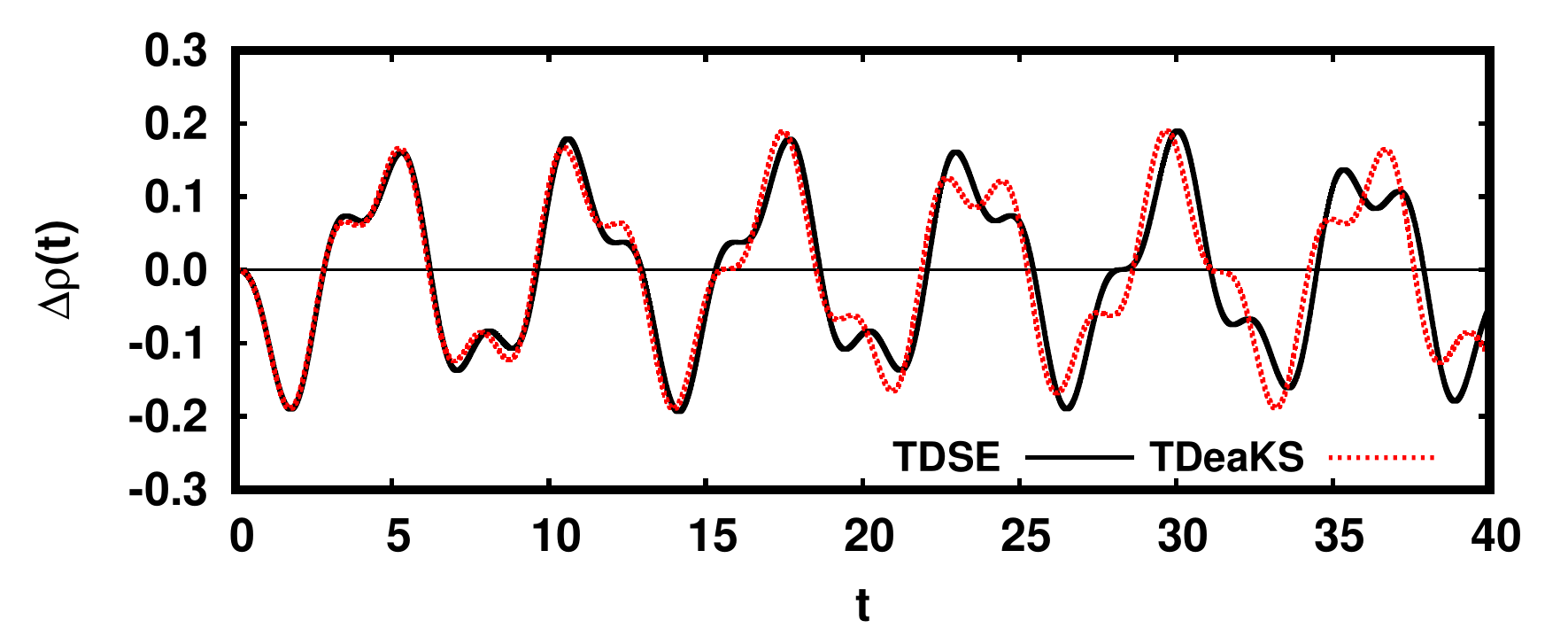}
\includegraphics[scale=0.27,angle=-90]{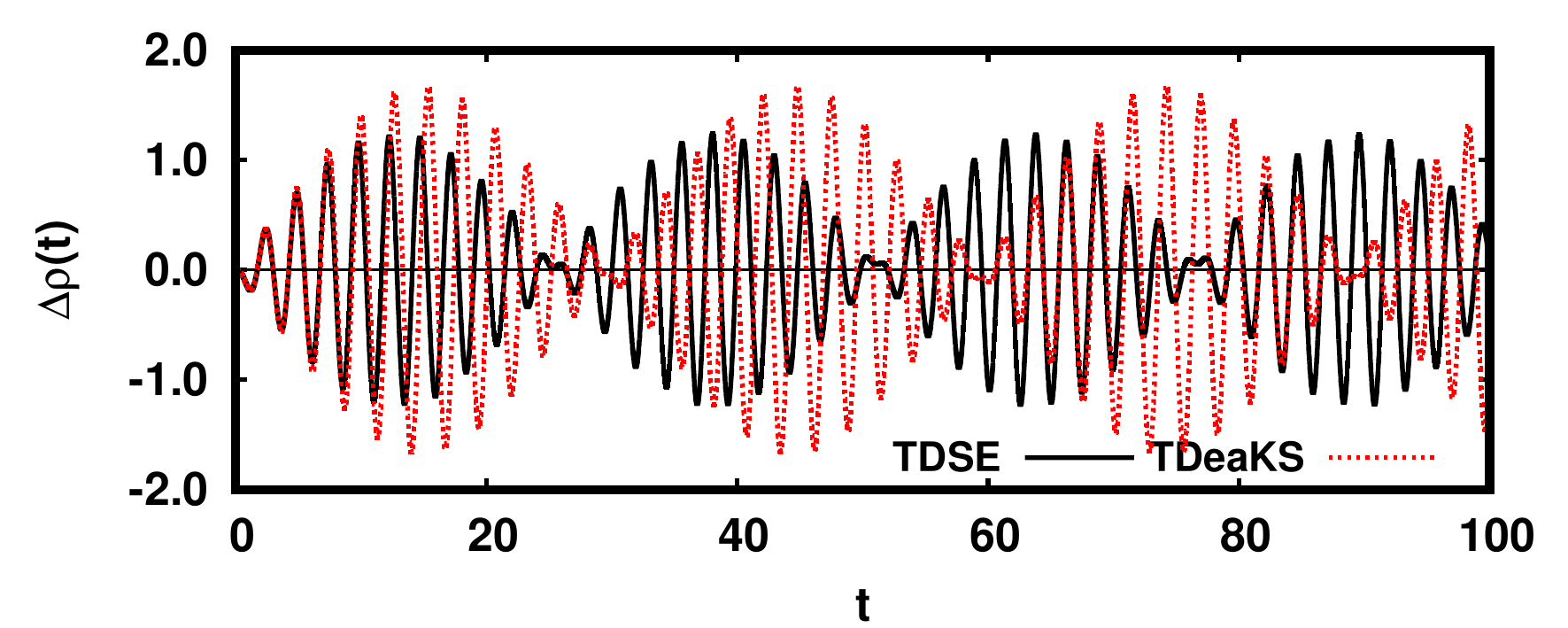}\\
\includegraphics[scale=0.27,angle=-90]{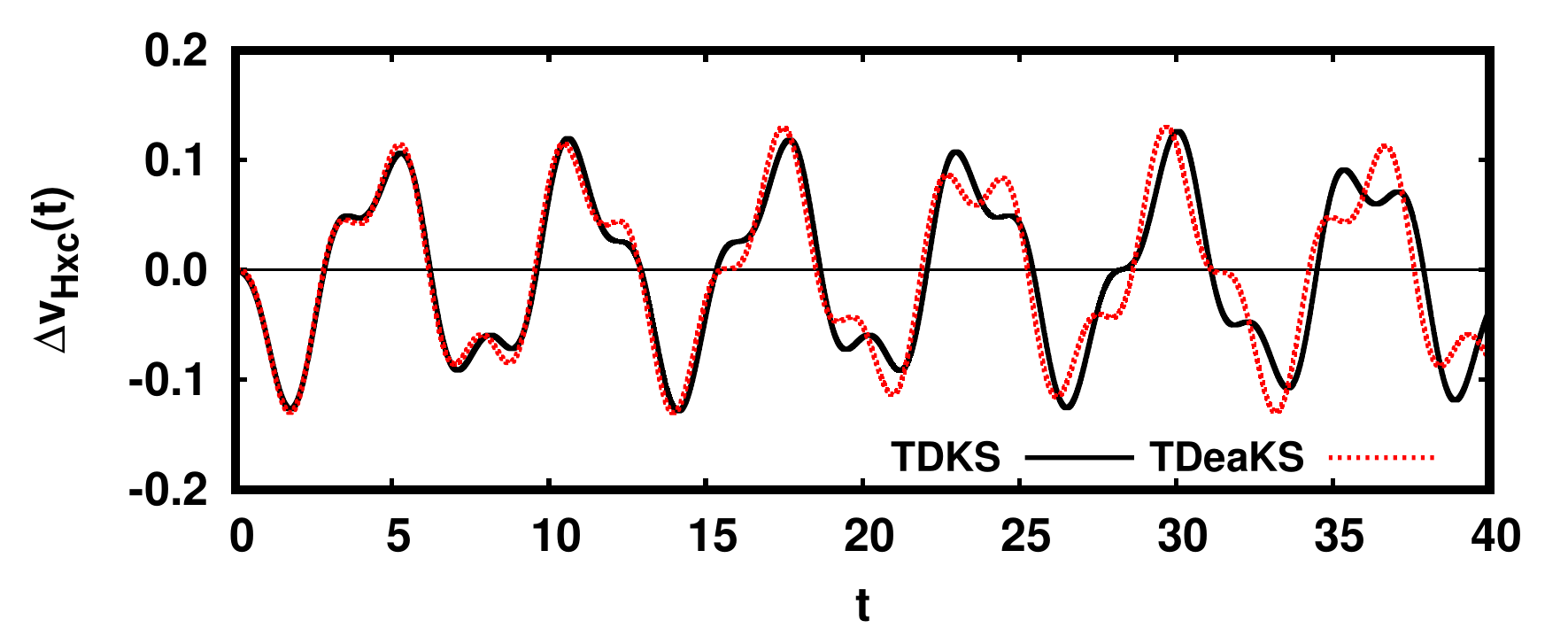}
\includegraphics[scale=0.27,angle=-90]{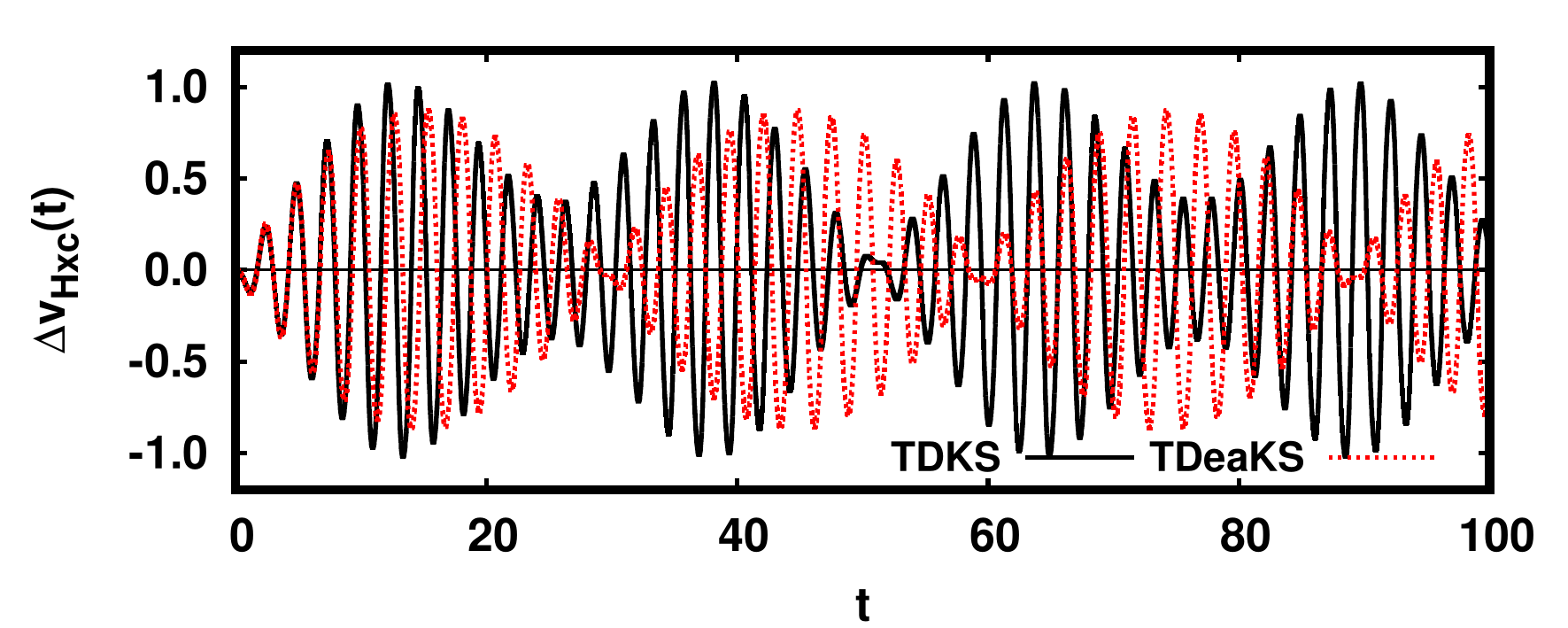}\\
\includegraphics[scale=0.27,angle=-90]{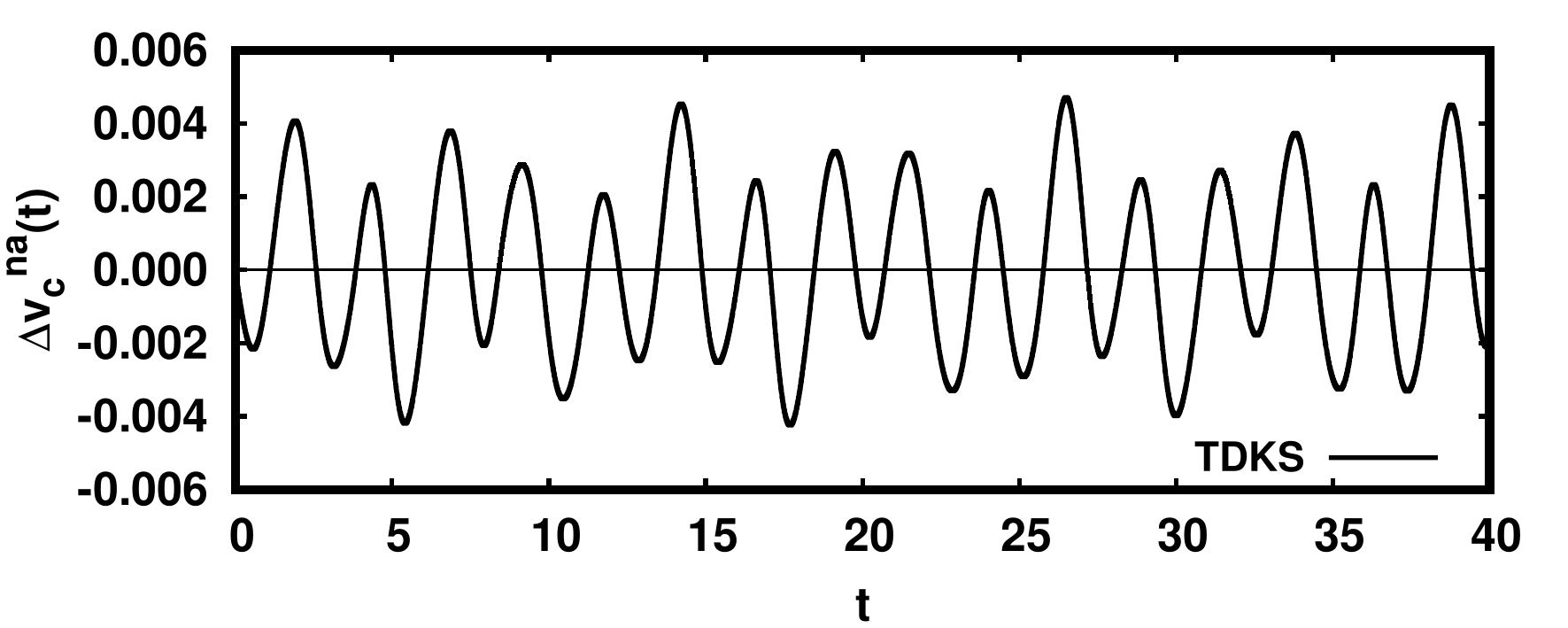}
\includegraphics[scale=0.27,angle=-90]{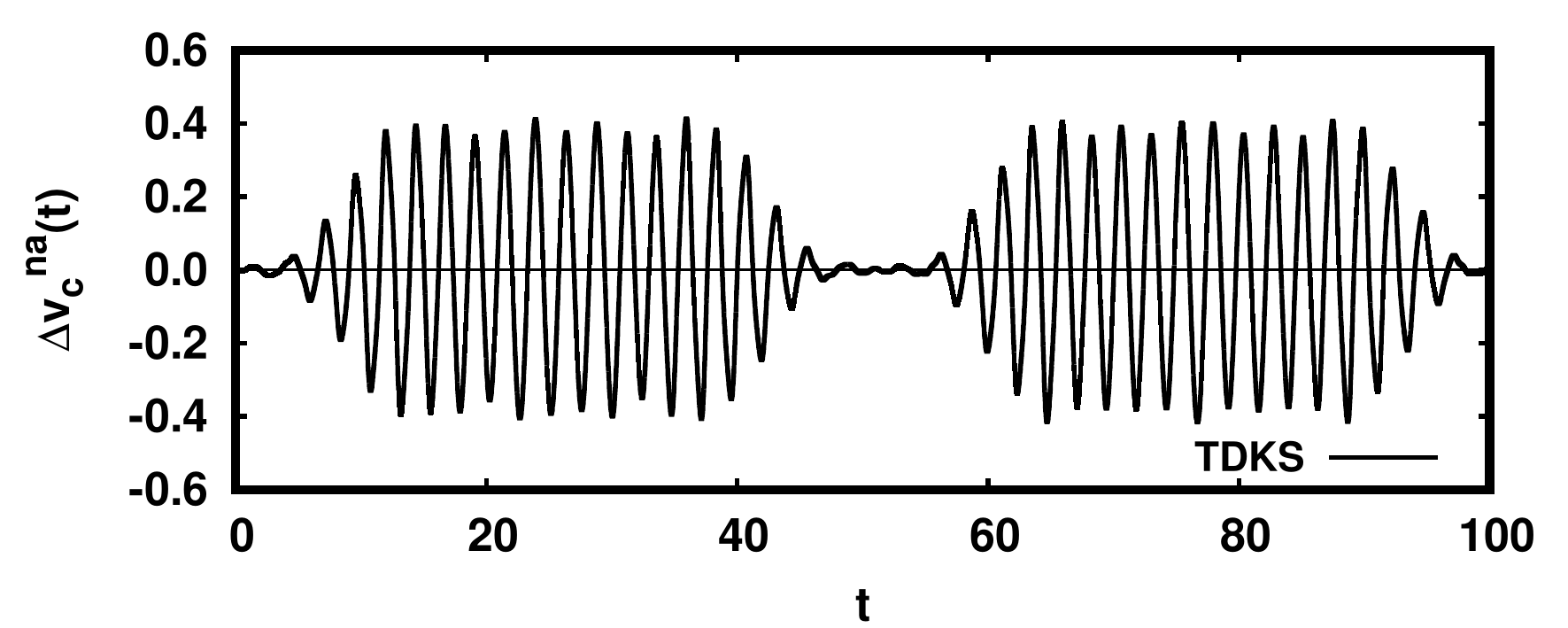}\\
\includegraphics[scale=0.27,angle=-90]{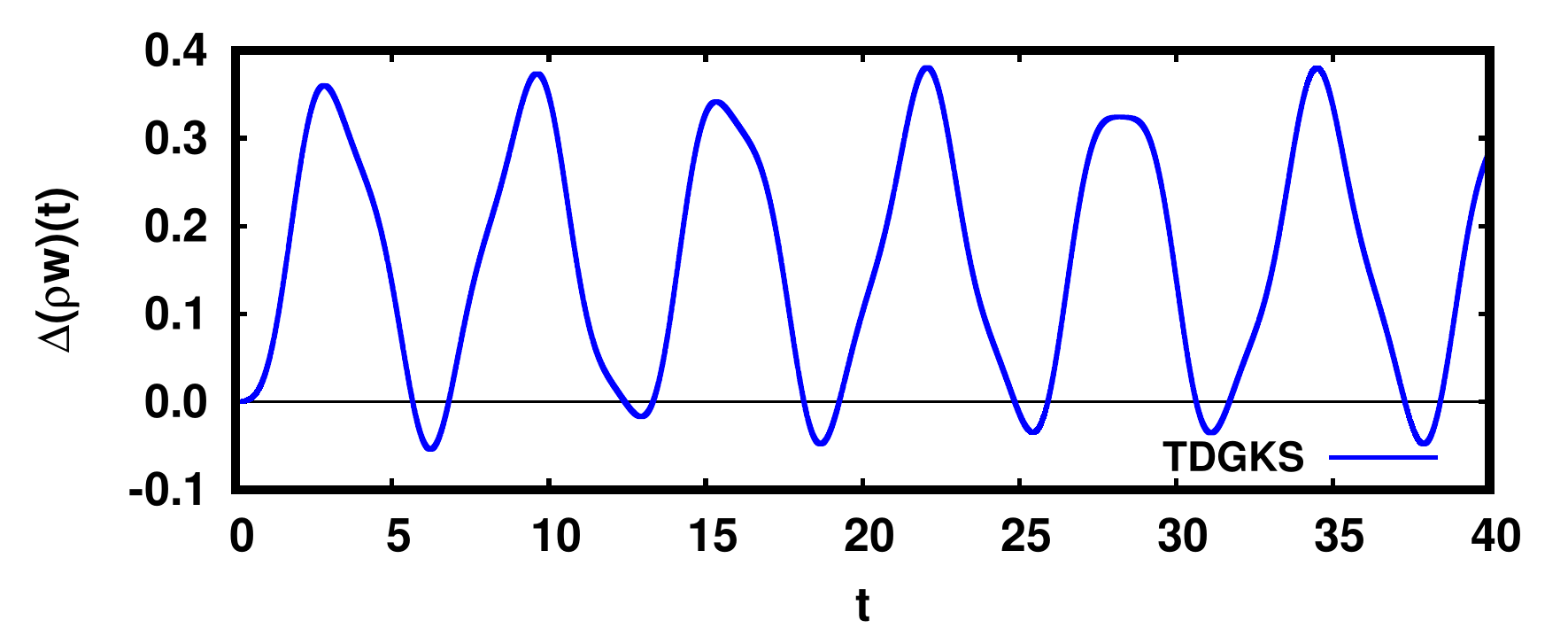}
\includegraphics[scale=0.27,angle=-90]{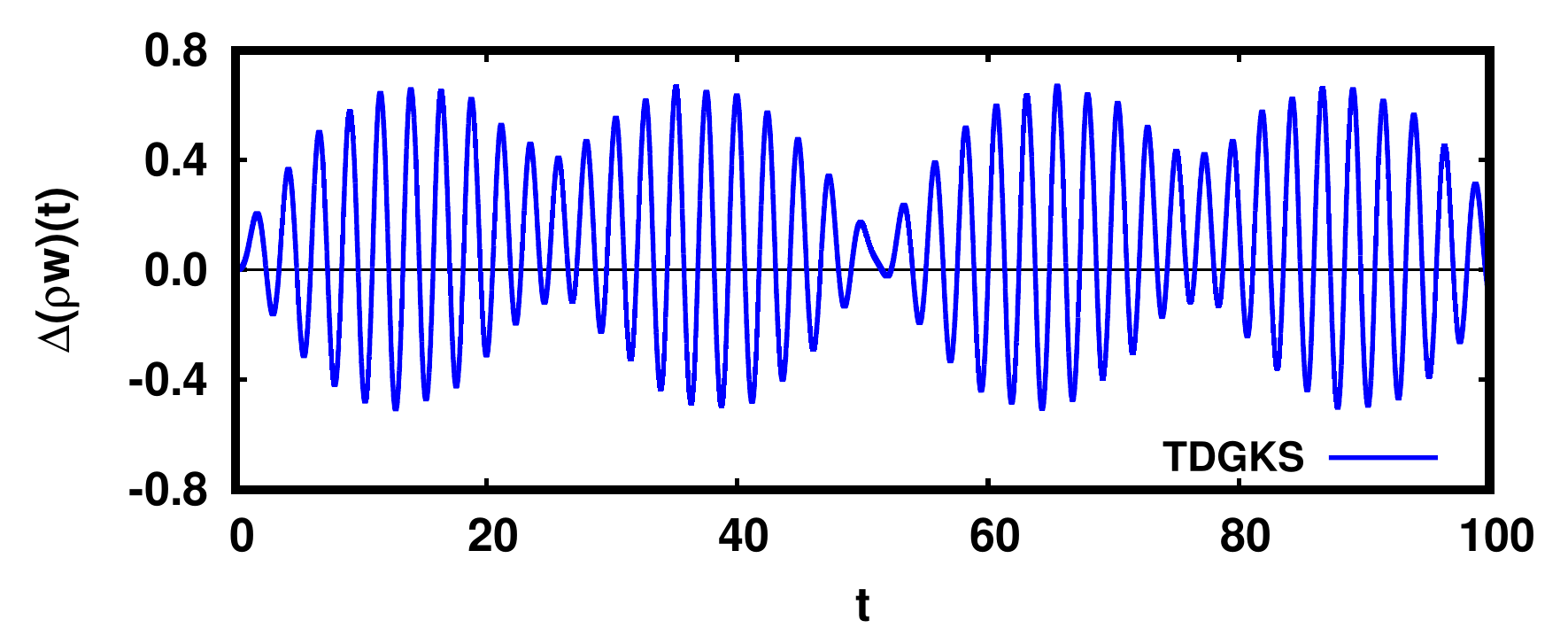}\\
\includegraphics[scale=0.27,angle=-90]{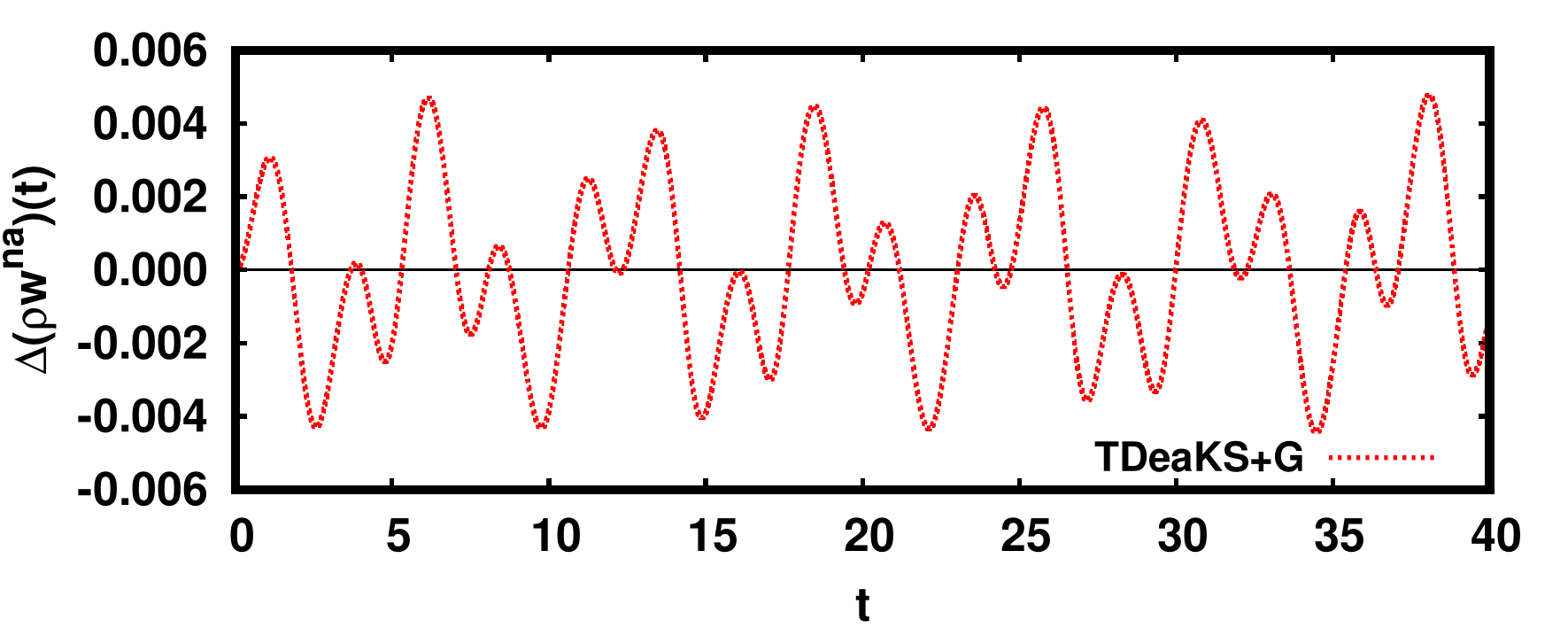}
\includegraphics[scale=0.27,angle=-90]{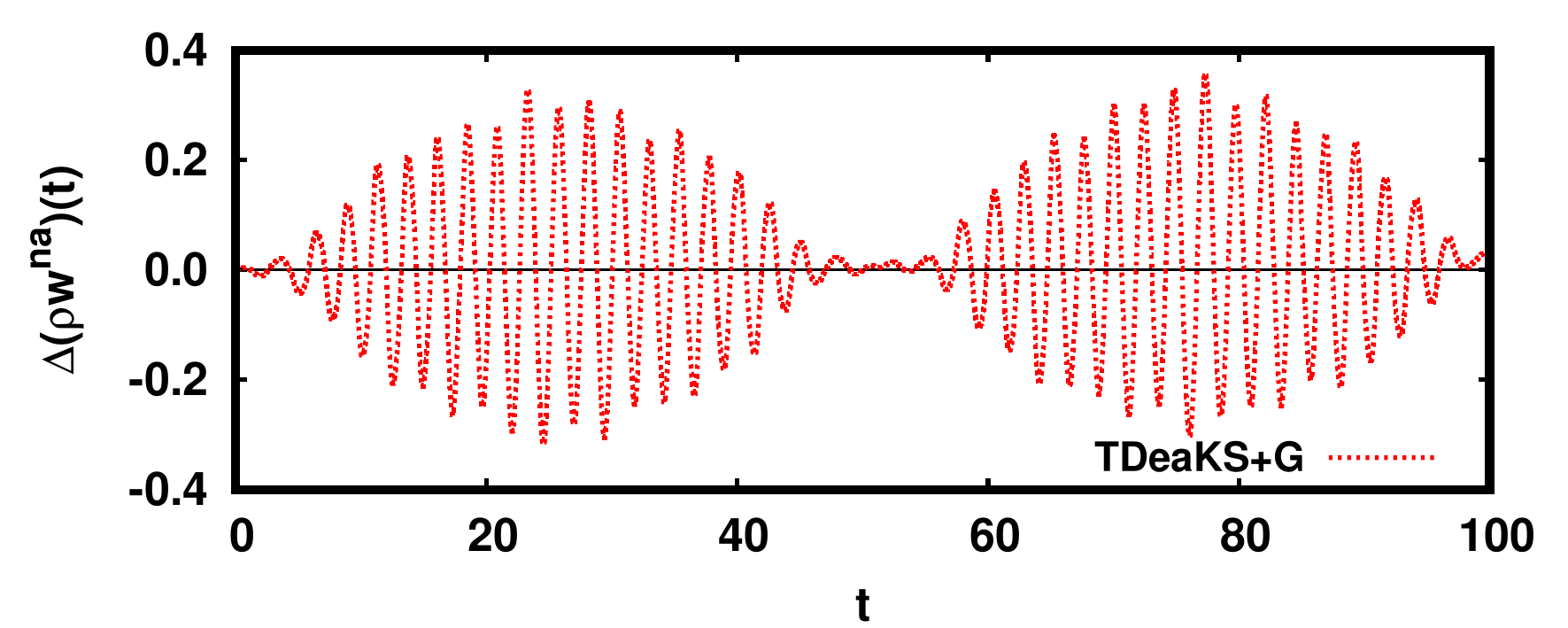}
\caption{Symmetric Hubbard dimer ($\tau=1$, ${\cal U}=1$, $\Delta v_{\text{ext}}^0 = 0$) starting from the initial delocalized ground state and driven by a time-dependent electric-dipole perturbation ($\cE_0=0.2$). The left panel corresponds to an off-resonant (nearly adiabatic) driving frequency ($\omega = 1$) and the right panel corresponds to a resonant (strongly non-adiabatic) driving frequency ($\omega = 2.56$). First row: TDSE and TDeaKS densities. Second row: TDKS and TDeaKS Hxc potentials. Third row: TDKS non-adiabatic correlation potential. Fourth row: TDGKS geometric potential. Fifth row: TDeaKS+G non-adiabatic geometric potential.}
\label{fig:HubbardDimer}
\end{figure*}

\begin{figure}[t]
\includegraphics[scale=0.22,angle=-90]{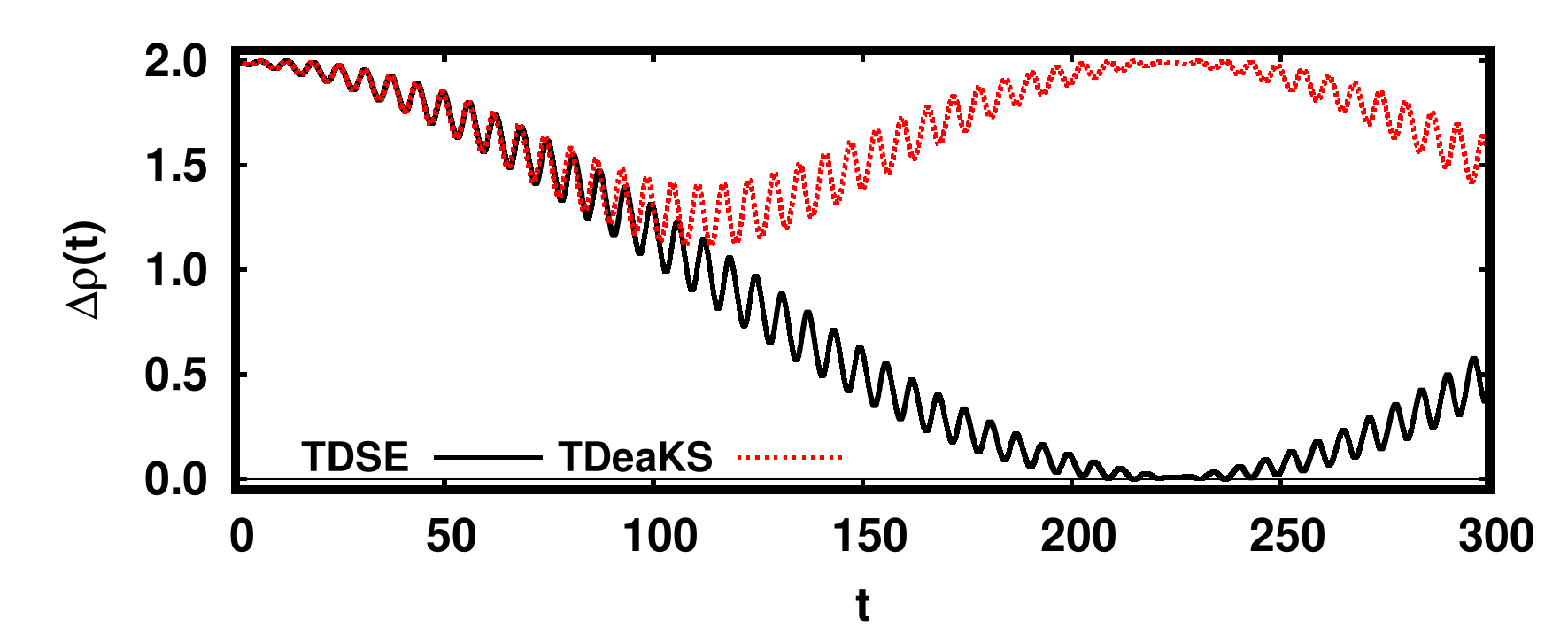}\\
\includegraphics[scale=0.22,angle=-90]{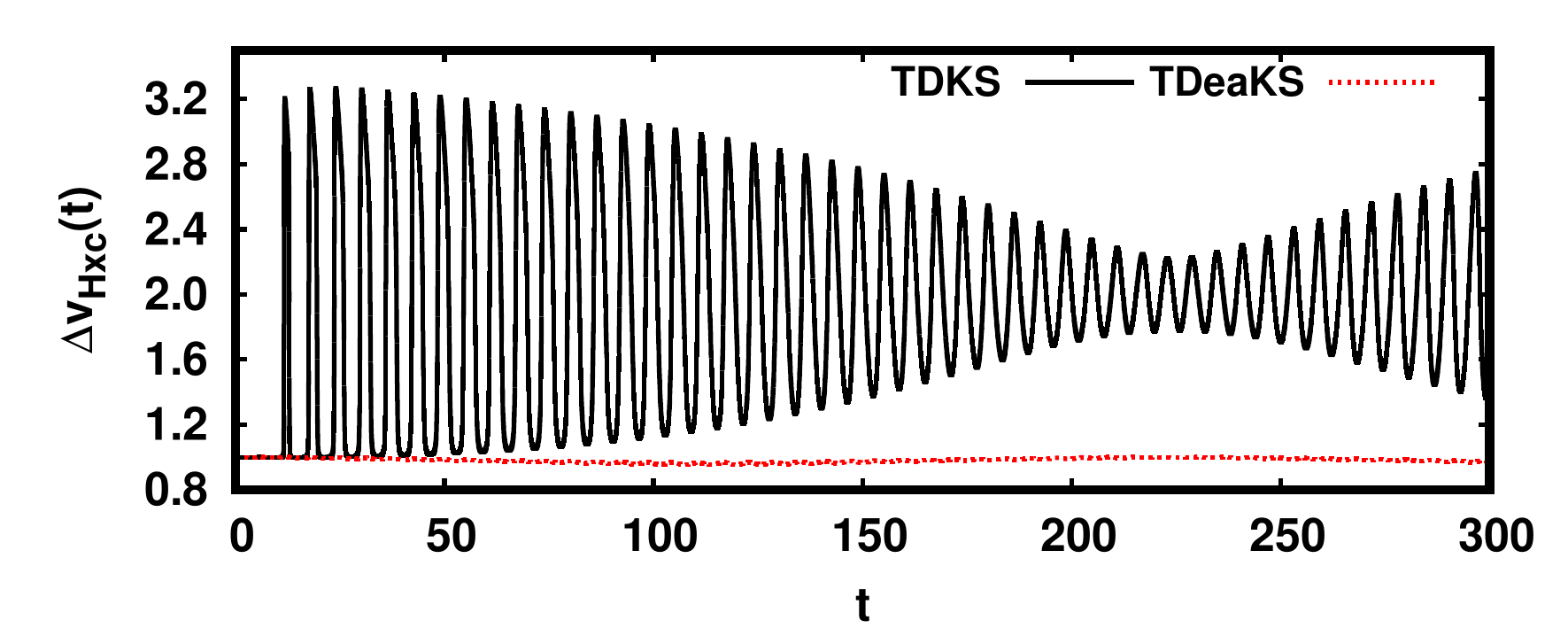}\\
\includegraphics[scale=0.22,angle=-90]{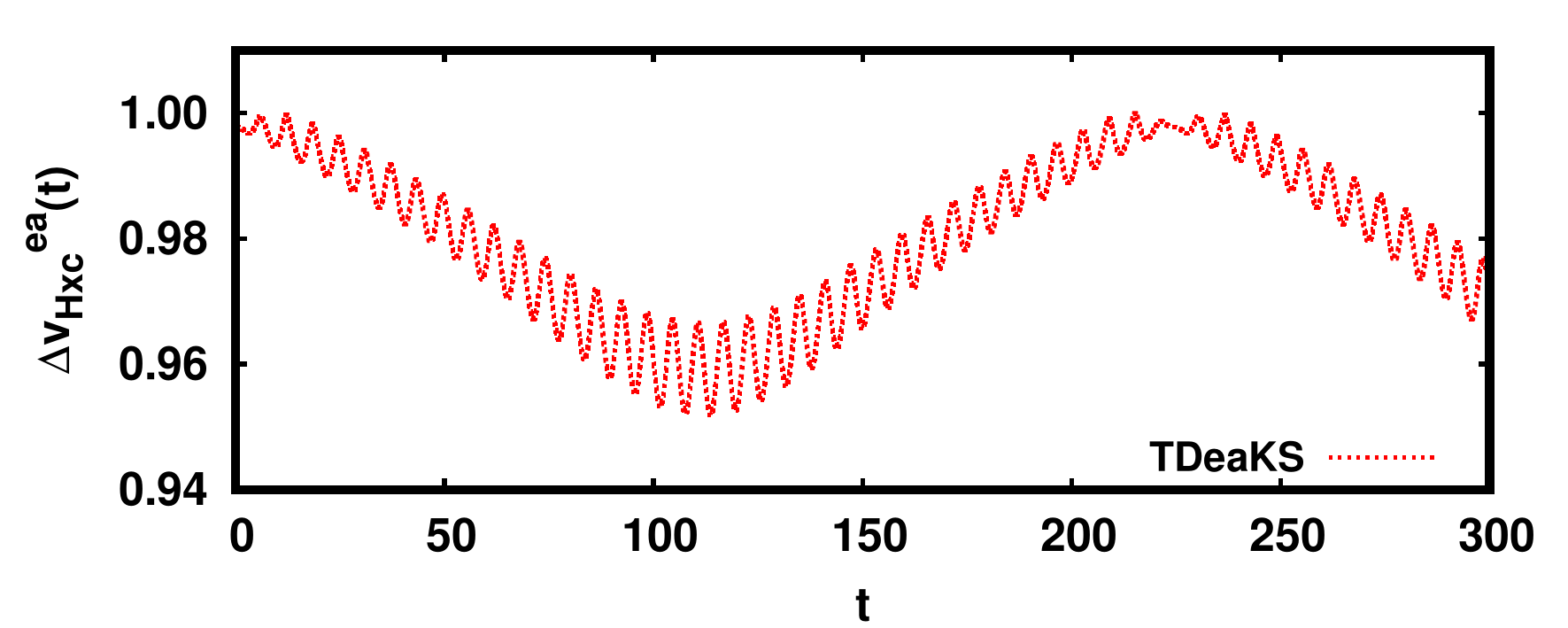}\\
\includegraphics[scale=0.22,angle=-90]{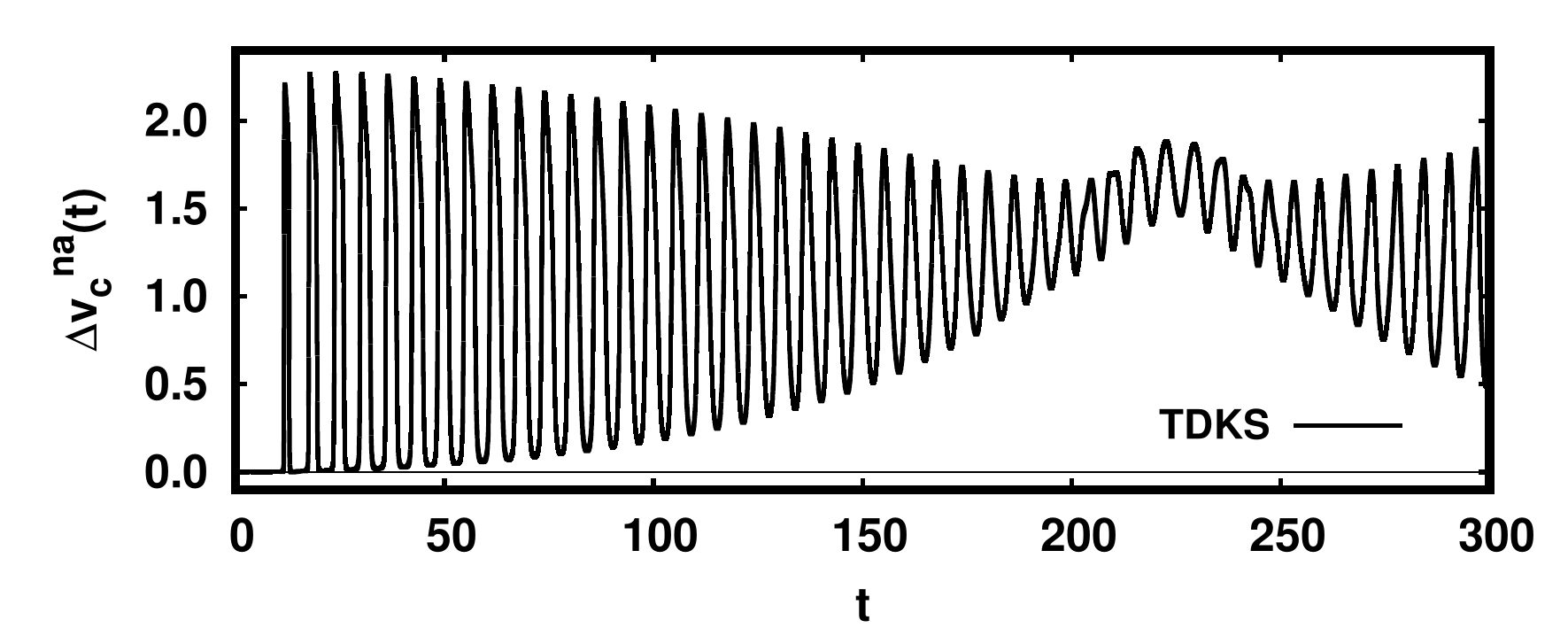}\\
\includegraphics[scale=0.22,angle=-90]{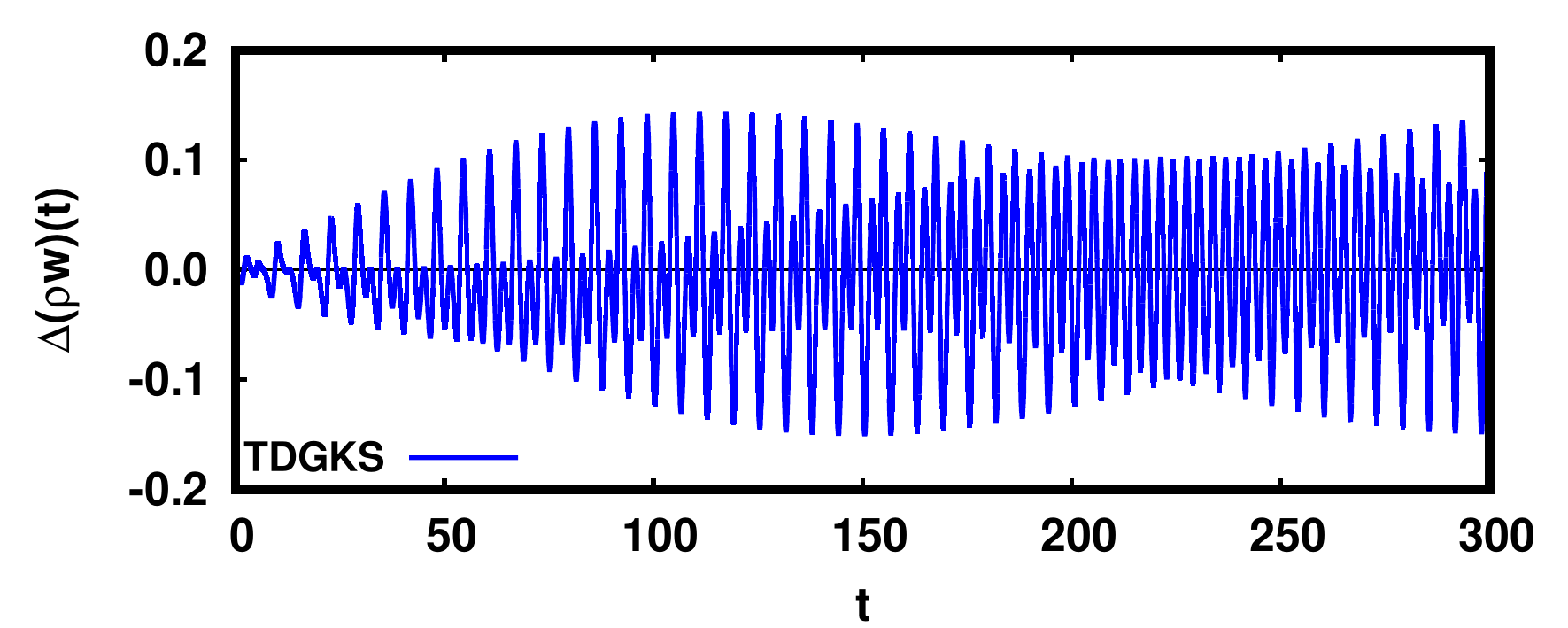}\\
\includegraphics[scale=0.22,angle=-90]{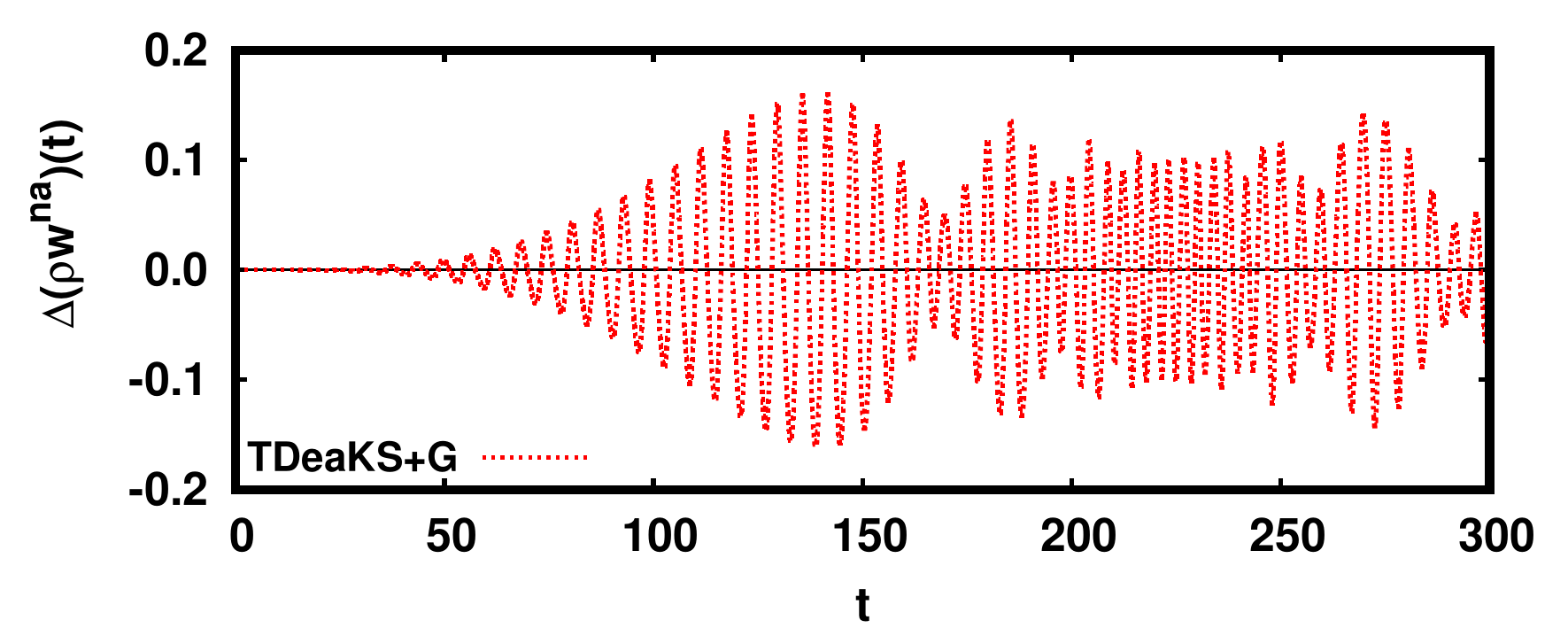}
\caption{Asymmetric Hubbard dimer ($\tau=0.05$, ${\cal U}=1$, $\Delta v_{\text{ext}}^0 = -2$) starting from the initial localized ground state and driven by a time-dependent electric-dipole perturbation ($\cE_0=0.2$) with resonant (non-adiabatic) driving frequency ($\omega = 1.0083$). First row: TDSE and TDeaKS densities. Second row: TDKS and TDeaKS Hxc potentials. Third row: zoom on TDeaKS. Fourth row: TDKS non-adiabatic correlation potential. Fifth row: TDGKS geometric potential. Sixth row: TDeaKS+G non-adiabatic geometric potential.}
\label{fig:HubbardDimer_CT}
\end{figure}

\paragraph*{Geometric Kohn--Sham scheme.}
We go on by solving the TDGKS equation, still starting from the exact ground-state Kohn--Sham orbital $\varphi_0 = \sqrt{\rho(0)/2}$,
\begin{align*}
i \partial_t \varphi^\text{GKS}(t) = \Big( h(t) + i w(t) \Big) \varphi^\text{GKS}(t),
\end{align*}
where $w(t)$ is the diagonal matrix containing the geometric potential
\begin{align*}
w(t) =
\begin{pmatrix}
w_{1}(t) & 0\\
0 & w_{2}(t)
\end{pmatrix},
\end{align*}
with the condition $\rho_1(t) w_1(t) + \rho_2(t) w_2(t) = 0$. We calculate the exact geometric potential $w(t)$ giving the exact density $\rho(t)$ using the modified continuity equation~\eqref{eq:modconteqGKS}. In this case, it is convenient to introduce the density-weighted potential difference $\Delta (\rho w)(t) = \rho_1(t) w_1(t) - \rho_2(t) w_2(t)$, from which $w_1(t)$ and $w_2(t)$ can be obtained as $w_1(t) = \Delta (\rho w)(t)/(2 \rho_1(t))$ and $w_2(t) = -\Delta (\rho w)(t)/(2 \rho_2(t))$. It can be checked that~\eqref{eq:modconteqGKS} leads to the following expression for $\Delta (\rho w)(t)$
\begin{align}
\Delta (\rho w)(t) = \frac{\Delta \rho'(t)}{2} + \tau \sqrt{4 - (\Delta \rho(t))^2} \sin \big(\beta(t) \big),
\label{eq:Deltarhow}
\end{align}
where now $\beta(t) = \arg(\varphi_2^\text{GKS}(t) \overline{\varphi_1^\text{GKS}(t)})$. Performing a similar calculation as in Section~\ref{sec:qubit_geometric}, the angle $\beta(t)$ satisfies the nonlinear differential equation
\begin{align}
\beta'(t) = \Delta v_\text{ext}(t) + 2 \tau \frac{\Delta \rho(t)}{\sqrt{4 - (\Delta \rho(t))^2}} \cos \big( \beta(t) \big),
\label{eq:dbetadt_dimer}
\end{align}
whose solution gives $\beta(t)$ as a nonlocal time functional of $\Delta v_\text{ext}$ and $\Delta \rho$, and of course depending on the initial condition $\beta(0)$. Note that, contrary to TDKS, there is no restriction in TDGKS on the values that the angle $\beta(t)$ can take.

\paragraph*{Exact adiabatic Kohn--Sham scheme with geometric correction.}
Finally, we solve the time-dependent exact adiabatic Kohn--Sham equation with geometric correction (TDeaKS+G)
\begin{align}
i \partial_t \varphi^\text{eaKS+G}(t) &= \Big( h(t) + v_\text{Hxc}^\text{ea}(t)
+  i w^\text{na}(t) \Big) \varphi^\text{eaKS+G}(t),
\label{eq:KS+Gdimer}
\end{align}
with $\rho_1(t) w_1^\text{na}(t) + \rho_2(t) w_2^\text{na}(t) = 0$. The exact non-adiabatic geometric potential $w^\text{na}(t)$ giving the exact density $\rho(t)$ is again obtained from the modified continuity equation~\eqref{eq:modconteqKS+G}. The density-weighted potential difference $\Delta (\rho w^\text{na})(t) = \rho_1(t) w_1^\text{na}(t) - \rho_2(t) w_2^\text{na}(t)$ still satisfies~\eqref{eq:Deltarhow} and~\eqref{eq:dbetadt_dimer} after making the replacement $\Delta v_\text{ext}(t) \to \Delta v_\text{ext}(t) + \Delta v_\text{Hxc}^\text{ea}(t)$.

In all cases, we calculate the dynamics with a home-made program using straightforward exponential propagation with a time step of $10^{-4}$.

\subsection{Numerical results}

We discuss now the numerical results for symmetric and asymmetric Hubbard dimers.

\paragraph*{Symmetric Hubbard dimer.}
In Figure~\ref{fig:HubbardDimer}, we report results for the symmetric Hubbard dimer (with parameters $\tau=1$, ${\cal U}=1$, $\Delta v_{\text{ext}}^0 = 0$) starting from the initial delocalized ground state ($\Delta \rho(0) = 0$) and driven by a time-dependent electric-dipole perturbation (with $\cE_0=0.2$). The left panel corresponds to an off-resonant (nearly adiabatic) driving frequency ($\omega = 1$) and the right panel corresponds to a resonant (non-adiabatic) driving frequency ($\omega = 2.56$, which is very close to the first excitation energy $E_1 - E_0$).

Let us start by discussing the results for the off-resonant driving frequency. The exact density $\Delta \rho(t)$ from the TDSE is slowly oscillating around zero at the driving perturbation period $2\pi/\omega \approx 6.3$. The TDeaKS density $\Delta \rho^\text{eaKS}(t)$, i.e. obtained with the exact adiabatic approximation, nearly follows the exact $\Delta \rho(t)$ at short times and then weakly differs from it at longer times. A similar behavior is observed for the exact Hxc potential $\Delta v_{\text{Hxc}}(t)$ and its self-consistent exact adiabatic approximation $\Delta v_{\text{Hxc}}^\text{ea,sc}(t)$. In fact, the time profiles of the potentials $\Delta v_{\text{Hxc}}(t)$ and $\Delta v_{\text{Hxc}}^\text{ea,sc}(t)$ closely ressemble the time profiles of their respective densities $\Delta \rho(t)$ and $\Delta \rho^\text{eaKS}(t)$, which can be understood from the fact that the time-dependent state stays close to the ground state for which the Hartree-exchange contribution is dominant, i.e. $\Delta v_{\text{Hxc}}(t) \approx \Delta v_{\text{Hx}}(t)={\cal U} \Delta \rho(t)/2$ and $\Delta v_{\text{Hxc}}^\text{ea,sc}(t) \approx \Delta v_{\text{Hx}}^\text{ea,sc}(t)={\cal U} \Delta \rho^\text{eaKS}(t)/2$. The exact non-adiabatic correlation potential $\Delta v_\text{c}^\text{na}(t)$, i.e. the potential needed to correct for the exact adiabatic approximation in standard TDKS, oscillates around zero and is two orders of magnitude smaller than $\Delta v_{\text{Hxc}}(t)$, confirming that non-adiabatic effects are very small for this off-resonant driving frequency. 
For the TDGKS calculations, the exact geometric potential $\Delta (\rho w)(t)$ has a somewhat peculiar time profile, which essentially comes from the $\sin (\beta(t))$ contribution in~\eqref{eq:Deltarhow}. The exact non-adiabatic geometry potential $\Delta (\rho w^\text{na})(t)$ from the TDeaKS+G scheme oscillates around zero and is of the same order of magnitude as $\Delta v_\text{c}^\text{na}(t)$, which could have been expected. 

Let us now discuss the results for the resonant driving frequency, leading to Rabi oscillations between the ground state and the first excited state with Rabi period $T_\text{R} = 2\pi/(|d_\text{ge}| \cE_0) \approx 51$, where $|d_\text{ge}|\approx 0.615$ is the transition dipole moment between the ground and excited state. The cases of TDKS and TDeaKS were discussed in~\cite{FukFarTokAppKurRub-PRA-13}, but note that we have made the choice of performing the TDeaKS calculations with the same driving frequency $\omega=2.56$ and not with the linear-response TDeaKS frequency as in~\cite{FukFarTokAppKurRub-PRA-13}. The exact density $\Delta \rho(t)$ has fast oscillations at the driving perturbation period $2\pi/\omega \approx 2.5$ with an envelope exhibiting slower oscillations of period $T_\text{R}/2 \approx 25.5$, which is expected since the ground- and excited-state densities are the same. The TDeaKS density $\Delta \rho^\text{eaKS}(t)$ follows the exact density at short times but quickly develops an increasing time delay in the envelope oscillations together with an overestimated amplitude.  The exact Hxc potential $\Delta v_{\text{Hxc}}(t)$ shows envelope oscillations at the Rabi period $T_\text{R}$, and not $T_\text{R}/2$ like the exact density. By contrast, the self-consistent exact adiabatic Hxc potential $\Delta v_{\text{Hxc}}^\text{ea,sc}(t)$ follows the time profile of the corresponding density $\Delta \rho^\text{eaKS}(t)$, with again $\Delta v_{\text{Hxc}}^\text{ea,sc}(t) \approx \Delta v_{\text{Hx}}^\text{ea,sc}(t) = \cU \Delta \rho^\text{eaKS}(t)/2$, showing that the adiabatic correlation contribution remains small. As expected, the exact non-adiabatic correlation potential $\Delta v_\text{c}^\text{na}(t)$ shows envelope oscillations at the Rabi period $T_\text{R}$ and represents an important contribution to $\Delta v_{\text{Hxc}}(t)$, especially at half Rabi cycles when it is dominant. Moving now to the TDGKS and TDeaKS+G calculations, we find that the exact geometric potential $\Delta (\rho w)(t)$ and the exact non-adiabatic geometry potential $\Delta (\rho w^\text{na})(t)$ also exhibit envelope oscillations at the Rabi period $T_\text{R}$. As we could have expected, $\Delta (\rho w^\text{na})(t)$ is quite significant for this case of resonant driving frequency and it is maximal at half Rabi cycles.

\paragraph*{Asymmetric Hubbard dimer.}
In Figure~\ref{fig:HubbardDimer_CT}, we report results for the asymmetric Hubbard dimer (with parameters $\tau=0.05$, ${\cal U}=1$, $\Delta v_{\text{ext}}^0 = -2$) starting from the initial localized ground state ($\Delta \rho(0) = 1.995$) and driven by a time-dependent electric-dipole perturbation (with $\cE_0=0.2$) with resonant (non-adiabatic) driving frequency ($\omega = 1.0083$, corresponding to the excitation energy $E_1 - E_0$ to the first excited state which has $\Delta \rho \approx 0$). This is a toy model of a long-range charge transfer in a molecule which was studied in~\cite{FukMai-PRA-14} for the cases of TDKS and TDeaKS. We see indeed that the exact TDSE density $\Delta \rho(t)$ evolves from its initial value of $1.995$ (corresponding to almost 2 electrons on site 1) to the value zero (corresponding to one electron on each site) at $t \approx 224$, corresponding to half a Rabi cycle. As observed in~\cite{FukMai-PRA-14}, the exact adiabatic approximation completely fails to describe this phenomenon: as seen from its density, TDeaKS is only able to transfer at most a fraction of electron from site 1 to site 2 at $t \approx 110$. The exact Hxc potential $\Delta v_{\text{Hxc}}(t)$ develops large-amplitude oscillations at the driving perturbation period $2\pi/\omega \approx 6.2$. As noticed in~\cite{FukMai-PRA-14}, the large peaks of $\Delta v_{\text{Hxc}}(t)$ correspond to the denominator $4 \tau^2 (4 - (\Delta \rho(t))^2) - (\Delta \rho'(t))^2$ in~\eqref{eq:DeltavHxc} being close to zero, meaning that we are near the edge of the non-interacting $v$-representability domain. Modulo these oscillations, $\Delta v_{\text{Hxc}}(t)$ overall makes a step from its ground-state value $\Delta v_{\text{Hxc}}^\text{gs} \approx 1$ to the value $\Delta v_{\text{Hxc}} \approx 2$ at half a Rabi cycle, the latter value being the value required to cancel out the static external potential $\Delta v_\text{ext}^0 = - 2$ and give the uniform density ($\Delta \rho = 0$) of the charge-transfer state. This step in $\Delta v_{\text{Hxc}}(t)$, essentially coming from the non-adiabatic correlation contribution $\Delta v_{\text{Hxc}}^\text{na}(t)$, is the way TDKS manages to transfer one electron from site 1 to site 2, in the absence of the two-electron repulsion. The self-consistent exact adiabatic Hxc potential $\Delta v_{\text{Hxc}}^\text{ea,sc}(t)$ is unable to reproduce this dynamical step. Moving now to the TDGKS and TDeaKS+G calculation, we see that the time profiles of the exact geometric potential $\Delta (\rho w)(t)$ and the exact non-adiabatic geometry potential $\Delta (\rho w^\text{na})(t)$ are quite different from those of the previously considered potentials. Neither $\Delta (\rho w)(t)$ nor $\Delta (\rho w^\text{na})(t)$ display any dynamical steps, but they present non-trivial oscillatory patterns, which can be traced back to the superposition of the two terms in the right-hand side of~\eqref{eq:Deltarhow}.

\subsection{An alternative definition of the adiabatic approximation}

We conclude this section with a remark on the possibility of using an alternative definition of the adiabatic approximation. In~\eqref{eq:vHxcea}, the exact adiabatic Hxc potential $v_{\text{Hxc}}^\text{ea}(t)$ is defined as the exact ground-state Hxc potential functional $v_{\text{Hxc}}^\text{gs}[\rho]$ evaluated at the exact time-dependent density $\rho(t)$
$$v_{\text{Hxc}}^\text{ea}(t) = v_{\text{Hxc}}^\text{gs}[\rho(t)]
= v_\text{s}^\text{gs}[\rho(t)] - v_\text{ext}^\text{gs}[\rho(t)],$$
where $v_\text{s}^\text{gs}[\rho]$ is the total ground-state Kohn--Sham potential (defined so that the non-interacting system has ground-state density $\rho$) and $v_\text{ext}^\text{gs}[\rho]$ is the total ground-state external potential, defined so that the interacting system has ground-state density $\rho$. An alternative definition of the exact adiabatic Hxc potential is
\begin{align}
\tilde{v}_{\text{Hxc}}^\text{ea}(t) = v_\text{s}^\text{gs}[\rho(t)] - v_\text{ext}(t),
\label{eq:vHxcea_alternative}
\end{align}
which is different from $v_{\text{Hxc}}^\text{ea}(t)$ since in general $v_\text{ext}^\text{gs}[\rho(t)] \not = v_\text{ext}(t)$. This leads to an alternative TDeaKS+G scheme
\begin{align*}
i \partial_t \tilde{\varphi}^\text{eaKS+G}(t) = \Big( h(t) + \tilde{v}_\text{Hxc}^\text{ea}(t)
+  i \tilde{w}^\text{na}(t) \Big) \tilde{\varphi}^\text{eaKS+G}(t),
\end{align*}
where the alternative exact non-adiabatic geometric potential $\tilde{w}^\text{na}(t)$ is still determined from~\eqref{eq:modconteqKS+G} so as to give the exact density $\rho(t)$. The corresponding density-weighted potential difference $\Delta (\rho \tilde{w}^\text{na})(t)$ still satisfies~\eqref{eq:Deltarhow} and~\eqref{eq:dbetadt_dimer} after making the replacement $\Delta v_\text{ext}(t) \to \Delta v_\text{ext}(t) + \Delta \tilde{v}_\text{Hxc}^\text{ea}(t)$. Contrary to $\Delta v_{\text{Hxc}}^\text{ea}(t)$, the alternative exact adiabatic Hxc potential difference $\Delta \tilde{v}_{\text{Hxc}}^\text{ea}(t)$ has a simple explicit expression
\begin{align*}
\Delta \tilde{v}_\text{Hxc}^\text{ea}(t) =& -\frac{2\tau \Delta \rho(t)}{\sqrt{4 - (\Delta \rho(t))^2}} - \Delta v_\text{ext}(t),
\label{eq:DeltavHxcea_alternative}
\end{align*}
which, after insertion in~\eqref{eq:dbetadt_dimer}, gives $\beta(t)=\beta(0)=0$ and thus
\begin{align*}
\Delta (\rho \tilde{w}_\text{na})(t) = \frac{\Delta \rho'(t)}{2}.
\end{align*}
This is a remarkable simple expression, which may suggest that the alternative definition of the exact adiabatic Hxc potential in~\eqref{eq:vHxcea_alternative} is perhaps a better definition, at least in the context of our geometric TDDFT approach. Note, however, that it would not make much sense to define an alternative self-consistent exact adiabatic Kohn--Sham scheme of the form $i \partial_t \tilde{\varphi}^\text{eaKS}(t) = ( h(t) + \tilde{v}_\text{Hxc}^\text{ea,sc}(t) ) \tilde{\varphi}^\text{eaKS}(t)$ with an alternative self-consistent exact adiabatic Hxc potential $\tilde{v}_{\text{Hxc,sc}}^\text{ea}(t) = v_\text{s}^\text{gs}[\tilde{\rho}^\text{eaKS}(t)] - v_\text{ext}(t)$ where $\tilde{\rho}^\text{eaKS}_m(t)=2|\tilde{\varphi}_m^\text{eaKS}(t)|^2$, since in this case $v_\text{ext}(t)$ would just cancel out and the density $\tilde{\rho}^\text{eaKS}(t)$ would remain equal to the initial ground-state density at all times.

\section{Conclusion}\label{sec:conclusion}
In this article, we have developed a geometric framework for Schrödinger dynamics with constraints and used this formalism to revisit the foundations of TDDFT. The resulting time-dependent Kohn--Sham-type equation involves a correction term that can either be interpreted as a local imaginary potential $w(t)$ or, better, a nonlocal exchange-type Hermitian operator (Eq.~\eqref{eq:KS_McLachlan_commutator} or Eq.~\eqref{eq:KS+G}). Numerical tests on the Hubbard dimer show that the geometric correction $w(t)$ has a very different structure from that of the standard exact time-dependent Kohn--Sham potential. This works paves the way for alternative approximations in TDDFT, which could potentially better describe systems in non-adiabatic regimes. In~\cite{MAQUI_TDDFT-26_ppt} we study the continuous case and introduce a geometric Kohn--Sham scheme based on a \emph{universal} functional (independent of the external potential).


\begin{acknowledgments}
This work has benefited from French State support managed by ANR under the France 2030 program through the MaQuI CNRS Risky and High-Impact Research programme (RI)$^2$ (grant agreement ANR-24-RRII-0001).
\end{acknowledgments}

\appendix

\section{The oblique principle for a single qubit} \label{app:qubit_oblique}

We study here the $\theta\to0$ limit of the oblique principle for a single qubit that we considered in Section~\ref{sec:qubit_oblique}. Our goal is to explain the rather singular behavior observed in Figure~\ref{fig:qubit_oblique}.

We look at the time-independent case $\rho_1(t)\equiv\rho_1$. Then~\eqref{eq:qbit_oblique} becomes an autonomous equation of the form
\begin{equation}
\tan(\theta)\,\beta'(t)=F_\theta\big(\beta(t)\big)
\label{eq:autonomous}
\end{equation}
with
$$F_\theta(x)=\frac{\rho_1+\rho_2}{\sqrt{\rho_1\rho_2}}\sin(x)
-\tan(\theta)\cos(x)\frac{\rho_2-\rho_1}{\sqrt{\rho_1\rho_2}}.$$
The function $F_\theta$ admits two consecutive zeroes $R_\theta$ and $R'_\theta$ that are close to $0$ and $\pi$, respectively, in the limit $\theta\to0$. In fact, taking $x=\pm\theta$ and $x=\pi\pm\theta$, we see that $|R_\theta|\leq\theta$ and $|R'_\theta-\pi|\leq\theta$. Hence if we start with a $\beta_0\in(0,\pi)$ we have $R_\theta<\beta_0<R'_\theta$ for $\theta$ small enough. In this case the solution $\beta_\theta(t)$ to the autonomous equation~\eqref{eq:autonomous} with $\beta_\theta(0)=\beta_0$ is increasing in time for $\theta>0$ and decreasing for $\theta<0$. It converges to either $R_\theta$ or $R_\theta'$ in the limits $t\to\pm\ii$.

When we take the limit $\theta\to0$, the coefficient $\tan(\theta)$ in front of the derivative compresses the function and increases the speed of convergence to its limits, so that it converges to a step function as we have observed in Figure~\ref{fig:qubit_oblique}. Let us for instance explain this phenomenon in the case $\theta>0$, where $\beta_\theta(t)$ is increasing in $t$. For $x$ in the interval $(\beta_0,\pi-\sqrt\theta)$ we have $F_\theta(x)\geq c\sqrt\theta x$ for some $c>0$. Hence $\beta_\theta'(t)\geq c\frac{\sqrt\theta}{\tan(\theta)}\beta_\theta(t)$ whenever $\beta_\theta(t)\leq \pi-\sqrt\theta$.
This shows that the time $T_\theta$ such that $\beta_\theta(T_\theta)=\pi-\sqrt\theta<R_\theta'$ must satisfy
$$ T_\theta\leq \frac{\tan(\theta)}{c\sqrt\theta}\log\left(\frac{\pi-\sqrt\theta}{\beta_0}\right).$$
In other words, $\beta_\theta(t)$ must have passed $\pi-\sqrt\theta$ at times of order $\sqrt{\theta}$. Now, if we fix a time $t>0$ and take $\theta\to0^+$, we find that $\beta_\theta(t)$ converges to $\pi$. We arrive at the claimed convergence
\begin{equation}
\lim_{\theta\to 0^+}\beta_\theta(t)=\begin{cases}
\beta_0&\text{for $t=0$,}\\
\pi&\text{for $t>0$.}
\end{cases}
\label{eq:limit_angle_oblique}
\end{equation}
When $\theta<0$ the limit is
$$\lim_{\theta\to 0^-}\beta_\theta(t)=\begin{cases}
\beta_0&\text{for $t=0$,}\\
0&\text{for $t>0$.}
\end{cases}
$$
This was all for $0<\beta_0<\pi$. If we start with $\beta_0=0$ or $\pi$ the analysis is similar but we have to determine whether the zero $R_\theta$ is positive or negative and this depends on the sign of $\rho_2-\rho_1$. Finally, if $\beta_0\in\{0,\pi\}$ and $\rho_1=1/2$ then we are at an eigenfunction of $H$ and nothing happens.

Our conclusion is that, in the limit $\theta\to0$, the solution $\psi_\theta(t)$ to the oblique principle goes extremely fast to one of the two eigenfunctions of $H+v_1\cO_1$, which are the two solutions to the variational principle. This happens for all initial conditions $\psi(0)$. The oblique principle picks a different state depending on the sign of $\theta$ (and possibly that of $\rho_2-\rho_1$). On the other hand, the potentials $u^\theta_m(t)$ can be seen to converge to Dirac delta's in the limit. These delta potentials are here to modify the initial condition and replace it by the desired eigenfunction. Indeed, we can express $u^\theta(t)$ in terms of the derivative $\beta_\theta'$ as follows
$$
\left\{
\begin{aligned}
u^\theta_1(t)
&=\frac{\rho_2}{(\rho_1+\rho_2)\cos(\theta)}\left(\beta_\theta'(t)+\cos(\beta_\theta(t))\frac{\rho_2-\rho_1}{\sqrt{\rho_2\rho_1}}\right)\\
u^\theta_2(t)
&=\frac{-\rho_1}{(\rho_1+\rho_2)\cos(\theta)}\left(\beta_\theta'(t)+\cos(\beta_\theta(t))\frac{\rho_2-\rho_1}{\sqrt{\rho_2\rho_1}}\right)
\end{aligned}
\right.
$$
and therefore we obtain
$$
\left\{
\begin{aligned}
\dps\lim_{\theta\to0^+}u^\theta_1(t) & = \frac{\rho_2}{\rho_1+\rho_2}(\pi-\beta_0)\delta_0(t)-\frac{\rho_2-\rho_1}{\sqrt{\rho_1\rho_2}}\frac{\rho_2}{\rho_1+\rho_2}\\
\dps\lim_{\theta\to0^+}u^\theta_2(t) &= \frac{\rho_1}{\rho_1+\rho_2}(\beta_0-\pi)\delta_0(t) +\frac{\rho_2-\rho_1}{\sqrt{\rho_1\rho_2}}\frac{\rho_1}{\rho_1+\rho_2}    
\end{aligned}
\right.
$$
by~\eqref{eq:limit_angle_oblique} for $0<\beta_0<\pi$. We get the expected Dirac delta, together with constant potentials. The latter are equal to the stationary solution $v_1[\rho_1]=\sqrt{\rho_2/\rho_1}-\sqrt{\rho_1/\rho_2}$ from~\eqref{eq:qubit_v_1}, up to a global constant $v_1[\rho_1]\rho_1/(\rho_1+\rho_2)$. This shift only introduces an additional phase in the state.

\section{Proof of Theorem~\ref{thm:GS_K} on the invertibility of \texorpdfstring{$K^\Psi$}{KPsi} for ground states}\label{app:proof_matrix_K}

In this appendix we provide the proof of Theorem~\ref{thm:GS_K} that states that when $\Psi$ is a non-degenerate ground state of $\bH_U+\cV$ and $S^{\Psi}$ is invertible, then $K^{\Psi}$ must be invertible on the orthogonal of the constant potentials.

Let us consider an arbitrary potential $v_m$ and again the operator $\cV:=\sum_{m=1}^Mv_m\cN_m$. We note that
\begin{align*}
\pscal{v,K^{\Psi}v}_{\R^M}&=\sum_{m,n=1}^Mv_mv_n K^{\Psi}_{mn}
= \frac12\pscal{\Psi,[\cV,[\bH_U,\cV]]\Psi}.
\end{align*}
The double commutator equals
$$[\cV,[\bH_U,\cV]]=2\cV \bH_U\cV-\cV^2\bH_U-\bH_U\cV^2$$
so that using $\bH_U\Psi=E_0\,\Psi$ with $E_0$ the ground-state energy, we find
\begin{align*}
\pscal{v,K^{\Psi}v}_{\R^M}&=\frac12\pscal{\Psi,[\cV,[\bH_U,\cV]]\Psi}\\
&=\pscal{\cV\Psi,(\bH_U-E_0)\cV\Psi}\\
&\geq g\left(\pscal{\cV\Psi,\cV\Psi}-|\pscal{\Psi,\cV\Psi}|^2\right)\\
&= g\norm{\cV\Psi-\pscal{\Psi,\cV\Psi}\Psi}^2.
\end{align*}
Here $g=E_1-E_0>0$ is the gap above the ground-state energy, that we have assumed to be strictly positive (this is the non-degeneracy of $\Psi$). The right-hand side of the last equation is the norm of the projection of $\cV\Psi$ on the orthogonal to $\Psi$ and it can also be interpreted as the variance of the observable $\cV$. The previous inequality implies that the Hermitian matrix $K^{\Psi}$ only has non-negative expectations, hence a non-negative spectrum. Furthermore, if $v$ is in its kernel then the left-hand side vanishes and so must do the right-hand side. But then we find
$$\cV\Psi-\pscal{\Psi,\cV\Psi}\Psi=\sum_{m=1}^M\!\!\left(v_m-\frac{\pscal{\Psi,\cV\Psi}}{N}\right)\cN_m\Psi=0$$
since $\sum_{m=1}^M\cN_m\Psi=N\Psi$. From~\eqref{eq:Hohenberg--Kohn} we conclude as we wanted that $v\equiv \pscal{\Psi,\cV\Psi}/N$ is constant. In other words, we have proved that the kernel of $K^{\Psi}$ only contains the constant potentials, corresponding to the trivial gauge consisting of adding global phases to the state. This kernel can be removed by erasing $\cN_M$ from the list of constraints, for instance, or by fixing $\sum_{m=1}^Mv_m(\rho_{\Psi})_m=\pscal{\Psi,\cV\Psi}=0$. This concludes the proof of Theorem~\ref{thm:GS_K}.

\section{On the invertibility of \texorpdfstring{$S^\Psi$}{Spsi}}\label{sec:independence}
Let us go back to the case of $N$ spin--$1/2$ fermions hopping on $M$ sites that we studied in Section~\ref{sec:Hubbard}. The whole geometric picture of our work relies on the invertibility of the matrix
$$(S^\Psi)_{mn}=\Re\pscal{\cN_m\Psi,\cN_n\Psi},$$
with $\cN_m$ the number operator at site $m$ defined in~\eqref{eq:cN}. In this appendix we relate the invertibility of $S^\Psi$ to an irreducibility property of the one-particle density matrix $\gamma_\Psi$ and we discuss what we can do when $S^\Psi$ is not invertible.

First, we recall that the two-particle density matrix is defined by
\begin{multline*}
\gamma_\Psi^{(2)}(m_1\sigma_1,m_2\sigma_2;m'_1\sigma'_1,m'_2\sigma'_2)\\:=\pscal{\Psi,a^\dagger_{m_1\sigma_1}a^\dagger_{m_2\sigma_2}a_{m_2'\sigma_2'}a_{m_2\sigma_2}\Psi}
\end{multline*}
whereas the two-particle density is its spin-averaged diagonal:
$$(\rho_\Psi^{(2)})_{m_1,m_2}:=\sum_{\sigma_1,\sigma_2\in\{\uparrow,\downarrow\}}\gamma_\Psi^{(2)}(m_1\sigma_1,m_2\sigma_2;m_1\sigma_1,m_2\sigma_2).$$
Using the definition of $\cN_m$, we can derive a formula for $S^\Psi$ in terms of the one-particle and two-particle densities only
\begin{equation}
(S^\Psi)_{mn}=(\rho_\Psi^{(2)})_{mn}+(\rho_\Psi)_m\delta_{m,n}.
\label{eq:S_in_terms_of_rho_2}
\end{equation}
We also recall that we can write the expectation of $S^\Psi$ for a vector $v\in\R^M$ as
$$\pscal{v,S^\Psi v}_{\R^M}=\norm{\cV\Psi}^2,$$
with $\cV:=\sum_{m=1}^Mv_m\cN_m$.
Hence $S^\Psi$ is invertible if and only if it satisfies the unique $v$-representability property~\eqref{eq:Hohenberg--Kohn}. As a first step we give a simple invertibility criterion in terms of $\gamma_\Psi$ only, that relies on the following concept.

\begin{definition}[Irreducibility of $\gamma$]\label{def:irreducible}
We say that a one-particle density matrix $\gamma$ \emph{acts irreducibly} (or simply \emph{is irreducible}) when no strict subset $J$ of $\{1,...,M\}$ is stabilized by $\gamma$. In other words, we have $[\1_J,\gamma]\neq0$ where $\1_J$ denotes the diagonal matrix so that $(\1_J)_{j\sigma,j\sigma'}=1$ if $j\in J$ and $\sigma=\sigma'$, and 0 otherwise.
\end{definition}

It is equivalent to require that for any strict subset $J\subset \{1,...,M\}$, we can find $m\in J$ and $m'\notin J$ together with $\sigma,\sigma'\in\{\uparrow,\downarrow\}$ such that $\gamma_{m\sigma,m'\sigma'}\neq0$. This means that $\gamma$ should have sufficiently many off-diagonal terms, so that any subset $J$ of the $M$ sites is linked to its complement.

We can also characterize the irreducibility of $\gamma$ in terms of the real symmetric matrix defined by
$$(\tilde{S}^\gamma)_{mn}:=-\frac12\tr\big([\gamma,\delta_m][\gamma,\delta_n]\big)=\sum_{\sigma,\sigma'\in\{\uparrow,\downarrow\}}|\gamma_{m\sigma,n\sigma'}|^2.$$
We recall that $\delta_m$ is the diagonal matrix so that  $(\delta_m)_{j\sigma,j\sigma'}=1$ if $j=m$ and $\sigma=\sigma'$, and 0 otherwise. This is just the orthogonal projection on the $m$th site. The link with the other matrix $S^\Psi$ will become clear later.
Denoting by $\|A\|_{\rm HS}=\sqrt{\tr(A^\dagger A)}$ the Hilbert--Schmidt norm of a matrix $A$, we have
\begin{align}
\pscal{v,\tilde S^\gamma v}_{\R^M}&=\sum_{m,n=1}^Mv_nv_m(\tilde{S}^\gamma)_{mn}\nn\\
&=\frac12\tr \left( (i[v,\gamma])^2 \right)=\frac12\big\|[v,\gamma]\big\|_{\rm HS}^2\label{eq:formula_tilde_S}
\end{align}
where, as usual, we see $v$ as a vector on the first line and as the corresponding diagonal matrix on the second line. This shows that $\tilde S^\gamma$ has a non-negative spectrum. It is clear that the constant potential always belongs to the kernel of $\tilde S^\gamma$:
$$(1,...,1)\in\ker(\tilde S^\gamma).$$
The following theorem says that the irreducibility of $\gamma$ is equivalent to the kernel of $\tilde S^\gamma$ having dimension one.

\begin{theorem}[Irreducibility criterion]
A  one-particle density matrix $\gamma$ is irreducible if and only if $\ker\tilde{S}^\gamma={\rm span}\{(1,...,1)\}$, that is, the kernel has dimension 1.
\end{theorem}

\begin{proof}
If $\gamma$ commutes with $\1_J$ for a non trivial $J$, then $\1_J$ is also in the kernel by~\eqref{eq:formula_tilde_S}. Hence $\ker(\tilde S^\gamma)$ has multiplicity 2 or more. Conversely, if $\gamma$ acts irreducibly, let us consider an arbitrary vector $v$ in the kernel of $\tilde S^\gamma$, which means that $[\gamma,v]=0$.  Let $J$ be the set of indices $m$ so that $v_m=v_1$. Since $\gamma$ commutes with $v$ it must commute with the spectral projection $\1_J=\1(v=v_1)$ (by~\cite[Sec.~4.9]{Lewin-Spectral}) and we conclude that necessarily $\1_J\in\ker(\tilde S^\gamma)$. The irreducibility assumption tells us that $J=\{1,...,M\}$, hence $v$ is constant and the kernel has multiplicity one, as claimed.
\end{proof}

Next we turn to the link with the matrix $S^\Psi$. We can prove that the irreducibility of $\gamma_\Psi$ implies the invertibility of $S^\Psi$ and is even equivalent to it for Slater determinants.

\begin{theorem}[Irreducibility of $\gamma_\Psi$ and invertibility of $S^\Psi$]\label{thm:irreducible}
Let $\Psi$ be a quantum state so that $\gamma_\Psi$ acts irreducibly as in Definition~\ref{def:irreducible}. Then $S^\Psi$ is invertible. If $\Psi$ is a Slater determinant, the two properties are in fact equivalent.
\end{theorem}

\begin{proof}
Assume that we have $\cV\Psi=0$ for some $\cV=\sum_{m=1}^Mv_m\cN_m$. We have $\cV|\Psi\rangle\langle\Psi|=|\Psi\rangle\langle\Psi|\cV=0$ and, after taking the partial trace in the last $N-1$ variables, we obtain that $v\gamma_\Psi=\gamma_\Psi v$, that is, $\gamma_\Psi$ commutes with $v$. In other words $v$ belongs to the kernel of the matrix $\tilde S^{\gamma_\Psi}$. The irreducibility assumption implies that $v$ is constant. Coming back to the equation $\cV\Psi=0$ we see that the constant must be 0 and we have shown that $S^\Psi$ is invertible.

Next we prove the converse for a Slater determinant $\Psi=\Phi$. Using the formula~\eqref{eq:S_in_terms_of_rho_2} and the explicit expression of $\rho^{(2)}_\Phi$ for Slater determinants, we can compute
\begin{align}
\pscal{v,S^\Phi v}&=\left(\sum_{m=1}^M v_m\rho_m\right)^2-\tr(v\gamma_\Phi v\gamma_\Phi)+\sum_{m=1}^M v_m^2\rho_m\nn\\
&=\left(\sum_{m=1}^M v_m\rho_m\right)^2+\big\langle v,\tilde S^{\gamma_\Phi} v\big\rangle.\label{eq:tilde_S_Slater}
\end{align}
This can also be written as
$$S^\Phi=|\rho\rangle\langle\rho|+\tilde{S}^{\gamma_\Phi},$$
namely $S^\Phi$ is a rank-one perturbation of $\tilde S^{\gamma_\Phi}$ for a Slater determinant. In the first line of~\eqref{eq:tilde_S_Slater}, the first two terms on the right-hand side are the direct and exchange terms, respectively. To go to the second line we used that
$$\sum_{m=1}^M v_m^2\rho_m=\tr(v^2\gamma_\Phi)=\tr(v^2\gamma_\Phi^2)$$
since $\gamma_\Phi$ is a projection for Slater determinants. The relation~\eqref{eq:tilde_S_Slater} tells us that a vector $v$ is in the kernel of $S^\Phi$ if and only if the two terms on the right-hand side vanish. Hence, for Slater determinants, we have
$$\ker(S^\Phi)=\ker(\tilde S^{\gamma_\Phi})\cap\rho^\perp$$
where $\rho^\perp$ is the space of $v$'s such that $\pscal{v,\rho}_{\R^M}=\sum_{m=1}^Mv_m\rho_m=0$. The kernel of $\tilde S^{\gamma_\Phi}$ always contains the constant potential, but the later is not in $\rho^\perp$, because $\pscal{1,\rho}_{\R^M}=N\neq0$. Therefore, $\ker(S^\Phi)\neq\{0\}$ implies $\dim\ker(\tilde S^{\gamma_\Phi})\geq2$. This concludes the proof.
\end{proof}

It is perfectly possible that $S^\Psi$ is invertible although $\gamma_\Psi$ is not irreducible, for a correlated state $\Psi$. For instance, take $M=3$ sites and the correlated state 
$$\Psi=\frac1{\sqrt3}\Big(|1\uparrow,1\downarrow,2\uparrow\rangle+|1\uparrow,2\downarrow,3\uparrow\rangle+|1\downarrow,3\uparrow,3\downarrow\rangle\Big),$$
with an obvious notation. Since two of the three above Slater determinants always have two different orbitals, the one-particle density matrix is just the combination of the individual density matrices. One finds
\begin{multline*}
\gamma_\Psi=\frac23\big(|1\uparrow\rangle\langle1\uparrow|+|1\downarrow\rangle\langle1\downarrow|+|3\uparrow\rangle\langle3\uparrow|\big)\\
+\frac13\big(|2\uparrow\rangle\langle2\uparrow|+|2\downarrow\rangle\langle2\downarrow|+|3\downarrow\rangle\langle3\downarrow|\big).
\end{multline*}
It is not at all irreducible because it is diagonal in the canonical basis. In fact, it commutes with all the $\delta_m$ and thus $\tilde S^{\gamma_\Phi}=0$. On the other hand, if we have $\cV\Psi=0$ we obtain the equations
$$
\left\{
\begin{aligned}
2v_1+v_2&=0\\
v_1+v_2+v_3&=0\\
v_1+2v_3&=0\\    
\end{aligned}
\right.
$$
that imply $v_1=v_2=v_3=0$ and therefore that $S^\Psi$ is invertible.

An interesting problem is to understand for what kind of lattice systems the ground state satisfies that $S^\Psi$ or $\tilde S^{\gamma_\Psi}$ are invertible (on the orthogonal of the constant). To our knowledge, only 1D chains have been handled so far, using Perron--Frobenius theory~\cite{PenLeu-21}.

We now address the following question. Imagine that we are given a trajectory of densities $\rho(t)$ and an initial state $\Psi_0$ such that $\rho_{\Psi_0}=\rho(0)$. If $S^{\Psi_0}$ is not invertible, we cannot apply Theorem~\ref{thm:McLachlan} to obtain a unique solution to the geometric principle. Could we slightly perturb $\Psi_0$ and replace it by a closeby $\Psi'_0$ so that $S^{\Psi'_0}$ is invertible and $\rho_{\Psi'_0}=\rho_{\Psi_0}$? In other words, is the set of states such that $S^{\Psi_0}$ is invertible dense within the set of states with given density?

By Theorem~\ref{thm:irreducible}, we know that $\gamma_{\Psi_0}$ is not irreducible. The same result implies that if a normalized wavefunction $\Psi_0'$ is such that $\gamma_{\Psi_0'}$ is irreducible, then $S^{\Psi_0'}$ is invertible. We can thus restrict our attention to the one-particle density matrices. In fact, if we can find a unitary matrix $U$ such that $\gamma'=U\gamma_{\Psi_0}U^\dagger$ is irreducible, then the state $\Psi_0':=U^{\otimes N}\Psi_0$ is such that $\gamma_{\Psi_0'}=\gamma'$ is irreducible. It is therefore natural to ask whether irreducibility can be restored by conjugation by a one-body unitary matrix. Note that if $\Psi_0$ is a Slater determinant and if we want $\Psi'_0$ to be a Slater determinant as well, then $\Psi_0'$ is necessarily of the form $\Psi_0'=U^{\otimes N}\Psi_0$. Indeed, as $\gamma_{\Psi_0}$ and $\gamma_{\Psi'_0}$ are then both rank-$N$ orthogonal projectors, they are unitary equivalent so that there exists a unitary matrix $U \in \C^{M \times M}$ such that $\gamma_{\Psi_0'}=U\gamma_{\Psi_0}U^\dagger$; this implies that $\Psi_0'=U^{\otimes N}\Psi_0$ up to an irrelevant global phase which can be absorbed in the unitary matrix $U$.

\begin{theorem}[Perturbing non-irreducible matrices]\label{thm:reducible_perturb}
Let $\gamma$ be a one-particle density matrix that does not act irreducibly. Upon relabelling the sites we can assume that $[\gamma,\1_J]=0$ with $J=\{1,...,K\}$ where $1\leq K\leq M-1$. Then $\gamma$ takes the block-diagonal form
$$\gamma=\begin{pmatrix}
\gamma_1&0\\
0&\gamma_2\\
\end{pmatrix}.$$
We assume that $\gamma_1$ acts irreducibly on $\{1,...,K\}$ and $\gamma_2$ acts irreducibly on $\{K+1,...,M\}$.

If
\begin{equation}
\max\sigma(\gamma_1)\leq\min\sigma(\gamma_2)\ \text{or}\ \max\sigma(\gamma_2)\leq\min\sigma(\gamma_1),
\label{eq:cond_reducible_extreme}
\end{equation}
then all the $\gamma'=U\gamma U^\dagger$ with $U$ a unitary matrix such that $\rho_{\gamma'}=\rho_\gamma$, are of the same block-diagonal form as $\gamma$, that is, commute with $\1_J$. Therefore, it is \textbf{not possible} to replace $\gamma$ with a unitarily equivalent matrix of the same density acting irreducibly.

In contrast, if
\begin{equation}
\min\sigma(\gamma_1)<\max\sigma(\gamma_2)\ \text{and}\ \min\sigma(\gamma_2)<\max\sigma(\gamma_1)
\label{eq:cond_reducible_interior}
\end{equation}
then for any $\eps>0$ we can find a unitary matrix $U_\eps$ such that $\|1-U_\eps\|\leq \eps$ and $\gamma'=U_\eps\gamma U_\eps^\dagger$ \textbf{acts irreducibly}, with $\rho_{\gamma'}=\rho_\gamma$. In particular, if $\gamma$ is the one-particle density matrix of some state $\Psi$, we obtain a state $\Psi'=U^{\otimes N}\Psi$ for which $S^{\Psi'}$ is invertible, by Theorem~\ref{thm:irreducible}.
\end{theorem}

For simplicity, we have stated the result for two irreducible blocks. If there are more, one argues by induction.

The condition~\eqref{eq:cond_reducible_extreme} means that the two spectra of $\gamma_1$ and $\gamma_2$ are ordered in the sense that one of them is completely above the other. Note that $\max\sigma(\gamma_1)\leq\min\sigma(\gamma_2)$ can be detected from the equivalent property that $\sum_{m=1}^K\rho_m=\sum_{k=1}^{2K}\lambda_k(\gamma)$, that is, the average number of electrons in $J$ is exactly equal to the sum of the $2K$ lowest eigenvalues $\lambda_k(\gamma)$ of $\gamma$ (the factor 2 is because of spin). For the reader familiar with the Schur--Horn theorem, we notice that this characterizes the faces of the corresponding convex polytope~\cite{LeiRicTom-99}. Although we are not going to use this result here, the proof of Theorem~\ref{thm:reducible_perturb} is somewhat inspired by the Schur--Horn theorem.

\begin{proof}
First, we assume, for instance, $\max\sigma(\gamma_1)\leq\min\sigma(\gamma_2)$ and consider any $\gamma'=U\gamma U^\dagger$ such that $\rho_{\gamma'}=\rho_\gamma$. We write
$$\gamma'=\begin{pmatrix}
\gamma'_1&\gamma'_{12}\\
\gamma'_{21}&\gamma'_2\\
\end{pmatrix},$$
where we recall that the block decomposition corresponds to the splitting $\{1,...,M\}=\{1,...,K\}\cup \{K+1,...,M\}$. We have
$$\tr(\gamma'_1)=\tr(\gamma_1)=\sum_{k=1}^{2K}\lambda_k(\gamma)=\sum_{k=1}^{2K}\lambda_k(\gamma').$$
The first part is because $\gamma$ and $\gamma'$ have the same density and the second is because they have the same spectrum. Next we use the variational characterization of the sum of the lowest eigenvalues, which is nothing but the Aufbau principle for fermions. We denote by $P:=\1_{}(\gamma'\leq\lambda_{2K})$ the spectral projection corresponding to the lowest $2K$ eigenvalues (the latter can have a rank $\geq2K$ in case of degeneracies). Next a simple calculation from~\cite{BacBarHelSie-99,HaiLewSer-05a} shows that
$$0=\tr\big((\gamma'-\lambda_{2K})(\1_J-P)\big)=\tr\big(|\gamma'-\lambda_{2K}|(\1_J-P)^2\big).$$
Let $P_<=\1(\gamma'<\lambda_{2K})$ be the projection corresponding to the eigenvalues strictly below $\lambda_{2K}$. The above relation implies $P_<(\1_J-P)P_<=P_<(\1_J-1)P_<=0$. Hence $\1_J=P_<+\delta$ for a $\delta$ supported on the kernel of $P_<$. Similarly, we find $P_>(\1_J-P)P_>=P_>\1_JP_>=0$ where $P_>$ is the projection on the eigenvalues strictly above $\lambda_{2K}$. Our conclusion is that $\delta$ is an orthogonal projection whose range is included in $\ker(\gamma'-\lambda_{2K})$. This shows that $\1_J$ is a spectral projection of $\gamma'$, hence commutes with it. This proves the claim that $\gamma'$ has the same block-diagonal structure as $\gamma$, hence is not irreducible.

Next we come to the second part of the theorem. We assume that the two spectra are interlaced. We call $\lambda_1$ and $\lambda_1'$ the lowest and largest eigenvalues of $\gamma_1$, with eigenvectors $\phi_1$ and $\phi_1'$. Similarly, we call $\lambda_2$ and $\lambda_2'$ the lowest and largest eigenvalues of $\gamma_2$, with eigenvectors $\phi_2$ and $\phi_2'$. We thus have $\lambda_1<\lambda_2'$ and $\lambda_1'<\lambda_2$. It is perfectly possible that $\lambda_1=\lambda_1'$ in case of degeneracy. We can always assume that $\phi_1$ and $\phi_1'$ are orthogonal, even when $K=1$, thanks to the spin. Similarly for $\phi_2$ and $\phi_2'$. Next we apply two rotations, replacing
$$(\phi_1,\phi_2')\mapsto(\cos\theta_1\phi_1+\sin\theta_1\phi_2',-\sin\theta_1\phi_1+\cos\theta_1\phi_2')$$
and
$$(\phi_1',\phi_2)\mapsto(\cos\theta_2\phi'_1+\sin\theta_2\phi_2,-\sin\theta_2\phi'_1+\cos\theta_2\phi_2).$$
By doing so we insert non-zero terms outside of the blocks, as needed to make the matrix irreducible. We obtain a new matrix $\gamma''$ whose diagonal blocks are
\begin{multline*}
\gamma''_1=\gamma_1+(\lambda_2'-\lambda_1)\sin^2\theta_1\,|\phi_1\rangle\langle\phi_1|\\
-(\lambda_1'-\lambda_2)\sin^2\theta_2\,|\phi'_1\rangle\langle\phi'_1|
\end{multline*}
and
\begin{multline*}
\gamma''_2=\gamma_2+(\lambda'_1-\lambda_2)\sin^2\theta_2\,|\phi_2\rangle\langle\phi_2|\\
-(\lambda_1'-\lambda_2)\sin^2\theta_2\,|\phi'_2\rangle\langle\phi'_2|.
\end{multline*}
We choose $\theta_1=\eps\ll1$ and $\theta_2=a\eps$ with $a=a(\eps)$ chosen so that
$$(\lambda_2'-\lambda_1)\sin^2\eps=(\lambda_1'-\lambda_2)\sin^2(a\eps).$$
to ensure $\tr(\gamma_1'')=\tr(\gamma_1)$ and $\tr(\gamma_2'')=\tr(\gamma_2)$. In other words $a\sim \sqrt{(\lambda_2'-\lambda_1)/(\lambda_1'-\lambda_2)}$ in the limit $\eps\to0$.

After the rotation, we have an irreducible matrix for $\eps\ll1$ but we have modified the density. The last step consists of applying a block-diagonal rotation in the form
$$U=\begin{pmatrix}
U_1'&0\\
0&U_2'\end{pmatrix}$$
to ensure that $U_1'\gamma_1''(U_1')^\dagger$ has the same density as $\gamma_1$ and $U_2'\gamma_2''(U_2')^\dagger$ has the same density as $\gamma_2$. The off-diagonal blocks get multiplied by $U'_1$ and $U'_2$ but they stay non -zero. Such $U'_1$ and $U'_2$ exist for $\eps\ll1$ because $\gamma_1$ and $\gamma_2$ are irreducible, hence $\gamma_1''$ and $\gamma''_2$ also, for $\eps\ll1$. For instance we can solve the time-dependent equation
$$\partial_t\gamma_1=[[w(t),\gamma_1(t)],\gamma_1(t)],\qquad \gamma_1(0)=\gamma_1''$$
with $w(t)$ satisfying $\sum_{m=1}^Mw_m(t)=0$ and chosen so as to reproduce the density $\rho(t)=(1-t)\rho_{\gamma_1''}+t\rho_{\gamma_1}$ for $0\leq t\leq1$. Then $U'_1$ is the value at time $t=1$ of the solution to $\partial_tU_1(t)=[w(t),\gamma_1(t)]U_1(t)$ and $\gamma'_1=\gamma_1(1)$. The existence of a solution follows from the same arguments as for Theorem~\ref{thm:McLachlan} and the mixed state case in Appendix~\ref{app:mixed}. We only have to verify that the solution exists until the time $t=1$. This follows from the proof of the Cauchy--Lipschitz theorem. Namely, the perturbation of $\gamma_1$ and of the desired density is of order $\eps^2\ll1$ hence the existence time given by the Banach fixed point theorem used in the proof of Cauchy--Lipschitz is at least of order $1/\eps^2$. The argument is the same for $\gamma'_2$. Putting everything together, we have constructed the desired unitary matrix.
\end{proof}

To conclude this appendix, we extract a result dealing specifically with projections.

\begin{theorem}[Perturbing non-irreducible projections]\label{cor:reducible_perturb_proj}
Let $\gamma$ be a one-particle rank-$N$ projection that does not act irreducibly. We assume that $0<\rho_\gamma<2$ everywhere. Then, for any $\eps>0$ we can find a unitary matrix $U_\eps$ such that $\|1-U_\eps\|\leq \eps$ and $\gamma'=U_\eps\gamma U_\eps^\dagger$ \textbf{acts irreducibly} with $\rho_{\gamma'}=\rho_\gamma$. In particular, for the corresponding Slater determinants, although we had $\det(S^\Phi)=0$ we obtain $\det(S^{\Phi'})\neq0$ for $\Phi'=U^{\otimes N}\Phi$.
\end{theorem}

This result says that in Kohn--Sham theory, if the sites are never empty or full, one can always perturb an initial Slater determinant $\Phi_0$ into a new Slater determinant with the same density, for which the geometric dynamics is well-posed for some time.

\begin{proof}
We decompose $\gamma$ into irreducible blocks, that is, we write $\{1,...,M\}=\cup J_k$ (disjoint union) with $[\gamma,\1_{J_k}]=0$ and $\gamma_k:=\1_{J_k}\gamma\1_{J_k}$ acting irreducibly on $J_k$. The spectrum of $\gamma$ is the union of the spectra of the $\gamma_k$'s and it only contains $0$'s and $1$'s. If one $\gamma_k=0$, then we have $\rho=0$ on $J_k$, which contradicts our assumption that $0<\rho_\gamma<2$. Similarly, if $\gamma_k=1$ on $J_k$ then we have $\rho\equiv2$ on $J_k$, which is also impossible. Our conclusion is that the spectra of the blocks must all contain both $0$'s and $1$'s. We can thus apply Theorem~\ref{thm:reducible_perturb} inductively and obtain an irreducible rank-$N$ projection close to $\gamma$.
\end{proof}

\section{Geometric principle for mixed states}\label{app:mixed}

In this appendix we provide the mixed state version of the geometric principle introduced in Section~\ref{sec:McLachlan}. Mixed states are more complicated objects than pure states. For pure states we only need to impose the normalization condition $\|\psi\|^2=1$ and can simply put the latter in the list of constraints, by requiring that the identity matrix $I_d$ belongs to $\text{span}_\R(\cO_1,...,\cO_M)$. For mixed states we have to require that $\Gamma$ is a Hermitian matrix satisfying $\tr(\Gamma)=1$ and $\sigma(\Gamma)\subset[0,1]$ (where $\sigma(\Gamma)$ designates the spectrum of $\Gamma$). These constraints are not so easy to handle.

To deal with this difficulty., we restrict our attention to time-dependent Schrödinger equations whose solution is evolving on the orbit 
$$\text{Orb}(\Gamma_0):=\left\{U\Gamma_0U^\dagger,\ U\in U(d)\right\}$$ 
of the initial state $\Gamma_0$, under the action of the unitary group. In other words, we require that the evolved state $\Gamma(t)=U(t)\Gamma_0U(t)^\dagger$ is unitarily equivalent to the initial state for all times. We recall that $\text{Orb}(\Gamma_0)$ forms a manifold, whose dimension depends on the spectrum of $\Gamma_0$. At any $\Gamma\in \text{Orb}(\Gamma_0)$, the manifold is locally parametrized by
$$e^{iT}\Gamma e^{-iT}=\Gamma+i[T,\Gamma]+O(T^2)$$ 
with $T^\dagger=T$. The tangent space at $\Gamma$ therefore consists of the operators of the form $i[T,\Gamma]$ with $T^\dagger=T$. The time-dependent equations must therefore take the von Neumann form $i\partial_t\Gamma(t)=[T(t),\Gamma(t)]$ for some $T(t)^\dagger=T(t)$, in such a way that the velocity always belongs to the tangent space, hence $\Gamma(t)$ belongs to $\text{Orb}(\Gamma_0)$. 

Next we discuss how to add constraints in this framework, in the form $\tr(\cO_m\Gamma(t))=o_m(t)$ for some Hermitian matrices $\cO_m$. This amounts to working in a sub-manifold of the orbit $\text{Orb}(\Gamma_0)$. The previous discussion leads us to look for a modified von Neumann equation in the special form
\begin{equation}
i\partial_t\Gamma(t)=\big[H(t)+G\big(t,\Gamma(t)\big)\,,\,\Gamma(t)\big]
\label{eq:GP_mixed_pre}
\end{equation}
where $G\big(t,\Gamma(t)\big)$ is a Hermitian operator used to impose the constraints and $H(t)$ is the time-dependent Hamiltonian. This is the mixed state version of the abstract modified Schrödinger equation~\eqref{eq:Schrodinger_perturbed}.

We explained in~\eqref{eq:McLachlan_rank2}-\eqref{eq:geom_pert_pure} that for the Geometric Principle, the pure state equation \eqref{eq:McLachlan} can indeed be written in the form~\eqref{eq:GP_mixed_pre} with
\begin{equation}
G\big(t,\Gamma(t)\big):=i\sum_{m=1}^Mw_m(t)[\cO_m, \Gamma(t)].
\label{eq:form_G_mixed}
\end{equation}
We claim this is the general form for mixed states too. 

To explain this, we again first restrict ourselves to time-independent constraints $o_m(t)\equiv o_m$, so that we are working in a fixed sub-set of the orbit $\text{Orb}(\Gamma_0)$, denoted by
$$\cC_{\Gamma_0}:=\left\{\Gamma=U\Gamma_0U^\dagger\ :\ \tr(\Gamma\cO_m)=o_m,\ m=1,...,M\right\}.$$
This set can again be decomposed into a smooth part $\cM_{\Gamma_0}$ and a singular part  $\cS_{\Gamma_0}$, which are properly defined below. The tangent space at a $\Gamma$ must consists of the $i[T,\Gamma]$ so that 
the expectation values of the observables $\cO_m$ do not change to leading order, leading to the condition that $\tr(\cO_m i[T,\Gamma])=0$ for all $m=1,...,M$. Noticing that $\tr(\cO_m i[T,\Gamma])=-\tr(i[\cO_m,\Gamma]T)$ we conclude that $T$ must be orthogonal to the operators $i[\cO_m,\Gamma]$ for the real Hilbert--Schmidt scalar product $\pscal{A,B}_\text{HS}=\tr(AB)$ of Hermitian matrices. We thus introduce the real linear space
(which one should note is \emph{not} the tangent space of $\cM_{\Gamma_0}$)
\begin{align*}
\cT_\Gamma 
    & :=\big\{T=T^\dagger,\ \pscal{T,i[\cO_m ,\Gamma]}_\text{HS}=0,\ 1\leq m\leq M\big\}
\intertext{and its orthogonal complement}
\cN_\Gamma
    &:=\text{span}_\R\big(i[\cO_1,\Gamma],...,i[\cO_M,\Gamma)\big)
    \\
    &=\left\{i\sum_{m=1}^Mw_m[\cO_m,\Gamma],\, \ w_1,...,w_M\in\R\right\}.
\end{align*}
The structure is therefore similar to the pure state case, if we use the Hilbert--Schmidt scalar product of matrices. We can now properly define the regular part $\cM_{\Gamma_0}$ of the set of constrained states as the $\Gamma\in\text{Orb}(\Gamma_0)$ so that the operators $i[\cO_m,\Gamma]$ are $\R$-linearly independent. This can be reformulated by requiring the matrix
\begin{equation}
\begin{aligned}
(\Sigma^\Gamma)_{mn}   
    & :=\frac12\pscal{i[\cO_m,\Gamma],i[\cO_n,\Gamma]}_\text{HS}
    \\ & =\frac12\tr\left([[\cO_m,\Gamma],\Gamma]\cO_n\right)
\end{aligned}
\label{eq:M_mixed}
\end{equation}
to be invertible, i.e. $\det(\Sigma^\Gamma)\neq0$. The singular set $\cS_{\Gamma_0}=\cC_{\Gamma_0}\setminus\cM_{\Gamma_0}$ is composed of the $\Gamma\in\text{Orb}(\Gamma_0)$ so that $\det(\Sigma^\gamma)=0$. The matrix $\Sigma^\Gamma$ plays the same role as the matrix $S^\psi$ that we had for pure states. Of course since we have restricted the dynamics to the orbit $\text{Orb}(\Gamma_0)$ we are here always working with states by definition  and we do not require anymore that $I_d\in\text{span}(\cO_1,...,\cO_M)$. (Otherwise $\Sigma^\Gamma$ cannot be invertible; Indeed, if $I_d=\sum_{m=1}^Mc_m\cO_m$ for some $c_m$'s then $\sum_{m=1}^Mc_m[\cO_m,\Gamma]=[I_d,\Gamma]=0$ because the identity matrix commutes with all states $\Gamma$.)

All in all, this leads to the statement that, in the geometric principle, $H(t) + G(t,\Gamma(t))$ is the projection of $H(t)$ onto the space $\cT_{\Gamma(t)}$ for the Hilbert-Schmidt scalar product. 
(This follows from the commutator form of \eqref{eq:GP_mixed_pre}, whence one should project onto the space $\cT_\Gamma$ and \emph{not} onto the tangent space, which consists of commutators $i[T,\Gamma]$ with $T\in \cT_\Gamma$.) 
That is, $G(t,\Gamma(t)) \in \cN_{\Gamma(t)}$ is of the form 
\begin{equation*}
    G(t,\Gamma(t)) = i\sum_{m=1}^Mw_m(t)[\cO_m,\Gamma(t)]
\end{equation*}
for some real-valued $w_m(t)$, as was claimed above. 

The argument is similar in the case of time-dependent constraints and our conclusion is that the \textbf{geometric principle for mixed states} reads
\begin{equation}
\boxed{i\partial_t\Gamma(t)=\left[H(t)+i\sum_{m=1}^Mw_m(t)[\cO_m,\Gamma(t)]\,,\,\Gamma(t)\right]}
\label{eq:GP_mixed_2}
\end{equation}
with real numbers $w_m(t)$ to be determined so as to fulfill the desired time-dependent constraints. 
Existence and uniqueness of the $w_m(t)$'s is proved in the same way as in Theorem~\ref{thm:McLachlan}, under the assumption that $\Sigma^{\Gamma_0}$ is invertible. 

Finally, we remark that, for $H(t)\equiv0$, Equation~\eqref{eq:GP_mixed_2} resembles the double-bracket flow, sometimes used to diagonalize matrices~\cite{BloBroRat-92,BacBru-10}, with the difference that the coefficients $w_m$ are in our case nonlinear functions of $\Gamma(t)$.
This can be thought of as a mixed state version of the resemblance with static DFT for imaginary time in the pure state case \cite{PenLee-PRA-25} mentioned in \Cref{sec:geomectric-structure-constraints} above.

\section{The algebraic viewpoint and other choices of the correction term}
\label{sec.algebra.corr.term}

In this appendix, we provide an algebraic viewpoint on the choice of the correction term $G$ in Eq.~\eqref{eq:McLachlan_rank2} as a linear combination of simple operators belonging to the Lie algebra generated by the observables, the Hamiltonian and the density matrix representing the state. This viewpoint encompasses the specific choices presented in the main body of the paper.

For convenience, we work in the mixed state formalism of Appendix~\ref{app:mixed}  and consider the von Neumann equation
\begin{align} \label{eq:Liouville_equation}
   i \partial_t \Gamma(t) & = [H(t)+G\big(t,\Gamma(t)\big),\Gamma(t)], \\
   \Gamma(0) & =\Gamma_0. \label{eq:initial_condition}
\end{align}
When $\Gamma_0=|\psi_0\rangle\langle\psi_0|$, this is equivalent to the modified Schr\"odinger equation~\eqref{eq:McLachlan_rank2} with initial condition $\psi(0)=\psi_0$. We assume that the constraints are given by 
\begin{equation}
    \label{eq:TD_constraints}
    \tr(\mathcal O_m(t)\Gamma(t))=o_m(t), \quad 1 \le m \le M, \quad t \ge 0,
\end{equation}
where $\mathcal O_m : [0,+\infty) \to \C^{d \times d}_{\rm herm}$ and $o_m : [0,+\infty) \to \R$ are continuously differentiable functions such that $o_m(0)=\tr(\mathcal O_m(0)\Gamma_0)$ for all $m=1,..., M$. In contrast with the formalism used in the main body of the paper, we allow here time-dependent observables, for later purposes. 

Differentiating \eqref{eq:TD_constraints} in time, we obtain that a necessary and sufficient condition for a solution of \eqref{eq:Liouville_equation} to satisfy the constraints~\eqref{eq:TD_constraints} is
\begin{equation}
    \label{eq:CNS_1}
    \tr(i[G(t,\Gamma(t)),\mathcal O_m(t)]) = b_m(t,\Gamma(t)),
\end{equation}
with
$$
b_m(t,\Gamma):=o_m'(t)-\tr\left( \left({\mathcal O}_m'(t)+i[H(t),\mathcal O_m(t)]\right)\Gamma\right).
$$
It is natural to choose $G$ of the form
\begin{equation}
    \label{eq:special_form_F}
    G(t,\Gamma)=\sum_{j=1}^J \alpha_j(t) \mathcal B_j(t,\Gamma),
\end{equation}
where $\alpha_j : [0,+\infty) \to \R$ and $\mathcal B_j : [0,+\infty) \times \C^{d \times d}_{\rm herm} \to \C^{d \times d}_{\rm herm}$ are continuously differentiable functions. 
The various approaches considered in the theoretical sections above correspond to various choices of operators $\mathcal B_j(t,\Gamma)$:
\begin{description}
    \item[{[VP]}] the variational principle corresponds to $J=M$ and $\mathcal B_j(t)=\mathcal O_j(t)$;
    \item[{[GP]}] the geometric principle  corresponds to $J=M$ and $\mathcal B_j(t,\Gamma)=i[\mathcal O_j(t),\Gamma]$;
    \item[{[OP]}] the oblique principle corresponds to $J=M$ and $\mathcal B_j(t,\Gamma)=\cos \theta \; \mathcal O_j(t) +\sin \theta \; i [\mathcal O_j(t),\Gamma]$.
\end{description}
More generally, it is natural to choose the $\mathcal B_j(t,\Gamma)$'s in the Lie algebra generated by the operators $\mathcal O_j(t)$, $\Gamma$, and $H(t)$. 
This leads naturally to also study
\begin{description}
    \item[{[CP]}] the {\em current principle}, corresponding to $J=M$ and $\mathcal B_j(t):=-i[H(t),\mathcal O_j(t)]$.
\end{description}
Further, a natural generalization of the oblique principle is to consider any matrix-type interpolation between the geometric and variational principle as follows
\begin{description}
    \item[{[gOP]}] the \emph{generalized oblique principle}, corresponding to $J=M$ and, for some real $M\times M$ matrices $\Theta^{(1)}(t)$ and $\Theta^{(2)}(t)$, 
    $$\mathcal B_j(t,\Gamma) := \sum_{k=1}^M \Bigl(\Theta^{(1)}_{jk}(t) \; \mathcal O_k(t) + \Theta^{(2)}_{jk}(t) \; i [\mathcal O_k(t),\Gamma] \Bigr).$$ 
\end{description}
For $G$ of the form \eqref{eq:special_form_F}, the necessary and sufficient condition \eqref{eq:CNS_1} reads
\begin{equation}
    \label{eq:CNS_2}
    \mathfrak M(t,\Gamma(t)) \alpha(t) = b(t,\Gamma(t)),
\end{equation}
where $\alpha(t) := (\alpha_1(t), \cdots,\alpha_J(t)) \in \R^J$, $b(t,\Gamma) := (b_1(t,\Gamma),\cdots,b_M(t,\Gamma)) \in \R^M$ and $\mathfrak M : [0,+\infty) \times \C^{d \times d}_{\rm herm} \to \R^{M \times J}$ is the matrix-valued function defined by
\begin{equation}
    [\mathfrak M(t,\Gamma)]_{mj} = \tr(i [\mathcal B_j(t,\Gamma),\mathcal O_m(t)]\Gamma).
\end{equation}
For each of the choices [VP], [GP], [OP], [CP], and [gOP], $\mathfrak M(t,\Gamma)$ is a square matrix. Assuming that \eqref{eq:Liouville_equation}-\eqref{eq:initial_condition}, with $G$ given by \eqref{eq:special_form_F} has a solution $\Gamma(\cdot)$ on the interval $[0,T)$, $T > 0$, satisfying \eqref{eq:TD_constraints}, and assuming that $\mathfrak M(t,\Gamma(t))$ is invertible at each $t \in [0,T)$, we obtain, by combining~\eqref{eq:Liouville_equation}, \eqref{eq:initial_condition}, \eqref{eq:special_form_F}, and \eqref{eq:CNS_2}, that $\Gamma(t)$ is a solution on $[0,T)$ to the Cauchy problem
\begin{align}
\label{eq:Cauchy_problem}
   i \partial_t \Gamma(t) & = \bigl[H(t) 
   +\mathfrak M(t,\Gamma(t))^{-1}  b(t,\Gamma(t)) \cdot \mathcal B(t),\Gamma(t)\bigr], \\
   \Gamma(0) & =\Gamma_0, \label{eq:initial_condition_2}
\end{align}
where $\mathcal B(t):=(\mathcal B_1(t),\cdots,\mathcal B_M(t))$ is a vector-valued operator. If $\mathfrak M(0,\Gamma_0)$ is invertible, it follows from the Cauchy--Lipschitz  theorem that ~\eqref{eq:Cauchy_problem}-\eqref{eq:initial_condition_2} has a unique maximal continuously differentiable solution on a time-interval $[0,T_*)$ with either $T_*=+\infty$ or $\displaystyle \lim_{t \to T_*}\mbox{det}(\mathfrak M(t,\Gamma(t)))=0$, and that, on the time interval $[0,T_*)$, $\Gamma(t)$ is the unique solution to~\eqref{eq:Liouville_equation}-\eqref{eq:initial_condition} satisfying \eqref{eq:TD_constraints}.

Depending on the setting under consideration, the matrix $\mathfrak M(t,\Gamma)$ has the following expression
\begin{description}
    \item[{[VP]}] 
    $[\mathfrak M^{\rm VP}(t,\Gamma)]_{mn} = \tr( i[\mathcal O_n(t),\mathcal O_m(t)] \Gamma)$. In particular, $\mathfrak M^{\rm VP}(t,\Gamma) \in \R^{M \times M}_{\rm antisym}$;
    \item[{[GP]}] $[\mathfrak M^{\rm GP}(t,\Gamma)]_{mn} = \tr( (i[\mathcal O_m(t),\Gamma]) (i[\mathcal O_n(t),\Gamma]))$. In particular, $\mathfrak M^{\rm GP}(t,\Gamma) \in \R^{M \times M}_{\rm sym}$;
    \item[{[OP]}] $\mathfrak M^{\rm OP}(t,\Gamma) = \cos \theta \; \mathfrak M^{\rm VP}(t,\Gamma) + \sin \theta \; \mathfrak M^{\rm GP}(t,\Gamma)$;
    \item[{[CP]}] $[\mathfrak M^{\rm CP}(t,\Gamma)]_{mn} =\tr(  [[H(t),\mathcal O_n(t)],\mathcal O_m(t)]\Gamma)$;
    \item[{[gOP]}] 
    $\mathfrak{M}^{\mathrm{gOP}}(t,\Gamma)$\\ $\phantom{xxx} =\mathfrak{M}^{\rm VP}(t,\Gamma) \Theta^{(1)}(t)^T +\mathfrak{M}^{\rm GP}(t,\Gamma) \Theta^{(2)}(t)^T$.
\end{description}
These matrices coincide with the ones considered in Sections~\ref{sec:TDVP}, \ref{sec:McLachlan}, and Appendix~\ref{app:mixed}, namely
\begin{align*}
\mathfrak M^{\rm VP}(t,|\Psi\rangle\langle\Psi|) &= 2 A^\Psi(t), \\ \mathfrak M^{\rm GP}(t,\ket{\Psi}\bra{\Psi}) &= 2\Sigma^{\ket{\Psi}\bra{\Psi}}(t), \\ \mathfrak M^{\rm CP}(t,|\Psi\rangle\langle\Psi|)& =-2K^\Psi(t). 
\end{align*}
Recall that $\Sigma^{\ket{\Psi}\bra{\Psi}}$ is the analogue of the matrix $S^\Psi$ in the density-matrix formalism.

We note that for any $\alpha\in \R^M$ we have
$$
\alpha^T \mathfrak M^{\rm GP}(t,\Gamma) \alpha = \left\| i \left[ \sum_{m=1}^M \alpha_m \mathcal O_m(t),\Gamma\right] \right\|_{\rm HS}^2. 
$$
Thus, $\mathfrak M^{\rm GP}(0,\Gamma_0)$ is not invertible if and only if $\Gamma_0$ commutes with some non-trivial linear combination of the observables $\mathcal O_m(0)$. We also see that if $M$ is odd, $\mathfrak M^{\rm VP}(0,\Gamma)$ is never invertible, since it is an antisymmetric matrix of odd order.

\medskip

\paragraph*{Commuting observables and the van Leeuwen equation.}
More can be said about the variational principle (i.e. $J=M$ and $B_m(t)=\mathcal O_m(t)$ for all $1 \le m \le M$), in the special case when all the observables $\mathcal O_m(t)$ commute (i.e. $[\mathcal O_m(t),\mathcal O_n(t)]=0$ for all $1 \le m,n \le M$). Then, $\mathfrak M^{\rm VP}(t,\Gamma)=0$ for all $t$ and $\Gamma$ and the necessary and sufficient condition~\eqref{eq:CNS_2} reads
\begin{equation} \label{eq:constraint_derivatives}
    \tr\left( \left( {\mathcal O}_m'(t)+i[H(t),\mathcal O_m(t)]\right)\Gamma(t)\right)= o_m'(t).
\end{equation}
If $\tr\left( \left({\mathcal O}_m'(0)+i[H(0),\mathcal O_m(0)]\right)\Gamma_0\right) \neq o_m'(0)$ the equations for the variational principle have no solution. If the condition  
\begin{equation}\label{eq:new_condition}
\tr\left( \left( {\mathcal O}'(0)+i[H(0),\mathcal O_m(0)]\right)\Gamma_0\right) = o_m'(0)
\end{equation}
is satisfied, we can replace the original set of constraints \eqref{eq:TD_constraints} by the new set of constraints~\eqref{eq:constraint_derivatives} since the set of equations \eqref{eq:Liouville_equation}, \eqref{eq:initial_condition}, \eqref{eq:TD_constraints} is equivalent to the set of equations~\eqref{eq:Liouville_equation}, \eqref{eq:initial_condition}, \eqref{eq:constraint_derivatives} if condition \eqref{eq:new_condition} is fulfilled. This amounts to taking $B_m(t)=\mathcal O_m(t)$ and the new set of observables and expectations values 
$$
\widetilde{\mathcal O}_m(t)= {\mathcal O}_m'(t)+i[H(t),\mathcal O_m(t)] \quad \mbox{and} \quad \widetilde o_m(t) = o_m'(t).
$$
The new necessary and sufficient condition reads
\begin{equation} \label{eq:CNS_3}
    \widetilde{\mathfrak M}^{\rm VP}(t,\Gamma(t)) \alpha(t) = \widetilde b(t,\Gamma(t)),
\end{equation}
with
\begin{align*}
   & [\widetilde{\mathfrak M}^{\rm VP}(t,\Gamma)]_{mn} := \tr\left( i [\mathcal O_n,\widetilde{\mathcal O}_m(t)]\Gamma \right) \nonumber \\
       &   = \tr \left( i\left[\mathcal O_n(t), {\mathcal O}_m'(t)\right] \Gamma \right)  +  \tr \left( [[H(t),\mathcal O_m(t)],\mathcal O_n(t)]\Gamma \right),
\end{align*}
and
\begin{align*}
    \widetilde b_m(t,\Gamma):
    &=\widetilde o_m'(t)-\tr\left( \left({\widetilde{\mathcal O}}_m'(t)+i[H(t),\widetilde{\mathcal O}_m(t)]\right)\Gamma\right) \nonumber 
    \\ & = o_m''(t) -  \tr \bigg( ({\mathcal O}_m''(t)+2i[H(t),{\mathcal O}_m'(t)] \\
    & \qquad +i [H'(t),\mathcal O_m(t)]- [H(t),[H(t),\mathcal O_m]])\Gamma \bigg).
\end{align*}
This is another form of the van Leeuwen equation~\eqref{eq:vanLeeuwen_finite_dim}.

For the example considered in Section~\ref{sec:Hubbard}, the matrix $\widetilde{\mathfrak M}^{\rm VP}(t,\Gamma)$ is given by 
$$
    [\widetilde{\mathfrak M}^{\rm VP}(t,\Gamma)]_{mn} = - \tr( (i[\mathcal J_m(t), \mathcal N_n]) \Gamma),
$$
where $\mathcal N_n$ is the density operator at site $n$ (Eq.~\eqref{eq:cN}), and $\mathcal J_m(t):=i[H(t),\mathcal N_m] $ the current operator at site $m$.

As a final remark, let us mention that the matrix $\widetilde{\mathfrak M}^{\rm VP}(t,\Gamma)$ also appears in the time-dependent current density-functional theory studied in \Cref{sec.current.DFT}, since, in this setting, the VP matrix is of the form 
$$
\mathfrak M^{\rm VP}(t,\Gamma) = \left( \begin{array}{cc} 0 & -\widetilde{\mathfrak M}^{\rm VP}(t,\Gamma)^T \\ \widetilde{\mathfrak M}^{\rm VP}(t,\Gamma) & * \end{array} \right),
$$
so that the above matrix is invertible if and only if $\widetilde{\mathfrak M}^{\rm VP}(t,\Gamma)$ is invertible.

\input{PRA_TDDFT_Finite_dimension_final.bbl}

\end{document}

%% file: PRA_TDDFT_Finite_dimension_final.bbl
%